\documentclass[journal]{IEEEtran}

\usepackage[T1]{fontenc}

\IEEEoverridecommandlockouts

\usepackage[noadjust]{cite}
\usepackage{amsmath,amssymb,amsfonts}
\usepackage{dsfont}
\usepackage{graphicx}
\usepackage{textcomp}
\usepackage{setspace}
\usepackage[dvipsnames]{xcolor}
\usepackage{multirow}
\usepackage{amsthm}
\usepackage[noend]{algpseudocode}
\usepackage{subcaption}
\usepackage{mathcomSTEv4}

\usepackage[textfont=normalfont,singlelinecheck=off,justification=raggedright,font=footnotesize]{caption}

\usepackage[margin=30pt,textfont=normalfont,singlelinecheck=off,justification=raggedright,font=footnotesize]{subcaption}
\usepackage{dsfont}

\usepackage{environ}         
\usepackage{etoolbox}        
\usepackage{graphicx}        

\usepackage{anyfontsize}         

\usepackage{xcolor}
\usepackage{multicol}
\usepackage[utf8]{inputenc}
\usepackage{bbm}
\usepackage{authblk}
\allowdisplaybreaks
\usepackage{tikz,pgfplots}
\usetikzlibrary{arrows}

\newcommand{\rtc}{\rsf \times  \csf }
\newcommand{\ctr}{\csf \times \rsf }

\usepackage{hyperref}
\hypersetup{
    colorlinks=true,
    linkcolor=blue,
    filecolor=magenta,      
    urlcolor=cyan,
    pdftitle={Overleaf Example},
    pdfpagemode=FullScreen,
    }

\newtheorem{theorem}{Theorem}
\newtheorem{remark}{Remark}

\newcommand{\comment}[1]{}

\usepackage{caption,booktabs}

\makeatletter
\def\endthebibliography{%
	\def\@noitemerr{\@latex@warning{Empty `thebibliography' environment}}%
	\endlist
}

\makeatother
\begin{document}

\title{
%
%
Straggler Mitigation through \\
Unequal Error Protection for \\
Distributed Approximate Matrix Multiplication
}

\author{Busra~Tegin,~\IEEEmembership{Student Member,~IEEE,}
	Eduin~E.~Hernandez,
	Stefano~Rini,
	and~Tolga~M.~Duman,~\IEEEmembership{Fellow,~IEEE}
	\thanks{Part of the material in this paper is presented in
the 2021 IEEE International Conference on Communications (ICC), Montreal,
Canada, June 2021 \cite{tegin2021straggler}. }
	\thanks{B. Tegin and T. M. Duman are with the Department
		of Electrical and Electronics Engineering, Bilkent University, 06800 Ankara,
		Turkey (e-mail: \{btegin, duman\}@ee.bilkent.edu.tr)
	}
	\thanks{E. E. Hernandez and S. Rini are with the Department
		of Electrical and Computer Engineering, National Yang Ming Chiao Tung University, 30010 Hsinchu,
		Taiwan (e-mail: \{eduin.ee08, stefano.rini\}@nycu.edu.tw)
}}

\maketitle

\begin{abstract}
Large-scale machine learning and data mining methods routinely distribute computations across multiple agents to parallelize processing. The time required for the computations at the agents is affected by the availability of local resources and/or poor channel conditions giving rise to the “straggler problem.” As a remedy to this problem, 
we employ Unequal Error Protection (UEP) codes to obtain an approximation of the matrix product in the distributed computation setting to provide higher protection for the blocks with higher effect on the final result. 
We characterize the performance of the proposed approach from a theoretical perspective by bounding the expected reconstruction error for matrices with uncorrelated entries. We also apply the proposed coding strategy to the computation of the back-propagation step in the training of a Deep Neural Network (DNN) for an image classification task in the evaluation of the gradients. Our numerical experiments show that it is indeed possible to obtain significant improvements in the overall time required to achieve the DNN training convergence by producing approximation of matrix products using UEP codes in the presence of stragglers. 

\end{abstract}


\section{Introduction}


%

\IEEEPARstart{D}{istributed} learning 
%
%
is a fundamental approach  to the training  of machine learning models as it allows for the parallel computation of model updates. 
Parallelizing computation enhances robustness, reliability, and allows for a drastic reduction in computational and memory resources requirements at the learner.
%
Distributed computation is often supported by a dedicated infrastructure comprised of   computing clusters with heterogeneous capabilities.
The widespread reliance on distributed computation clusters presents several opportunities over traditional computing paradigms, but also offers a new set of challenges.
Among the most well-recognized issues is that of the stochasticity in the time required for the computation. 
This gives rise to the phenomenon of ``stragglers'', that is, agents with large response times which delay computation.  
Another important reason for having stragglers is due to the wireless communication characteristics where the workers observe different channel conditions, resulting in delays for the ones with the poor quality links. As a result, it may not be possible to transmit all the local computations at the same transmission rate. Therefore, the central server will receive some of the local computations later; hence such workers at wireless edge computing scenarios are stragglers.
As a remedy to stragglers, channel coding can be applied to reduce the delays in  distributed computation \cite{reisizadeh2019coded}.

In this paper, we propose a novel scheme for distributed computation with stragglers which makes use of the  variations in the magnitude of the matrix entries which naturally occur in many applications, such as gradient computation for back-propagation in Deep Neural Network (DNN) training.
%
We first identify the matrix sub-products which are expected to have the largest norms and use Unequal Error Protection (UEP) coding to provide resiliency against stragglers. 
The proposed solution offers an improved resilience by providing an improved  approximate reconstruction of the matrix product by a given deadline.
%
%
%
%

\subsection{Literature Review}
As matrix multiplication is a fundamental algebraic operation, distributed approximate matrix multiplication has been investigated in many contexts. 
In the big-data paradigm, computation and storage are distributed, hence computer processing architectures can be devised for efficiently performing this operation \cite{choi1994pumma,van1997summa}.
In a cloud-computing setting, distributed matrix computation is investigated in \cite{gupta2018oversketch,kim2019mpec}.
DNN training through back-propagation  involves multiplication of large matrices, 
for which distributed matrix computation is studied in  \cite{plancher2019application,son2018distributed}.
More recently, the problem of distributed matrix multiplication in the presence of stragglers has been considered. 
%
%
%
Coding for matrix multiplication can be applied to mitigate the effect of stragglers \cite{lee2017speeding}.
Since its inception in \cite{lee2017speeding}, this line of research received significant attention in the literature.
In \cite{wang2015using}, the authors use  the theory of extreme order statistics to analyze how task replication reduces latency. 
In \cite{dutta2016short}, the authors  introduce redundant computations in a coding theory inspired fashion for computing linear transforms of long vectors. 
Product codes for distributed matrix multiplication are studied in \cite{baharav2018straggler}.
A new class of codes, called polynomial codes, is proposed in \cite{yu2017polynomial}, and their optimality is argued for the straggler problem.

While the above literature focuses on minimizing the time for completing a computation task, one can also consider approximate computation.
Along these lines, in \cite{gupta2018oversketch}, the authors propose OverSketch, an algorithm that uses matrix sketching to approximate matrix multiplication.
%
%
Further research considers the intersection of distributed matrix computation and other relevant aspects of computation.
%
%
For instance, the authors of \cite{buyukates2020timely} consider the distributed matrix multiplication problem when the usefulness of the computation outcome is evaluated through an age-of-information paradigm \cite{kosta2017age}. 
%
%
%
%

%

\subsection{Contribution}

In this paper, we investigate the trade-off between accuracy and delay in distributed approximate matrix multiplication with stragglers. %
%
%
Since for typical machine learning problems, only approximate matrix multiplication results are sufficient, we consider a distributed matrix multiplication scheme in which the sub-blocks of the matrices being multiplied are encoded using UEP codes and distributed across different workers.
Due to, for instance, wireless channel effects, the workers respond at random completion times, with the results of the products of the coded sub-blocks.
The parameter server (PS) chooses the protection level of each matrix sub-block according to its norm so that the sub-products with the largest contribution suffer the least from the effects of stragglers. 
%
Our main goal is to produce an approximation of the product of two matrices as quickly as possible; with a more and more accurate approximation with more and more workers responding, i.e., producing a progressively improving matrix approximation in time, exploiting the UEP code constraints.

%
%
%
%
%

Our main contribution is the proposal of employing UEP codes to improve the quality of the approximation of matrix multiplications by exploiting the variations in the matrix entries' magnitudes. 
In particular, we leverage the construction of UEP codes described in \cite{vukobratovic2012unequal} through Random Linear Codes (RLC) to offer more protection  to the  sub-products with larger norms (as induced by the choice of loss) and reduce the effect of the randomness in the service time.
Specifically, we consider two schemes: Non-Overlapping Windows (NOW) and Expanding Window (EW) RLC codes for UEP, and analyze the performance of the proposed approximate matrix multiplication schemes. 
Different from the existing literature, we consider two different partitioning schemes for the matrices to be multiplied: (i) row-times-column block products, and (ii) column-times-row block products which are encoded and distributed among a set of workers which can perform sub-matrix multiplications.  
To illustrate the importance of our proposed strategy for distributed machine learning algorithms, we construct a DNN training with CIFAR-10 and MNIST datasets, which are extensively used datasets when evaluating the performance of machine learning applications, in a scenario where multiplications in the  back-propagation step are distributed among workers using the NOW-UEP and EW-UEP codes.
%

%
%
%
%

%
To showcase our results, the performance of our approach for this scenario is presented in Fig. \ref{fig:cifar10_big} where the  training performance  attainable through our algorithm for the CIFAR10 image classification database between $30$ and $120$ epochs are depicted. 
We let the response time of the servers be exponentially distributed with a mean  inversely proportional to the number of sub-block multiplications, thus accounting for the larger number of tasks when employing coding.
Three reference curves  in the plot are the red curve, corresponding to the case with no stragglers (the response time being deterministic), the blue curve, corresponding the performance with uncoded transmission, and the purple curve for which computations are simply replicated. 
The performance attainable through UEP codes for the approximate computation of the weight updates are depicted as green and yellow lines. The results clearly show that UEP codes provide a higher model accuracy in the presence of stragglers. Further analysis and interpretation are provided in Sec. \ref{subsec: DNN performance}.
%
%

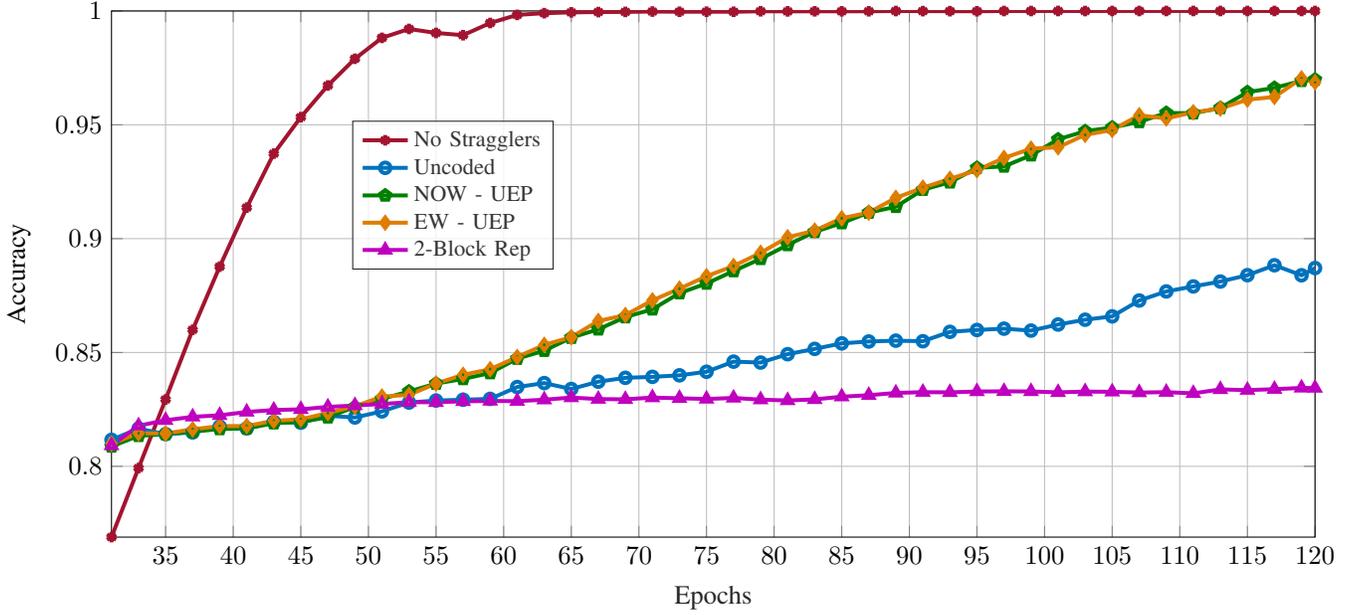
\begin{figure*}[!htbp] 
 	\centering
	\begin{tikzpicture}
\definecolor{mycolor1}{rgb}{0.63529,0.07843,0.18431}%
\definecolor{mycolor2}{rgb}{0.00000,0.44706,0.74118}%
\definecolor{mycolor3}{rgb}{0.00000,0.49804,0.00000}%
\definecolor{mycolor4}{rgb}{0.87059,0.49020,0.00000}%
\definecolor{mycolor5}{rgb}{0.00000,0.44700,0.74100}%
\definecolor{mycolor6}{rgb}{0.74902,0.00000,0.74902}%

\begin{axis}[%
width= 16cm,
height=7cm,
scale only axis,
xmin=31,
xmax=120,
xlabel style={font=\color{white!15!black}},
xlabel={Epochs},
ymin=0.768906,
ymax=1,
ylabel={Accuracy},
axis background/.style={fill=white},
xmajorgrids,
ymajorgrids,
legend style={legend cell align=left, align=left, draw=white!15!black, nodes={scale=0.85, transform shape}, at={(0.2,0.65)}, anchor=west, fill opacity=0.8}
]

\addplot [color=mycolor1, line width=1.5pt, mark=asterisk, mark options={solid, mycolor1}]
  table[row sep=crcr]{%
31	0.768906	\\
33	0.799241	\\
35	0.82935	\\
37	0.859945	\\
39	0.887684	\\
41	0.913637	\\
43	0.937325	\\
45	0.95337	\\
47	0.967275	\\
49	0.979013	\\
51	0.988186	\\
53	0.992107	\\
55	0.990352	\\
57	0.989362	\\
59	0.994718	\\
61	0.998279	\\
63	0.999025	\\
65	0.99938	\\
67	0.99952	\\
69	0.999595	\\
71	0.999745	\\
73	0.999655	\\
75	0.999675	\\
77	0.99965	\\
79	0.99987	\\
81	0.99982	\\
83	0.999855	\\
85	0.99985	\\
87	0.999865	\\
89	0.99998	\\
91	0.99993	\\
93	0.999915	\\
95	0.999865	\\
97	0.999985	\\
99	0.999985	\\
101	1	\\
103	0.999975	\\
105	0.999995	\\
107	1	\\
109	1	\\
111	1	\\
113	1	\\
115	1	\\
117	1	\\
119	1	\\
120	1	\\
};
\addlegendentry{No Stragglers}

\addplot [color=mycolor2, line width=1.5pt, mark=o, mark options={solid, mycolor2}]
  table[row sep=crcr]{%
31	0.811625	\\
33	0.816561	\\
35	0.814126	\\
37	0.814891	\\
39	0.817452	\\
41	0.816676	\\
43	0.819497	\\
45	0.819172	\\
47	0.822213	\\
49	0.821458	\\
51	0.824129	\\
53	0.82801	\\
55	0.82889	\\
57	0.82925	\\
59	0.82962	\\
61	0.834827	\\
63	0.836553	\\
65	0.833967	\\
67	0.837183	\\
69	0.838863	\\
71	0.839294	\\
73	0.839964	\\
75	0.841534	\\
77	0.845976	\\
79	0.845571	\\
81	0.849247	\\
83	0.851578	\\
85	0.853998	\\
87	0.854819	\\
89	0.855184	\\
91	0.854944	\\
93	0.85906	\\
95	0.85988	\\
97	0.860455	\\
99	0.85961	\\
101	0.862276	\\
103	0.864387	\\
105	0.865857	\\
107	0.872829	\\
109	0.876811	\\
111	0.879011	\\
113	0.881157	\\
115	0.883953	\\
117	0.888314	\\
119	0.883918	\\
120	0.887084	\\
};
\addlegendentry{Uncoded}

\addplot [color=mycolor3, line width=1.5pt, mark=pentagon, mark options={solid, mycolor3}]
  table[row sep=crcr]{%
31	0.808709	\\
33	0.81339	\\
35	0.814141	\\
37	0.815181	\\
39	0.816391	\\
41	0.816601	\\
43	0.818947	\\
45	0.819517	\\
47	0.821538	\\
49	0.826074	\\
51	0.82964	\\
53	0.832857	\\
55	0.836348	\\
57	0.838378	\\
59	0.840959	\\
61	0.847301	\\
63	0.850782	\\
65	0.856504	\\
67	0.86018	\\
69	0.865532	\\
71	0.868968	\\
73	0.87607	\\
75	0.880297	\\
77	0.885758	\\
79	0.891125	\\
81	0.897222	\\
83	0.902979	\\
85	0.90683	\\
87	0.911447	\\
89	0.913997	\\
91	0.921485	\\
93	0.924751	\\
95	0.931103	\\
97	0.931713	\\
99	0.93663	\\
101	0.943602	\\
103	0.947198	\\
105	0.948824	\\
107	0.951269	\\
109	0.955221	\\
111	0.955071	\\
113	0.957406	\\
115	0.964339	\\
117	0.966219	\\
119	0.969315	\\
120	0.96987	\\
};
\addlegendentry{NOW - UEP}

\addplot [color=mycolor4, line width=1.5pt, mark=diamond, mark options={solid, mycolor4}]
  table[row sep=crcr]{%
31	0.809964	\\
33	0.814521	\\
35	0.814446	\\
37	0.816256	\\
39	0.817662	\\
41	0.817707	\\
43	0.820002	\\
45	0.820653	\\
47	0.823418	\\
49	0.826344	\\
51	0.830621	\\
53	0.831391	\\
55	0.836338	\\
57	0.840199	\\
59	0.84264	\\
61	0.848001	\\
63	0.853343	\\
65	0.856729	\\
67	0.863826	\\
69	0.866432	\\
71	0.872889	\\
73	0.878016	\\
75	0.883508	\\
77	0.887979	\\
79	0.893776	\\
81	0.900603	\\
83	0.903534	\\
85	0.908961	\\
87	0.911452	\\
89	0.917904	\\
91	0.922375	\\
93	0.926261	\\
95	0.930063	\\
97	0.935434	\\
99	0.939516	\\
101	0.940136	\\
103	0.945588	\\
105	0.947738	\\
107	0.95418	\\
109	0.952795	\\
111	0.955501	\\
113	0.957126	\\
115	0.961048	\\
117	0.962258	\\
119	0.970386	\\
120	0.968715	\\
};
\addlegendentry{EW - UEP}

\addplot [color=mycolor6, line width=1.5pt, mark=triangle, mark options={solid, mycolor6}]
  table[row sep=crcr]{%
31	0.809144	\\
33	0.817832	\\
35	0.820232	\\
37	0.821743	\\
39	0.822403	\\
41	0.823879	\\
43	0.824719	\\
45	0.825034	\\
47	0.826104	\\
49	0.8267	\\
51	0.82738	\\
53	0.82816	\\
55	0.82814	\\
57	0.82855	\\
59	0.828715	\\
61	0.828625	\\
63	0.82926	\\
65	0.830211	\\
67	0.82953	\\
69	0.829435	\\
71	0.830191	\\
73	0.829881	\\
75	0.829585	\\
77	0.830031	\\
79	0.829285	\\
81	0.828915	\\
83	0.8294	\\
85	0.830576	\\
87	0.831236	\\
89	0.832236	\\
91	0.832576	\\
93	0.832486	\\
95	0.832867	\\
97	0.832932	\\
99	0.832887	\\
101	0.832486	\\
103	0.832822	\\
105	0.832736	\\
107	0.832406	\\
109	0.832576	\\
111	0.832051	\\
113	0.833777	\\
115	0.833477	\\
117	0.833882	\\
119	0.834437	\\
120	0.834377	\\
};
\addlegendentry{2-Block Rep}

\end{axis}

\end{tikzpicture}%
	\caption{CIFAR-10 classification accuracy between epoch 30 and epoch 120 with $\lambda = 0.5, T_{max} = 1$.
Evaluation details are presented in Sec. \ref{subsec: DNN performance}.
	}
	\label{fig:cifar10_big}
\end{figure*}

\subsubsection*{Organization}
The paper is organized as follows. In Sec. \ref{sec:System model}, we formulate the distributed approximate matrix multiplication problem for both (i) row-times-column block products and (ii) column-times-row block products.
In Sec. \ref{sec:Relevant Results}, we go over some of the existing results in the literature for coded matrix computation and approximate matrix multiplication. 
In Sec. \ref{sec:Approximate Matrix Multiplication with UEP Codes},
we present our  proposed scheme  in which UEP codes are used to encode the matrix multiplication factors,
while a theoretical evaluation of the expected error is provided in Sec. \ref{sec:Theoretical Analysis}.
In Secs. \ref{sec:Numerical Examples} and \ref{sec:Back-propagation Matrices}, we provide
numerical examples using both synthetic data and an actual data from DNN training.
Finally, the paper is concluded in Sec. \ref{sec: conclusion}.

\subsubsection*{Notation} In the paper we adopt the following notation. Matrices are denoted with bold capital Roman letters, e.g., $\Av$, column vectors with bold lower-case Roman letters, e.g., $\mathbf{v}$. 
The Frobenius  norm of the matrix $\Av$ is shown as $\| \Av \|_F$. 
The set of integers $\{1,\ldots,N\} \subset \Nbb$ is denoted as $[N]$.
Given two matrices $\Av_1$ and $\Av_2$ with the same number of rows, we depict their column-wise concatenation as $\Av = [\Av_1 \: , \: \Av_2]$. Similarly, given $\Av_1$ and $\Av_2$ with the same number of columns, their row-wise concatenation is represented as $\Av = [\Av_1 \: ; \: \Av_2]$ which can also be equivalently expressed as $\Av = [\Av_1^{\intercal} \: , \: \Av_2^{\intercal}]^{\intercal}$.
Capital Roman letters are used for scalars. 
$\mathcal{N}(\mu, \sigma^2)$ indicates the Gaussian distribution with mean $\mu$ and variance $\sigma^2$.
Finally, the expectation is denoted as $\Ebb [\cdot]$, and $\onev(\cdot)$ is used for the indicator function.

{\bf Note well:}
In the following,  we will often not explicitly indicate the support of the independent variables indexing the various matrix sub-blocks.
We shall  use lower case Roman letters for such independent variables, i.e., $n$, and let the corresponding upper case Roman letter indicate the interval  $n \in [N]$, in other words
\ean{ \tag{1}
\sum_{n \in [N]} f_n \triangleq \sum_{n} f_n.
}


\section{System model}
\label{sec:System model}

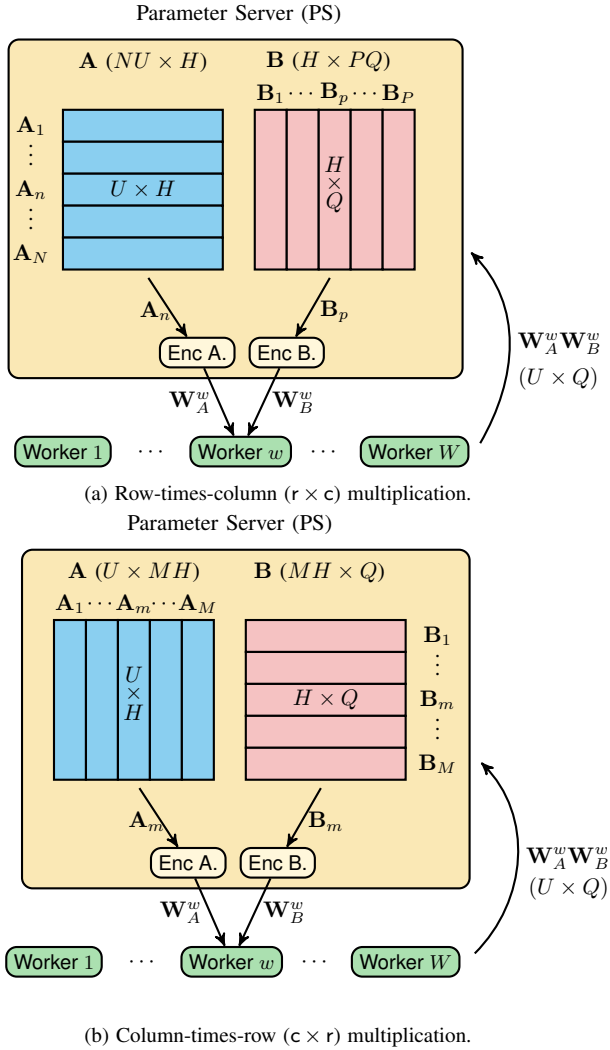
\begin{figure}

	\centering
	\begin{subfigure}[b]{0.4\textwidth}
	\begin{tikzpicture}[->,>=stealth',auto,node distance=3cm,
  thick,main node/.style={circle,draw,font=\sffamily\Large\bfseries},scale=0.85, every node/.style={scale=0.85}]
\definecolor{apricot}{rgb}{0.98, 0.81, 0.69}
\definecolor{antiquebrass}{rgb}{0.8, 0.58, 0.46}
\definecolor{arylideyellow}{rgb}{0.91, 0.84, 0.42}
\definecolor{bananamania}{rgb}{0.98, 0.91, 0.71}
\definecolor{babyblue}{rgb}{0.54, 0.81, 0.94}
\definecolor{babypink}{rgb}{0.96, 0.76, 0.76}
\definecolor{caribbeangreen}{rgb}{0.0, 0.8, 0.6}
\definecolor{celadon}{rgb}{0.67, 0.88, 0.69}
\definecolor{cornsilk}{rgb}{1.0, 0.97, 0.86}	
	
\tikzstyle{arrow} = [thick,->,>=stealth]

\node at (2.75, 2){Parameter Server (PS)};
\draw[rounded corners, fill=bananamania] (-0.85,1.6) rectangle (6.25,-3.7);
\node at (1.25, 1.25){$\mathbf{A}$ ($NU \times H$)};
\node at (4.15, 1.25){$\mathbf{B}$ ($H \times PQ$)};
\draw[fill=babyblue] (0,0) rectangle (2.5,0.5);
\node at (-0.5, 0.25){$\mathbf{A}_1$};
\node at (1.25,0.25){};
\draw[fill=babyblue] (0,-0.5) rectangle (2.5,0);
\node at (-0.5, -0.1){$\vdots$};
\node at (1.25,-0.25){};
\draw[fill=babyblue] (0,-1) rectangle (2.5,-0.5);
\node at (-0.5, -0.75){$\mathbf{A}_n$};
\node at (1.25,-0.75){$U \times H$};
\draw[fill=babyblue] (0,-1.5) rectangle (2.5,-1);
\node at (-0.5, -1.1){$\vdots$};
\node at (1.25,-1.25){};
\draw[fill=babyblue] (0,-2) rectangle (2.5,-1.5);
\node at (-0.5, -1.75){$\mathbf{A}_N$};
\node at (1.25,-1.75){};
\node (Ai) at (1.25,-2){};

\draw [fill=babypink](3,0.5) rectangle (3.5,-2);
\node at (3.25, 0.75){$\mathbf{B}_1$};

\draw [fill=babypink](3.5,0.5) rectangle (4,-2);
\node at (3.75, 0.75){$\cdots$};

\draw [fill=babypink](4,0.5) rectangle (4.5,-2);
\node at (4.25, 0.75){$\mathbf{B}_p$};
\node at (4.25,-0.35){$H$};
\node at (4.25,-0.65){$\times$};
\node at (4.25,-0.95){$Q$};

\draw [fill=babypink](4.5,0.5) rectangle (5,-2);
\node at (4.75, 0.75){$\cdots$};

\draw [fill=babypink](5,0.5) rectangle (5.5,-2);
\node at (5.25, 0.75){$\mathbf{B}_P$};
\node (B) at (6.25,-1.6){};
\node (Bj) at (4.25,-2){};
	
\node[draw,rectangle, rounded corners, fill=celadon, font=\small \sffamily] (W1) at (0,-4.85) {Worker $1$};
\node at (1.4, -4.85){$\cdots$};
\node[draw,rectangle, rounded corners, fill=celadon, font=\small \sffamily] (Ww) at (2.78,-4.85) {Worker $w$};
\node at (4.1, -4.85){$\cdots$};
\node[draw,rectangle, rounded corners, fill=celadon,font=\small \sffamily] (W1) at (5.5,-4.85) {Worker $W$};
\node (W) at (6.45,-4.85){};
\node (Wij) at (2.75,-3.25){};
\node[draw,rectangle, rounded corners, fill=cornsilk, font=\small \sffamily] (WA) at (2.1,-3.3) {Enc A.};
\node[draw,rectangle, rounded corners, fill=cornsilk, font=\small \sffamily] (WB) at (3.5,-3.3) {Enc B.};

\node (e_out) at (2.75,-3.5){};

\draw [->] (W) to [out = 60, in=-45] node[anchor=west] {$\mathbf{W}_A^w \mathbf{W}_B^w$} (B);
\node at (7.75,-3.7){$(U \times Q)$};

\draw [->] (Ai) to  (WA);
\node at (1.45,-2.65){$\mathbf{A}_n$};
\draw [->] (Bj) to  (WB);
\node at (4.25,-2.65){$\mathbf{B}_p$};
\draw [->] (WA) to node[anchor=east] {$\mathbf{W}_A^w$} (Ww);
\draw [->] (WB) to node[anchor=west] {$\mathbf{W}_B^w$} (Ww);

 \end{tikzpicture}
        \captionsetup{justification=centering}
        \caption{Row-times-column ($\rtc$) multiplication.}
		\label{subfig:systemmodel1}
	\end{subfigure}

    \centering
    \begin{subfigure}[b]{0.4\textwidth}
		\begin{tikzpicture}[->,>=stealth',auto,node distance=3cm,
  thick,main node/.style={circle,draw,font=\sffamily\Large\bfseries},scale=0.85, every node/.style={scale=0.85}]
\definecolor{apricot}{rgb}{0.98, 0.81, 0.69}
\definecolor{antiquebrass}{rgb}{0.8, 0.58, 0.46}
\definecolor{arylideyellow}{rgb}{0.91, 0.84, 0.42}
\definecolor{bananamania}{rgb}{0.98, 0.91, 0.71}
\definecolor{babypink}{rgb}{0.54, 0.81, 0.94}
\definecolor{babyblue}{rgb}{0.96, 0.76, 0.76}
\definecolor{caribbeangreen}{rgb}{0.0, 0.8, 0.6}
\definecolor{celadon}{rgb}{0.67, 0.88, 0.69}
\definecolor{cornsilk}{rgb}{1.0, 0.97, 0.86}	
	
\tikzstyle{arrow} = [thick,->,>=stealth]

\node at (2.75, 2){Parameter Server (PS)};
\draw[rounded corners, fill=bananamania] (-0.5,1.6) rectangle (6.5,-3.7);
\node at (1.25, 1.25){$\mathbf{A}$ ($U \times MH$)};
\node at (4.15, 1.25){$\mathbf{B}$ ($MH \times Q$)};
\draw[fill=babyblue] (3,0) rectangle (5.5,0.5);
\node at (6, 0.25){$\mathbf{B}_1$};
\node at (1.25,0.25){};
\draw[fill=babyblue] (3,-0.5) rectangle (5.5,0);
\node at (6, -0.1){$\vdots$};
\node at (1.25,-0.25){};
\draw[fill=babyblue] (3,-1) rectangle (5.5,-0.5);
\node at (6, -0.75){$\mathbf{B}_m$};
\node at (4.25,-0.75){$H \times Q$};
\draw[fill=babyblue] (3,-1.5) rectangle (5.5,-1);
\node at (6, -1.1){$\vdots$};
\node at (1.25,-1.25){};
\draw[fill=babyblue] (3,-2) rectangle (5.5,-1.5);
\node at (6, -1.75){$\mathbf{B}_M$};
\node at (1.25,-1.75){};
\node (Ai) at (1.25,-2){};

\draw [fill=babypink](0,0.5) rectangle (0.5,-2);
\node at (0.25, 0.75){$\mathbf{A}_1$};

\draw [fill=babypink](0.5,0.5) rectangle (1,-2);
\node at (0.75, 0.75){$\cdots$};

\draw [fill=babypink](1,0.5) rectangle (1.5,-2);
\node at (1.25, 0.75){$\mathbf{A}_m$};
\node at (1.25,-0.35){$U$};
\node at (1.25,-0.65){$\times$};
\node at (1.25,-0.95){$H$};

\draw [fill=babypink](1.5,0.5) rectangle (2,-2);
\node at (1.75, 0.75){$\cdots$};

\draw [fill=babypink](2,0.5) rectangle (2.5,-2);
\node at (2.25, 0.75){$\mathbf{A}_M$};
\node (B) at (6.55,-1.6){};
\node (Bj) at (4.25,-2){};
	
\node[draw,rectangle, rounded corners, fill=celadon, font=\small \sffamily] (W1) at (0,-4.85) {Worker $1$};
\node at (1.4, -4.85){$\cdots$};
\node[draw,rectangle, rounded corners, fill=celadon, font=\small \sffamily] (Ww) at (2.78,-4.85) {Worker $w$};
\node at (4.1, -4.85){$\cdots$};
\node[draw,rectangle, rounded corners, fill=celadon,font=\small \sffamily] (W1) at (5.5,-4.85) {Worker $W$};
\node (W) at (6.45,-4.85){};
\node (Wij) at (2.75,-3.25){};
\node[draw,rectangle, rounded corners, fill=cornsilk, font=\small \sffamily] (WA) at (2.1,-3.3) {Enc A.};
\node[draw,rectangle, rounded corners, fill=cornsilk, font=\small \sffamily] (WB) at (3.5,-3.3) {Enc B.};

\node (e_out) at (2.75,-3.5){};

\draw [->] (W) to [out = 45, in=-45] node[anchor=west] {$\mathbf{W}_A^w \mathbf{W}_B^w$} (B);
\node at (8.05,-3.7){$(U \times Q)$};

\draw [->] (Ai) to  (WA);
\node at (1.45,-2.65){$\mathbf{A}_m$};
\draw [->] (Bj) to  (WB);
\node at (4.25,-2.65){$\mathbf{B}_m$};
\draw [->] (WA) to node[anchor=east] {$\mathbf{W}_A^w$} (Ww);
\draw [->] (WB) to node[anchor=west] {$\mathbf{W}_B^w$} (Ww);

 \end{tikzpicture}
        \captionsetup{justification=centering}
        \caption{Column-times-row ($\ctr$)  multiplication.}
		\label{subfig:systemmodel2}
	\end{subfigure}
	
	\caption{System models with the $\rtc$ and $\ctr$ multiplication schemes.} \label{fig:systemmodel}
\end{figure}

We consider the scenarios in Fig. \ref{fig:systemmodel}
where a PS wishes to compute the matrix product  $\Cv=\Av \Bv$ by distributing various factors of the matrix  multiplications  among $W$ workers. 
Each worker receives two separate linear combinations of sub-matrices of $\Av$ and $\Bv$, computes their product, and returns it to the PS. 
%
%
The time required for the response of a computation is a random variable due to variations in the channel quality and/or server speed for different workers \cite{wang2019efficient}. Note that we distribute the same amount of computational load to all the workers. Thus, the response time is independent of the computational capacity of the workers for our system.
Due to transmission rate constraints of the wireless channel, we consider the delay due to stragglers’ channel conditions as the communication cost.
By a given deadline, the PS produces an approximation $\Chv$ of the matrix $\Cv$ by using sub-products from the workers received by the prescribed deadline.

\subsection{Distributed Matrix Computation Model}
\label{sec:Distributed Matrix Computation Model}

\begin{figure}[!htbp]
 	\centering
	
	\definecolor{apricot}{rgb}{0.98, 0.81, 0.69}
\definecolor{antiquebrass}{rgb}{0.8, 0.58, 0.46}
\definecolor{arylideyellow}{rgb}{0.91, 0.84, 0.42}
\definecolor{bananamania}{rgb}{0.98, 0.91, 0.71}
\definecolor{babypink}{rgb}{0.96, 0.76, 0.76}
\definecolor{babyblue}{rgb}{0.54, 0.81, 0.94}
\definecolor{caribbeangreen}{rgb}{0.0, 0.8, 0.6}
\definecolor{celadon}{rgb}{0.67, 0.88, 0.69}
\definecolor{cornsilk}{rgb}{1.0, 0.97, 0.86}
\definecolor{random}{rgb}{1.00,0.83,0.71}

\begin{tikzpicture}[scale=0.7, every node/.style={scale=1}]

\draw[fill=babyblue] (0,5) rectangle (3,3);
\draw[dashed] (0,4.5) -- (3,4.5);
\draw[dashed] (0,4) -- (3,4);
\draw[dashed] (0,3.5) -- (3, 3.5);

\node[scale=2] at (3.5,4.5){$\times$};

\draw[fill=babypink] (4,5) rectangle (6,2);
\draw[dashed] (4.5,2) -- (4.5,5);
\draw[dashed] (5.0,2) -- (5.0,5);
\draw[dashed] (5.5,2) -- (5.5,5);

\node[scale=2] at (7,4.5) {$=$};

\draw[fill=apricot] (8,5) rectangle (10,3);
\draw[dashed] (8,4.5) -- (10,4.5);
\draw[dashed] (8,4.0) -- (10,4.0);
\draw[dashed] (8,3.5) -- (10,3.5);

\draw[dashed] (8.5,3) -- (8.5,5);
\draw[dashed] (9,3) -- (9,5);
\draw[dashed] (9.5,3) -- (9.5,5);

\node at (1.5,2.5){$N \times M$};
\node at (5,1.5){$M \times P$};
\node at (9,2.5){$N \times P$};

\end{tikzpicture}
	\caption{Row-times-column multiplication ($\rtc$) in \eqref{eq:sub-blocks rXc}.}
	\label{fig:rows-times-cols}
	
	
	\centering
	\definecolor{apricot}{rgb}{0.98, 0.81, 0.69}
\definecolor{antiquebrass}{rgb}{0.8, 0.58, 0.46}
\definecolor{arylideyellow}{rgb}{0.91, 0.84, 0.42}
\definecolor{bananamania}{rgb}{0.98, 0.91, 0.71}
\definecolor{babypink}{rgb}{0.96, 0.76, 0.76}
\definecolor{babyblue}{rgb}{0.54, 0.81, 0.94}
\definecolor{caribbeangreen}{rgb}{0.0, 0.8, 0.6}
\definecolor{celadon}{rgb}{0.67, 0.88, 0.69}
\definecolor{cornsilk}{rgb}{1.0, 0.97, 0.86}

\begin{tikzpicture}[scale=0.47, every node/.style={scale=0.7}]
\draw[fill=babyblue] (0,2) rectangle (2,5);
\draw[dashed] (0.5,2) -- (0.5,5);
\draw[dashed] (1.0,2) -- (1.0,5);
\draw[dashed] (1.5,2) -- (1.5,5);

\node[scale=2] at (2.5,4.5){$\times$};

\draw[fill=babypink] (3,5) rectangle (6,3);
\draw[dashed] (3, 4.5) -- (6, 4.5);
\draw[dashed] (3, 4.0) -- (6, 4.0);
\draw[dashed] (3, 3.5) -- (6, 3.5);
\node[scale=2] at (7,4.5) {$=$};

\draw[fill=apricot, dashed] (9,2) rectangle (12,5);
\draw[fill=apricot, dashed] (8.5,1.5) rectangle (11.5,4.5);
\draw[fill=apricot, dashed] (8,1) rectangle (11,4);
\draw[fill=apricot, dashed] (7.5,0.5) rectangle (10.5,3.5);
\node[scale=2] at (6.5,2) {$\sum$};
\node at (6.5, 1.2) {$M$};

\node[scale=2] at (13,4.5) {$=$};
\draw[fill=apricot] (14, 2) rectangle (17.5, 5);

\node at (1.0,1.5){$N \times M$};
\node at (4.5,2.5){$M \times P$};
\node at (9,0){$N \times P \times M$};
\node at (16,1.5){$N \times P$};

\end{tikzpicture}
	\caption{Column-times-row multiplication ($\ctr$) in \eqref{eq:sub-blocks cXr}.}
	\label{fig:cols-times-rows}
\end{figure}

Let us next describe the problem setting in more detail. 
Consider the matrices $\Av$ and $\Bv$ with elements from a finite field $\mathbb{F}$.
The matrix $\Av$ is comprised of $N \times M$ sub-blocks of dimensions $U \times H$, thus resulting in the overall dimensions  $N U \times M H$. 
Similarly, $\Bv$ is comprised of $M \times P$ sub-blocks of  dimensions $H \times Q$ resulting in $M H \times P Q$. 
Accordingly, the matrix $\Cv$ has $N \times P$ sub-blocks of dimension $U \times Q$.  
Thus, $\Av \in \mathbb{F}^{N U \times M H}$, $\Bv \in \mathbb{F}^{M H \times P Q}$, and $\Cv \in \mathbb{F}^{NU \times P Q}$.

The aim of the PS is to produce  $\Chv$ as an approximate expression for the matrix  multiplication $\Cv=\Av \Bv$ with respect to the loss\footnote{In the following, we only consider the case of a Frobenius norm: the case of a more general loss is not discussed here for brevity.}
\ea{
	\Lcal (\Cv,\Chv) = \|\Cv-\Chv\|_F^2. \tag{2}
	\label{eq:loss}
}

To accomplish this, the PS  divides the matrix products into sub-products and distributes them across a set of workers. 
Specifically, following \cite{lee2017high}, we partition $\Av$ and $\Bv$ in two possible ways:
\begin{itemize}
	\item{\bf Row-times-Column ($\rtc$):}  that is, $M = 1$ and $H = MH$ such that  $H$
	has the the same number of columns as $A$ and number of rows as $B$.

	In this case, if  uncoded multiplication are distributed to the servers,  what is returned are $N P$ sub-blocks of sizes $U \times Q$ belonging to the $\Cv$ matrix, so that
	\begin{align}  \label{eq:sub-blocks rXc}
	\Av & = 
	[\Av_1 \:;\: \cdots  \Av_n \:;\: \cdots \Av_N], \nonumber \\
	\Bv & = [\Bv_1  \:,\: \cdots  \Bv_p   \:,\:  \cdots \Bv_P]. \tag{3}
	\end{align}
	We indicate this case with the notation $\rtc$.
	This partitioning is presented in Figs. \ref{subfig:systemmodel1} and \ref{fig:rows-times-cols}.
	
	\item{\bf Column-times-Row ($\ctr$):}  that is, $N = P = 1$, $U = NU$, $Q = PQ$  such that  $U$ has the same number of rows as $A$ 
	and $Q$ has the same number of columns as $B$.
	
	In this case, if  uncoded multiplication are distributed to the servers,  what is returned are $M$ matrices of sizes $U \times Q$, so that
	\begin{align} \label{eq:sub-blocks cXr}
	\Av & = [\Av_1 \:,\: \cdots  \Av_m \:,\: \cdots \Av_M],  \nonumber \\
	\Bv & = [\Bv_1  \:;\: \cdots  \Bv_m  \:;\: \cdots  \Bv_M]. \tag{4}
	\end{align}
	We indicate this case with the notation $\ctr$.
	This partitioning is presented in Figs. \ref{subfig:systemmodel2} and \ref{fig:cols-times-rows}.
\end{itemize}

\begin{table*}[t]
	\centering
	
	\caption{ 
    A summary of the quantities in Sec. \ref{sec:System model} (column II and III) and Sec. \ref{sec:Approximate Matrix Multiplication with UEP Codes} (column IV and V). 
    }
	\label{tab:my_label}
	{
		\begin{tabular}{|c|c|c|c|c|}
			\hline
			Multiplication Case & Matrix      & Size           & Constant & Value \\ \hline
			\multirow{3}{*}{General} & $\mathbf A$ & $NU \times MH$  & \# of workers     & $W$     \\
			& $\mathbf B$ & $MH \times PQ$  &  \# of importance levels ($\mathbf{A}/ \mathbf{B}$)     & $S$     \\
			& $\mathbf C$ & $NU \times PQ$  &  \# of importance levels ($\mathbf C$)     & $L$     \\ \hline
			\multirow{3}{*}{$\rtc$}& $\Av_n$ & $U \times H$ & \# of  row blocks ($\Av$)  & $N$  \\
			& $\Bv_p$ & $H \times Q$  & \# of  column  blocks ($\Bv$) & $P$ \\
			& $\Cv_{np}$ & $U \times Q$   &  Response time scaling & $\Omega$ \\ \hline
			\multirow{3}{*}{$\ctr$} & $\Av_m$ & $U \times H$   &  \# of column blocks ($\Av$)  & $M$ \\ 
			& $\Bv_m$ & $H \times Q$   &  \# of row blocks ($\Bv$) & $M$  \\
			& $\Cv_{m}$ & $U \times Q$  & Deadline & $T_{\max}$ \\ \hline
		\end{tabular}
	}
\end{table*}

%
%


In general, not all the sub-blocks have the same impact on the final matrix multiplication result, as some sub-blocks may have larger Frobenius norms than others.
%
%
This motivates the use of codes to efficiently trade-off the matrix approximation with the computation delay.
In other words, 
codes can be employed to better protect the more impactful sub-products when distributing the computation to the workers, so that a more precise approximation is produced in a shorter time. 
%
%
For this reason, we consider the  coding problem in which the PS sends the matrices $\Wv_{A}^w$ and $\Wv_{B}^w$  obtained as 
\ea{  \label{eq:encode}
	\Wv_{A}^w & = \lcb \p{f_{ {\rm enc} - A} \lb \Av_1, \cdots , \Av_N \rb &  
	\rtc \\
	f_{ {\rm enc} - A} \lb \Av_1, \cdots , \Av_M \rb & 
	\ctr,\\
	} \tag{5} \rnone\\ 
	\vspace{0.3cm}
	\Wv_{B}^w & =  \lcb \p{f_{{\rm enc} - B} \lb \Bv_1, \cdots , \Bv_P \rb & 
	\rtc \\
	f_{{\rm enc} - B} \lb \Bv_1, \cdots , \Bv_M \rb &  
	\ctr,
	} \tag{6} \label{eq:encode2}\rnone
}
 to each worker $w$ and sets a time deadline $T_{\max}$ by which it expects the matrix products $\{\Wv_{A}^w \Wv_{B}^w\}_{w \in [W]}$ to be returned where $f_{{\rm enc} - A}$ and $f_{{\rm enc} - B} $ are the encoding functions for the sub-matrices of $\Av$ and $\Bv$, respectively.
At time $T_{\max}$, the PS produces the  approximation of the matrix product $\Cv=\Av \Bv$ as
\ea{
	\Chv = f_{ {\rm dec}-C} \lb \Wcal(T_{\max})\rb, \tag{7}
	\label{eq:Ch}
}
where $\Wcal(T_{\max}) \subseteq \{\Wv_{A}^w \Wv_{B}^w\}_{w \in [W]}$ is the set of matrix products received up to time $T_{\max}$ where $f_{{\rm dec} - C}$ is the decoding function for the final product estimate $\Chv$.
Using  $\Chv$ in \eqref{eq:Ch}, the loss in \eqref{eq:loss} can be evaluated: let us denote it as $\Lcal(T_{\max})$.

Note that the set $\Wcal(T_{\max})$ is a random, which follows from the randomness of the response times.
%
%
More precisely, we assume that the response time of the workers are random variables denoted by $T_w$ which are identical and independently distributed (i.i.d.) with a cumulative distribution function (CDF) 
 \ea{ 
 F_{T_w}(t)=F(t), \quad w \in [W], \tag{8}
 \label{eq:waiting}
 }
 for some CDF $F(\cdot)$. 
%
%

\smallskip

The problem we consider in the following is to design the functions in \eqref{eq:encode}, \eqref{eq:encode2} and \eqref{eq:Ch} 
such that the loss in \eqref{eq:loss}, averaged over the randomness in $\Wcal(T_{\max})$, is minimized over some dataset of matrix multiplications $\Dcal(\{ \Av,\Bv\})$. 
Let us define this quantity as a function of the waiting time $T_{\max}$ as
\ea{ 
\Lsf(T_{\max})=\min \quad \Ebb \lsb \sum_{\Dcal}  \Lcal (\Cv,\Chv) \rsb , \tag{9}
}
where (i)  the minimization is over  $f_{ {\rm enc} - A}$,$f_{ {\rm enc} - B}$ and $f_{ {\rm dec} - C}$, (ii) the expectation is over $\Wcal(T_{\max})$, (iii) the summation is over all matrices in the database $\Dcal$, and all the matrices in the database are equally likely.

\begin{remark}{\bf Comparison across models.}
\label{rem:Comparison across models}
In the following, we will be interested in comparing the performance when the number of servers changes.
For this comparison, we will consider the scaling of the response times as in \eqref{eq:waiting} as $F(\Omega t)$, where $\Omega$ is the number of matrix sub-products divided by the number of workers.
%
%
%
\end{remark}

\begin{remark}{\bf Matrix partitioning paradigms.}
A representation of the $\rtc$ paradigm is provided in  Fig. \ref{subfig:systemmodel1}:  in Fig. \ref{fig:rows-times-cols} we represent the resulting block-matrix structure of $\Cv$. We observe that each sub-product in $\Wcal(T_{\max})$ contributes one sub-block in $\Cv$. 
The $\ctr$ paradigm is represented in Fig. \ref{subfig:systemmodel2}: in this paradigm  $\Cv$  is obtained as a sum of rank-one terms, or outer product, as shown in Fig. \ref{fig:cols-times-rows}. 
In this case, each sub-product in $\Wcal(T_{\max})$ contributes to one of such rank-one terms.

\end{remark}

\begin{remark}

Let us elaborate further on the notion used in \eqref{eq:sub-blocks rXc} and \eqref{eq:sub-blocks cXr}.
%
Stated more precisely, the sub-block matrix structure of $\Av,\Bv$ and $\Cv$ is as follows: 
%
%
The matrix $\Av$ is comprised of $N  M$  sub-matrices $\Av_{nm}$ with $n \in [N]$, $m \in [M]$, and $\Av_{nm}  \in \mathbb{F}^{U  \times H}$.
%
%
Similarly, the matrix $\Bv$ is comprised of $MP$ sub-matrices $\Bv_{mp}$ with $m \in [M]$, $p \in [P]$, and $\Bv_{mp}  \in \mathbb{F}^{H\times Q}$.
Accordingly, the matrix $\Cv$ is composed of $NP$ sub-matrices $\Cv_{np}$  for $n \in [N]$, $p \in [P]$, and with  $\Cv_{np} \in \mathbb{F}^{U  \times Q}$.
Note that here $NP$  block-matrix multiplications are needed to produce $\Cv$. 

\noindent
$\bullet$ \underline{$\rtc$ scenario:}
For the row-times-column case in \eqref {eq:sub-blocks rXc}, we have that $\Av_{n}  \in \mathbb{F}^{U  \times H}$ and $\Bv_{p}  \in \mathbb{F}^{H\times Q}$ for $ n \in [N]$, and $p\in [P]$. $\Cv_{np} \in \mathbb{F}^{U  \times Q}$ and $NP$ such block-matrix multiplications are needed to produce $\Cv$. 
%

\noindent
$\bullet$ \underline{$\ctr$ scenario:}
For the column-times-row case in \eqref{eq:sub-blocks cXr}, we have $\Av_{m}  \in \mathbb{F}^{U  \times H}$ and $\Bv_{m}  \in \mathbb{F}^{H\times Q}$ for $m\in [M]$. $\Cv_{m} \in \mathbb{F}^{U  \times Q}$ and $M$ such block matrix multiplications and summations are needed to produce $\Cv$.
%
Note that the notation $\rtc$ / $\ctr$ for the row-times-column case/column-times-row indicates that  $M=1$ / $N=P=1$, respectively.

\end{remark}

A summary of the notation introduced in this section is provided in Table \ref{tab:my_label}.

\subsection{Deep Learning Motivation}
\label{sec:Deep Learning Motivation}

Let us now briefly motivate the choice of system model in Sec. \ref{sec:Distributed Matrix Computation Model} in the context of distributed training of DNN. 
Note that we will further comment on this application of our results in Sec. \ref{subsec: DNN performance}.

Generally speaking, we observe the matrices involved in the back-propagation, both weight and gradient matrices, have a sparse nature, which is often also enforced through sparsification techniques. 
It can also be observed that the sparsity level varies across DNN layers, often deeper layers being sparser than shallower ones. 
The presence of sparsity in these matrices means that the UEP codes have the potential to drastically improve the back-propagation speed through approximate matrix multiplication. 

%
%
In DNN training, sparsity is often explicitly introduced in order to introduce robustness or reduce the communication load of the training process. One of the most straightforward sparsification approaches is to set all the DNN weights below a certain threshold to zero. These threshold values are increased as the training progresses, so as to enforce sparsity in the final weights.
Additionally, the sparsity level increases with the layer depth, as deeper layers are generally more sparse than shallow ones. 
%
In Fig. {\ref{fig:Gaussian_modeling}} we plot Gaussian fitting of the gradients, weights, and inputs at different layers for a DNN model trained over the MNIST dataset for digit classification. Since the inputs have gone through Rectified Linear Units (ReLU), these values are nonnegative. 
In these simulations, we appropriately choose a sparsity level and report it in Table \ref{tab:Gaussian_modeling_sparseness}. 
Note that at least half of the matrix entries are well approximated by zero entries. 
%
%
%
As it can be observed from Fig. \ref{fig:Gaussian_modeling_a}, the non-sparse entry can be well described as being drawn from a Gaussian distribution with a mean close to zero. 
In Figs. \ref{fig:Gaussian_modeling_b} and \ref{fig:Gaussian_modeling_c}, we plot a Gaussian fitting of the gradient and weight values, showing that the non-sparse matrix entries are roughly Gaussian distributed with a  variance increasing with the layer depth. 
Finally, as the inputs is the DNN are generated through a ReLU activation function, these matrix entries are roughly half Gaussian distributed. 
\begin{figure}
   \hspace{-0.2cm}
	\begin{subfigure}[b]{0.475\textwidth}
	    \centering
	    \begin{tikzpicture}
\definecolor{mycolor1}{rgb}{0.647058824,0.164705882,0.164705882}%
\definecolor{mycolor2}{rgb}{0.854901961, 0.647058824, 0.125490196}%

\begin{axis}[%
width=7cm,
height=3cm,
scale only axis,
xmin=-0.0015,
xmax=0.0015,
xlabel style={font=\color{white!15!black}},
ymin=0,
ymax=2.35E+03,
ylabel style={font=\color{white!15!black}},
ylabel={P(X=k)},
axis background/.style={fill=white},
xmajorgrids,
ymajorgrids,
legend style={legend cell align=left, align=left, draw=white!15!black, nodes={scale=0.85, transform shape}, at={(0.01,0.80)}, anchor=west, fill opacity=0.8}
]

\addplot [color=mycolor1, line width=1.5pt, mark=square, mark options={solid, mycolor1}]
  table[row sep=crcr]{%
-0.00758676	2.10249	\\
-0.00743785	0	\\
-0.00728894	0	\\
-0.00714003	0	\\
-0.00699111	0	\\
-0.0068422	0	\\
-0.00669329	0	\\
-0.00654438	0	\\
-0.00639546	0	\\
-0.00624655	0	\\
-0.00609764	1.05124	\\
-0.00594872	0	\\
-0.00579981	0	\\
-0.0056509	0	\\
-0.00550199	1.05124	\\
-0.00535307	0	\\
-0.00520416	1.05124	\\
-0.00505525	0	\\
-0.00490634	2.10249	\\
-4.76E-03	0	\\
-0.00460851	1.05124	\\
-0.0044596	0	\\
-0.00431068	2.10249	\\
-0.00416177	2.10249	\\
-0.00401286	3.15373	\\
-0.00386395	3.15373	\\
-0.00371503	2.10249	\\
-0.00356612	6.30746	\\
-0.00341721	4.20497	\\
-0.00326829	5.25622	\\
-0.00311938	7.3587	\\
-0.00297047	2.10249	\\
-0.00282156	9.46119	\\
-0.00267264	7.3587	\\
-0.00252373	13.6662	\\
-0.00237482	12.6149	\\
-0.00222591	14.7174	\\
-0.00207699	22.0761	\\
-0.00192808	23.1273	\\
-0.00177917	34.691	\\
-0.00163025	31.5373	\\
-0.00148134	30.486	\\
-0.00133243	38.896	\\
-0.00118352	57.8184	\\
-0.0010346	64.1258	\\
-0.00088569	95.6631	\\
-0.000736777	109.329	\\
-0.000587865	173.455	\\
-0.000438952	226.017	\\
-0.000290039	496.187	\\
-0.000141126	1939.54	\\
7.79E-06	1895.39	\\
1.57E-04	399.472	\\
0.000305612	213.402	\\
0.000454525	133.508	\\
0.000603437	109.329	\\
0.00075235	84.0994	\\
0.000901263	71.4845	\\
0.00105018	51.5109	\\
0.00119909	37.8448	\\
0.001348	26.2811	\\
0.00149691	33.6398	\\
0.00164583	16.8199	\\
0.00179474	28.3836	\\
0.00194365	23.1273	\\
0.00209257	14.7174	\\
0.00224148	15.7686	\\
0.00239039	9.46119	\\
0.0025393	13.6662	\\
0.00268822	13.6662	\\
0.00283713	9.46119	\\
0.00298604	5.25622	\\
0.00313495	9.46119	\\
0.00328387	8.40994	\\
0.00343278	9.46119	\\
0.00358169	6.30746	\\
0.00373061	4.20497	\\
0.00387952	0	\\
0.00402843	3.15373	\\
0.00417734	1.05124	\\
0.00432626	4.20497	\\
0.00447517	1.05124	\\
0.00462408	5.25622	\\
0.004773	1.05124	\\
0.00492191	1.05124	\\
0.00507082	0	\\
0.00521973	3.15373	\\
0.00536865	0	\\
0.00551756	0	\\
0.00566647	1.05124	\\
0.00581538	0	\\
0.0059643	0	\\
0.00611321	0	\\
0.00626212	1.05124	\\
0.00641104	0	\\
0.00655995	0	\\
0.00670886	0	\\
0.00685777	0	\\
0.00700669	0	\\
0.0071556	0	\\
0.00730451	2.10249	\\
};
\addlegendentry{Sampled Grads}

\addplot [color=mycolor2, line width=1.5pt, mark=o, mark options={solid, mycolor2}]
  table[row sep=crcr]{%
-0.00490634	7.57277e-311	\\
-0.00460851	2.02E-273	\\
-0.00431068	2.26E-238	\\
-0.00401286	1.06E-205	\\
-0.00371503	2.06E-175	\\
-0.00341721	1.69E-147	\\
-0.00311938	5.78E-122	\\
-0.00282156	8.27E-99	\\
-0.00252373	4.95E-78	\\
-0.00222591	1.24E-59	\\
-0.00192808	1.30E-43	\\
-0.00163025	5.70E-30	\\
-0.00133243	1.05E-18	\\
-0.0010346	8.03E-10	\\
-0.00088569	2.85E-06	\\
-0.000736777	2.58E-03	\\
-0.000587865	0.592952	\\
-0.000438952	34.6603	\\
-0.000290039	515.27	\\
-0.000141126	1948.16	\\
-7.09E-05 2268.214 \\
7.79E-06	1873.28	\\
1.57E-04	458.109	\\
0.000305612	28.492	\\
0.000454525	0.450678	\\
0.000603437	0.001813	\\
0.00075235	1.85E-06	\\
0.000901263	4.83E-10	\\
0.00119909	5.38E-19	\\
0.00149691	2.50E-30	\\
0.00179474	4.88E-44	\\
0.00209257	3.98E-60	\\
0.00239039	1.36E-78	\\
0.00268822	1.94E-99	\\
0.00298604	1.16E-122	\\
0.00328387	2.89E-148	\\
0.00358169	3.02E-176	\\
0.00387952	1.32E-206	\\
0.00417734	2.42E-239	\\
0.00447517	1.85E-274	\\
0.004773	5.92697e-312	\\
};
\addlegendentry{Gaussian Model}

\end{axis}

\end{tikzpicture}%
        \captionsetup{justification=centering}
	    \caption{Gaussian fitting of the sampled gradient from  layer 1.} \label{fig:Gaussian_modeling_a}
	    \vspace{0.2cm}
	\end{subfigure}

  \hspace{-0.2cm}
	\begin{subfigure}[b]{0.475\textwidth}
	    \centering
	    \begin{tikzpicture}
\definecolor{mycolor1}{rgb}{0.854901961, 0.647058824, 0.125490196}%
\definecolor{mycolor2}{rgb}{0.254901961,0.411764706,0.882352941}%
\definecolor{mycolor3}{rgb}{0.419607843, 0.556862745, 0.137254902}%

\begin{axis}[%
width=7cm,
height=3cm,
scale only axis,
xmin=-0.0015,
xmax=0.0015,
xlabel style={font=\color{white!15!black}},
ymin=0,
ymax=4.35E+03,
ylabel style={font=\color{white!15!black}},
ylabel={P(X=k)},
axis background/.style={fill=white},
xmajorgrids,
ymajorgrids,
]

\addplot [color=mycolor3, line width=1.5pt, mark=pentagon, mark options={solid, mycolor3}]
  table[row sep=crcr]{%
-0.00780591	2.89E-194	\\
-0.00725313	4.67E-167	\\
-0.00670034	7.11E-142	\\
-0.00614756	1.02E-118	\\
-0.00559477	1.39E-97	\\
-0.00504199	1.78E-78	\\
-0.00448921	2.16E-61	\\
-0.00393642	2.47E-46	\\
-0.00338364	2.67E-33	\\
-0.00283085	2.72E-22	\\
-0.00227807	2.62E-13	\\
-0.00172528	2.38E-06	\\
-0.0011725	2.04E-01	\\
-0.000619714	1.65E+02	\\
-6.69E-05	1.26E+03	\\
0.000485855	9.10E+01	\\
0.00103864	6.20E-02	\\
0.00159142	3.99E-07	\\
0.00214421	2.42E-14	\\
0.00269699	1.39E-23	\\
0.00324978	7.52E-35	\\
0.00380256	3.84E-48	\\
0.00435535	1.85E-63	\\
0.00490813	8.44E-81	\\
0.00546091	3.63E-100	\\
};
 \addlegendentry{Layer 3}

\addplot [color=mycolor2, line width=1.5pt, mark=o, mark options={solid, mycolor2}]
  table[row sep=crcr]{%
-0.00222243	6.19E-304	\\
-0.0020768	2.29E-264	\\
-0.00193116	1.57E-227	\\
-0.00178553	2.02E-193	\\
-0.00163989	4.81E-162	\\
-0.00149426	2.13E-133	\\
-0.00134863	1.76E-107	\\
-0.00120299	2.71E-84	\\
-0.00105736	7.74E-64	\\
-0.000911723	4.12E-46	\\
-0.000766089	4.08E-31	\\
-0.000620454	7.52E-19	\\
-0.00047482	2.58E-09	\\
-0.000402003	1.43E-05	\\
-0.000329186	0.0164896	\\
-0.000256368	3.94478	\\
-0.000183551	196.011	\\
-0.000110734	2022.93	\\
-3.79E-05	4336.39	\\
3.49E-05	1930.73	\\
0.000107718	178.549	\\
0.000180535	3.42958	\\
0.000253352	0.0136826	\\
0.000326169	1.13E-05	\\
0.000471803	6.98E-14	\\
0.000617438	7.99E-25	\\
0.000763072	1.70E-38	\\
0.000908707	6.75E-55	\\
0.00105434	4.99E-74	\\
0.00119998	6.85E-96	\\
0.00134561	1.75E-120	\\
0.00149124	8.34E-148	\\
0.00163688	7.39E-178	\\
0.00178251	1.22E-210	\\
0.00192815	3.74E-246	\\
0.00207378	2.13E-284	\\
};
 \addlegendentry{Layer 2}

\addplot [color=mycolor1, line width=1.5pt, mark=asterisk, mark options={solid, mycolor1}]
  table[row sep=crcr]{%
-0.00490634	7.57277e-311	\\
-0.00460851	2.02E-273	\\
-0.00431068	2.26E-238	\\
-0.00401286	1.06E-205	\\
-0.00371503	2.06E-175	\\
-0.00341721	1.69E-147	\\
-0.00311938	5.78E-122	\\
-0.00282156	8.27E-99	\\
-0.00252373	4.95E-78	\\
-0.00222591	1.24E-59	\\
-0.00192808	1.30E-43	\\
-0.00163025	5.70E-30	\\
-0.00133243	1.05E-18	\\
-0.0010346	8.03E-10	\\
-0.00088569	2.85E-06	\\
-0.000736777	2.58E-03	\\
-0.000587865	0.592952	\\
-0.000438952	34.6603	\\
-0.000290039	515.27	\\
-0.000141126	1948.16	\\
-7.09E-05 2268.214 \\
7.79E-06	1873.28	\\
1.57E-04	458.109	\\
0.000305612	28.492	\\
0.000454525	0.450678	\\
0.000603437	0.001813	\\
0.00075235	1.85E-06	\\
0.000901263	4.83E-10	\\
0.00119909	5.38E-19	\\
0.00149691	2.50E-30	\\
0.00179474	4.88E-44	\\
0.00209257	3.98E-60	\\
0.00239039	1.36E-78	\\
0.00268822	1.94E-99	\\
0.00298604	1.16E-122	\\
0.00328387	2.89E-148	\\
0.00358169	3.02E-176	\\
0.00387952	1.32E-206	\\
0.00417734	2.42E-239	\\
0.00447517	1.85E-274	\\
0.004773	5.92697e-312	\\
};
 \addlegendentry{Layer 1}

\end{axis}

\end{tikzpicture}%
	    \caption{Gradient fitting: \\
	    Layer 1 $ \approx \mathcal{N}(-7.09\Esf-05,7.24\Esf-01)$ \\
	    Layer 2 $ \approx \mathcal{N}(-3.90\Esf-05,6.31\Esf-01)$ \\
	    Layer 3 $ \approx \mathcal{N}(-1.02\Esf-04,2.56\Esf-04)$}
	    \label{fig:Gaussian_modeling_b}
	     \vspace{0.2cm}
	\end{subfigure}

   \hspace{-0.2cm}
	\begin{subfigure}[b]{0.475\textwidth}
	    \centering
	    \begin{tikzpicture}
\definecolor{mycolor1}{rgb}{0.854901961, 0.647058824, 0.125490196}%
\definecolor{mycolor2}{rgb}{0.254901961,0.411764706,0.882352941}%
\definecolor{mycolor3}{rgb}{0.419607843, 0.556862745, 0.137254902}%

\begin{axis}[%
width=7cm,
height=3cm,
scale only axis,
xmin=-0.5,
xmax=0.5,
xlabel style={font=\color{white!15!black}},
ymin=0,
ymax=8.12,
ylabel style={font=\color{white!15!black}},
ylabel={P(X=k)},
axis background/.style={fill=white},
xmajorgrids,
ymajorgrids,
]

\addplot [color=mycolor3, line width=1.5pt, mark=pentagon, mark options={solid, mycolor3}]
  table[row sep=crcr]{%
-0.499154	0.00385084	\\
-0.460017	0.0111453	\\
-0.420879	0.0294245	\\
-0.381742	0.0708601	\\
-0.342605	0.155658	\\
-0.303467	0.311904	\\
-0.26433	0.570096	\\
-0.225193	0.950502	\\
-0.186055	1.44556	\\
-0.146918	2.00538	\\
-0.107781	2.53767	\\
-0.0686433	2.92921	\\
-0.029506	3.08422	\\
0.00963139	2.96222	\\
0.0487687	2.59518	\\
0.0879061	2.07393	\\
0.127043	1.51182	\\
0.166181	1.00527	\\
0.205318	0.609741	\\
0.244455	0.337353	\\
0.283593	0.170256	\\
0.32273	0.0783786	\\
0.361867	0.0329132	\\
0.401005	0.0126073	\\
0.440142	0.00440505	\\
0.47928	1.40E-03	\\
};
\addlegendentry{Layer 3}

\addplot [color=mycolor2, line width=1.5pt, mark=o, mark options={solid, mycolor2}]
  table[row sep=crcr]{%
-0.334594	3.08E-03	\\
-0.310034	0.00878892	\\
-0.285474	0.0231107	\\
-0.260914	0.0560452	\\
-0.236355	0.125346	\\
-0.211795	0.258543	\\
-0.187235	0.491816	\\
-0.162675	0.862819	\\
-0.138115	1.396	\\
-0.113555	2.08304	\\
-0.0889949	2.86654	\\
-0.064435	3.63803	\\
-0.039875	4.25816	\\
-0.0153151	4.59649	\\
0.00924486	4.57592	\\
0.0338048	4.20125	\\
0.0583647	3.55734	\\
0.0829247	2.77793	\\
0.107485	2.00062	\\
0.132045	1.32879	\\
0.156604	0.813945	\\
0.181164	0.459814	\\
0.205724	0.239562	\\
0.230284	0.115106	\\
0.254844	0.0510071	\\
0.279404	0.0208453	\\
0.303964	7.86E-03	\\
0.328524	2.73E-03	\\
};
\addlegendentry{Layer 2}

\addplot [color=mycolor1, line width=1.5pt, mark=asterisk, mark options={solid, mycolor1}]
  table[row sep=crcr]{%
-0.206969	5.85E-06	\\
-0.202114	1.82E-03	\\
-0.192405	4.01E-03	\\
-0.182695	8.52E-03	\\
-0.172986	1.74E-02	\\
-0.163276	0.0341696	\\
-0.153567	0.0645068	\\
-0.143857	0.117096	\\
-0.134148	0.204387	\\
-0.124438	0.34303	\\
-0.114729	0.553585	\\
-0.105019	0.859027	\\
-0.0953098	1.28174	\\
-0.0856003	1.83893	\\
-0.0758907	2.5369	\\
-0.0661812	3.3652	\\
-0.0564717	4.2923	\\
-0.0467622	5.2643	\\
-0.0370527	6.20815	\\
-0.0273432	7.03972	\\
-0.0176336	7.67573	\\
-0.00792413	8.0474	\\
0.00178538	8.11265	\\
0.0114949	7.86395	\\
0.0212044	7.32977	\\
0.0309139	6.56918	\\
0.0406234	5.66113	\\
0.050333	4.69101	\\
0.0600425	3.73767	\\
0.069752	2.86357	\\
0.0794615	2.10952	\\
0.089171	1.49428	\\
0.0988805	1.01778	\\
0.10859	0.666567	\\
0.1183	0.419765	\\
0.128009	0.254179	\\
0.137719	0.147994	\\
0.147428	0.0828555	\\
0.157138	0.0446035	\\
0.166847	0.0230881	\\
0.176557	1.15E-02	\\
0.186266	5.50E-03	\\
0.195976	2.53E-03	\\
0.205685	1.12E-03	\\
};
\addlegendentry{Layer 1}

\end{axis}

\end{tikzpicture}%
	    \caption{Weight fitting: \\
	    Layer 1 $ \approx \mathcal{N}(-1.07\Esf-03,9.99\Esf-01)$\\
	    Layer 2 $ \approx \mathcal{N}(-4.40\Esf-03,1.00\Esf+00)$ \\
	    Layer 3 $ \approx \mathcal{N}(-2.71\Esf-02,9.98\Esf-01)$}
	    \label{fig:Gaussian_modeling_c}
	     \vspace{0.2cm}
	\end{subfigure}
	
    \hspace{-0.2cm}
	\begin{subfigure}[b]{0.475\textwidth}
	    \centering
	    \begin{tikzpicture}
\definecolor{mycolor1}{rgb}{0.854901961, 0.647058824, 0.125490196}%
\definecolor{mycolor2}{rgb}{0.254901961,0.411764706,0.882352941}%
\definecolor{mycolor3}{rgb}{0.419607843, 0.556862745, 0.137254902}%

\begin{axis}[%
width=7cm,
height=3cm,
scale only axis,
xmin=0.000127394,
xmax=5.20646,
xlabel style={font=\color{white!15!black}},
xlabel={k},
ymin=0,
ymax=7.66E-01,
ylabel style={font=\color{white!15!black}},
ylabel={P(X=k)},
axis background/.style={fill=white},
xmajorgrids,
ymajorgrids,
]

\addplot [color=mycolor3, line width=1.5pt, mark=pentagon, mark options={solid, mycolor3}]
  table[row sep=crcr]{%
0.000127394	0.687168	\\
0.208381	0.697513	\\
0.416634	0.674799	\\
0.624887	0.622198	\\
0.83314	0.546783	\\
1.04139	0.457967	\\
1.24965	0.365583	\\
1.4579	0.278144	\\
1.66615	0.201691	\\
1.87441	0.139391	\\
2.08266	0.0918154	\\
2.29091	0.0576407	\\
2.49917	0.0344886	\\
2.70742	0.0196677	\\
2.91567	0.0106897	\\
3.12393	0.00553742	\\
3.33218	0.0027339	\\
3.54043	0.00128644	\\
3.74868	5.77E-04	\\
3.95694	2.47E-04	\\
4.16519	1.00E-04	\\
4.37344	3.90E-05	\\
4.5817	1.44E-05	\\
4.78995	5.09E-06	\\
4.9982	1.71E-06	\\
5.20646	5.48E-07	\\
};
\addlegendentry{Layer 3}

\addplot [color=mycolor2, line width=1.5pt, mark=o, mark options={solid, mycolor2}]
  table[row sep=crcr]{%
0.000564475	0.765601	\\
0.243964	0.717391	\\
0.487364	0.643463	\\
0.730764	0.552465	\\
0.974165	0.454046	\\
1.21756	0.357198	\\
1.46096	0.268987	\\
1.70436	0.193896	\\
1.94776	0.133789	\\
2.19116	0.0883657	\\
2.43456	0.0558678	\\
2.67796	0.0338107	\\
2.92136	0.0195867	\\
3.16476	0.0108613	\\
3.40816	0.00576519	\\
3.65156	0.00292929	\\
3.89496	0.0014247	\\
4.13836	6.63E-04	\\
4.38176	2.96E-04	\\
4.62516	1.26E-04	\\
4.86856	5.15E-05	\\
};
\addlegendentry{Layer 2}

\end{axis}

\end{tikzpicture}%
	    \caption{Input fitting: \\
	    Layer 2 $ \approx ReLU(\mathcal{N}(-2.40\Esf-01,2.28\Esf+00))$ \\
	    Layer 3 $ \approx ReLU(\mathcal{N}(1.69\Esf-01,1.66\Esf+00))$}
	    \label{fig:Gaussian_modeling_d}
	\end{subfigure}

	\caption{Gaussian modeling for the dense portion of each layer for the MNIST classification task at mini-batch iteration 389 / 937 with the DNN model defined in Fig. \ref{fig:mnist_dnn_sketch} and parameters in Table \ref{tab:DNN_parameters}.} \label{fig:Gaussian_modeling}
\end{figure}

\begin{table}[b]
	\footnotesize
	\centering
	\caption{Sparseness of the matrices modeled in Fig. \ref{fig:Gaussian_modeling}.}
	\label{tab:Gaussian_modeling_sparseness}
	\begin{tabular}{|c|c|c|c|}
		\hline
		Layer  & Gradients & Weight & Input \\ \hline
		1   & 50.09\% & 0.15\% & -          \\ \hline
		2 &  59.09\% & 0.11\% & 33.11\%     \\ \hline
		3 &  57.97\% & 0.10\% & 38.63\%      \\ \hline  
	\end{tabular}
\end{table}



%
%



\section{Relevant Results}
\label{sec:Relevant Results}

\subsection{Approximate Matrix Computation}
\label{sec:Approximate Matrix Computation}

%
In many applications, one is interested in the trade-off between the time required for computation and the precision of the computation. 
For instance, consider the scenario in which the PS is training a DNN so that the matrix multiplication corresponds to the gradient back-propagation operation. 
In this scenario, an approximate evaluation of the matrix product results in an overall improvement of the training time as the speed-up benefits far out-weight the loss of precision in the computation results. 

Approximate matrix multiplication has a long history in mathematics, computer science, and engineering. 
The problem was initially considered in \cite{frieze2004fast}, inspired by the problem of finding low-rank matrix approximations in \cite{drineas2006fast}. 
%
%
%
When considering a simple algorithm in which the matrix multiplication is approximated by randomly sampling
column from $\Av$ and  rows from $\Bv$ (i.e., an outer product)  and accumulating these rank-one  matrices to produce an approximation of $\Cv$, approximation loss of
\ea{ \tag{10}
\Lcal(\Cv,\Chv)= \Ocal \lb \f {\| \Av \|_F \|\Bv \|_F }{ \sqrt{c}}\rb,
\label{eq:ex loss}
}
is obtained where $c$ is the number of outer products accumulated to produce $\Chv$ \cite{drineas2006fast}. A similar interpretation is also presented in \cite{charalambides2021approximate} for block/submatix sampling.


The above approach is referred as random projection or “sketching.” Many studies are based on the Johnson-Lindenstrauss (JL) Lemma \cite{dasgupta2003elementary, sarlos2006improved, clarkson2009numerical}.
%
%
 %
 In \cite{kane2014sparser}, the authors aims more speed up by introducing a sparse JL. 
 In \cite{clarkson2017low}, the authors come up with a  very sparse subspace embedding matrix
 with modifications where their results can be used to speed up approximate matrix multiplication by the reduction in the dimension after the JL projection. 


\subsection{UEP Codes}
\label{sec:UEP codes}

\begin{figure}[t]
	\centering
	\begin{tikzpicture}

\definecolor{apricot}{rgb}{0.98, 0.81, 0.69}
\definecolor{antiquebrass}{rgb}{0.8, 0.58, 0.46}
\definecolor{arylideyellow}{rgb}{0.91, 0.84, 0.42}
\definecolor{bananamania}{rgb}{0.98, 0.91, 0.71}
\definecolor{babypink}{rgb}{0.54, 0.81, 0.94}
\definecolor{babyblue}{rgb}{0.96, 0.76, 0.76}
\definecolor{caribbeangreen}{rgb}{0.0, 0.8, 0.6}
\definecolor{celadon}{rgb}{0.67, 0.88, 0.69}
\definecolor{cornsilk}{rgb}{1.0, 0.97, 0.86}	

	\definecolor{battleshipgrey}{rgb}{0.52, 0.52, 0.51}
	\definecolor{antiquewhite}{rgb}{0.98, 0.92, 0.84}
	\definecolor{darkgray}{rgb}{0.66, 0.66, 0.66}
	\definecolor{gainsboro}{rgb}{0.86, 0.86, 0.86}	
	\definecolor{gray(x11gray)}{rgb}{0.815, 0.815, 0.815}		
	\definecolor{isabelline}{rgb}{0.96, 0.94, 0.93}
	\definecolor{ivory}{rgb}{1.0, 1.0, 0.94}	
\definecolor{linen}{rgb}{0.98, 0.94, 0.9}
\definecolor{pastelgray}{rgb}{0.81, 0.81, 0.77}
	\definecolor{platinum}{rgb}{0.9, 0.89, 0.89}
	\definecolor{anti-flashwhite}{rgb}{0.95, 0.95, 0.96}
	\definecolor{beige}{rgb}{0.96, 0.96, 0.86}
	\definecolor{champagne}{rgb}{0.97, 0.91, 0.81}
	\definecolor{smokyblack}{rgb}{0.25, 0.25, 0.25}		
	
	\draw [fill=smokyblack](-25,26.25) rectangle (-23,23.5);
	\draw [fill=champagne](-25,26) rectangle (-23,24);
	\draw [fill=battleshipgrey](-23,26.25) rectangle (-21,23.5);
	\draw [fill=champagne](-23,26) rectangle (-19.25,24);
	\draw [fill=gray(x11gray)](-19.25,26.25) rectangle (-17.25,23.5);
	\draw [fill=champagne](-19.25,26) rectangle (-17.25,24);
	
		\node  (0) at (-25, 26) {};
		\node  (1) at (-24.5, 26) {};
		\node  (2) at (-24, 26) {};
		\node  (3) at (-25, 24) {};
		\node  (4) at (-24.5, 24) {};
		\node  (5) at (-24, 24) {};
		\node  (6) at (-23, 24) {};
		\node  (7) at (-23, 26) {};
		\node  (8) at (-23, 26) {};
		\node  (9) at (-22.5, 26) {};
		\node  (10) at (-22, 26) {};
		\node  (11) at (-23, 24) {};
		\node  (12) at (-22.5, 24) {};
		\node  (13) at (-22, 24) {};
		\node  (14) at (-21.5, 24) {};
		\node  (15) at (-21.5, 26) {};
		\node  (16) at (-19.25, 26) {};
		\node  (17) at (-18.75, 26) {};
		\node  (18) at (-18.25, 26) {};
		\node  (19) at (-19.25, 24) {};
		\node  (20) at (-18.75, 24) {};
		\node  (21) at (-18.25, 24) {};
		\node  (22) at (-17.75, 24) {};
		\node  (23) at (-17.75, 26) {};
		\node  (24) at (-25, 26.25) {};
		\node  (25) at (-23, 26.25) {};
		\node  (27) at (-25, 23.75) {};
		\node  (28) at (-23, 23.75) {};
		\node  (30) at (-19.25, 26.25) {};
		\node  (33) at (-19.25, 23.75) {};
		\node  (34) at (-24.75, 25) {};
		\node  (35) at (-24.25, 25) {};
		\node  (36) at (-17.5, 25) {};
		\node  (37) at (-20.15, 25) {$\cdots$};
		\node  (38) at (-24, 23.75) {\color{white} $1$-st window};
		\node  (39) at (-22, 23.75) {\color{white}$2$-nd window};
		\node  (40) at (-18.25, 23.75) {\color{white}$L$-th window};
		\node  (41) at (-23.5, 26) {};
		\node  (42) at (-23.5, 24) {};
		\node  (46) at (-21, 26) {};
		\node  (47) at (-21, 24) {};
		\node  (48) at (-17.25, 26) {};
		\node  (49) at (-17.25, 24) {};
		\node  (50) at (-17.25, 26.25) {};
		\node  (51) at (-17.25, 23.75) {};
		\node  (52) at (-21, 26.25) {};
		\node  (53) at (-21, 23.75) {};
		\draw (0.center) to (1.center);
		\draw (1.center) to (2.center);
		\draw (2.center) to (7.center);
		\draw (2.center) to (5.center);
		\draw (5.center) to (6.center);
		\draw (4.center) to (5.center);
		\draw (4.center) to (3.center);
		\draw (3.center) to (0.center);
		\draw (1.center) to (4.center);
		\draw (8.center) to (9.center);
		\draw (9.center) to (10.center);
		\draw (10.center) to (15.center);
		\draw (10.center) to (13.center);
		\draw (13.center) to (14.center);
		\draw (12.center) to (13.center);
		\draw (12.center) to (11.center);
		\draw (11.center) to (8.center);
		\draw (9.center) to (12.center);
		\draw (16.center) to (17.center);
		\draw (17.center) to (18.center);
		\draw (18.center) to (23.center);
		\draw (18.center) to (21.center);
		\draw (21.center) to (22.center);
		\draw (20.center) to (21.center);
		\draw (20.center) to (19.center);
		\draw (19.center) to (16.center);
		\draw (17.center) to (20.center);
		\draw (15.center) to (14.center);
		\draw (23.center) to (22.center);
		\draw (14.center) to (19.center);
		\draw (16.center) to (15.center);
		\draw (24.center) to (25.center);
		\draw (27.center) to (24.center);
		\draw (33.center) to (30.center);
		\draw (25.center) to (28.center);
		\draw (41.center) to (42.center);
		\draw (46.center) to (47.center);
		\draw (30.center) to (50.center);
		\draw (23.center) to (48.center);
		\draw (48.center) to (49.center);
		\draw (49.center) to (22.center);
		\draw (50.center) to (51.center);
		\draw (25.center) to (52.center);
		\draw (52.center) to (53.center);

\end{tikzpicture}
    \caption{Window definition for UEP-NOW codes in \cite{vukobratovic2012unequal}.} \label{sm_NOW} 
	\centering
	\begin{tikzpicture}
\definecolor{apricot}{rgb}{0.98, 0.81, 0.69}
\definecolor{antiquebrass}{rgb}{0.8, 0.58, 0.46}
\definecolor{arylideyellow}{rgb}{0.91, 0.84, 0.42}
\definecolor{bananamania}{rgb}{0.98, 0.91, 0.71}
\definecolor{babypink}{rgb}{0.54, 0.81, 0.94}
\definecolor{babyblue}{rgb}{0.96, 0.76, 0.76}
\definecolor{caribbeangreen}{rgb}{0.0, 0.8, 0.6}
\definecolor{celadon}{rgb}{0.67, 0.88, 0.69}
\definecolor{cornsilk}{rgb}{1.0, 0.97, 0.86}	

\definecolor{battleshipgrey}{rgb}{0.52, 0.52, 0.51}
\definecolor{antiquewhite}{rgb}{0.98, 0.92, 0.84}
\definecolor{darkgray}{rgb}{0.66, 0.66, 0.66}
\definecolor{gainsboro}{rgb}{0.86, 0.86, 0.86}	
\definecolor{gray(x11gray)}{rgb}{0.815, 0.815, 0.815}		
\definecolor{isabelline}{rgb}{0.96, 0.94, 0.93}
\definecolor{ivory}{rgb}{1.0, 1.0, 0.94}	
\definecolor{linen}{rgb}{0.98, 0.94, 0.9}
\definecolor{pastelgray}{rgb}{0.81, 0.81, 0.77}
\definecolor{platinum}{rgb}{0.9, 0.89, 0.89}
\definecolor{anti-flashwhite}{rgb}{0.95, 0.95, 0.96}
\definecolor{beige}{rgb}{0.96, 0.96, 0.86}
\definecolor{champagne}{rgb}{0.97, 0.91, 0.81}
\definecolor{smokyblack}{rgb}{0.25, 0.25, 0.25}

\draw [fill=gray(x11gray)](-25,26.75) rectangle (-17.25,22.5);
\draw [fill=battleshipgrey](-25,26.5) rectangle (-21,23);
\draw [fill=smokyblack](-25,26.25) rectangle (-23,23.5);

\draw [fill=champagne](-25,26) rectangle (-23,24);
\draw [fill=champagne](-23,26) rectangle (-21,24);
\draw [fill=champagne](-21,26) rectangle (-19.25,24);
\draw [fill=champagne](-19.25,26) rectangle (-17.25,24);

		\node [] (0) at (-25, 26) {};
		\node [] (1) at (-24.5, 26) {};
		\node [] (2) at (-24, 26) {};
		\node [] (3) at (-25, 24) {};
		\node [] (4) at (-24.5, 24) {};
		\node [] (5) at (-24, 24) {};
		\node [] (6) at (-23, 24) {};
		\node [] (7) at (-23, 26) {};
		\node [] (8) at (-23, 26) {};
		\node [] (9) at (-22.5, 26) {};
		\node [] (10) at (-22, 26) {};
		\node [] (11) at (-23, 24) {};
		\node [] (12) at (-22.5, 24) {};
		\node [] (13) at (-22, 24) {};
		\node [] (14) at (-21.5, 24) {};
		\node [] (15) at (-21.5, 26) {};
		\node [] (16) at (-19.25, 26) {};
		\node [] (17) at (-18.75, 26) {};
		\node [] (18) at (-18.25, 26) {};
		\node [] (19) at (-19.25, 24) {};
		\node [] (20) at (-18.75, 24) {};
		\node [] (21) at (-18.25, 24) {};
		\node [] (22) at (-17.75, 24) {};
		\node [] (23) at (-17.75, 26) {};
		\node [] (34) at (-24.75, 25) {};
		\node [] (35) at (-24.25, 25) {};
		\node [] (36) at (-17.5, 25) {};
		\node [] (37) at (-20.15, 25) {$\cdots$};
		\node [] (38) at (-24, 23.75) {{\color{white}$1$-st window}};
		\node [] (39) at (-23, 23.25) {{\color{white}$2$-nd window}};
		\node [] (40) at (-21, 22.75) {{\color{white}$L$-th window}};
		\node [] (41) at (-23.5, 26) {};
		\node [] (42) at (-23.5, 24) {};
		\node [] (46) at (-21, 26) {};
		\node [] (47) at (-21, 24) {};
		\node [] (48) at (-17.25, 26) {};
		\node [] (49) at (-17.25, 24) {};
		\node [] (50) at (-25, 26.25) {};
		\node [] (51) at (-23, 26.25) {};
		\node [] (52) at (-23, 23.5) {};
		\node [] (53) at (-25, 23.5) {};
		\node [] (54) at (-25, 23) {};
		\node [] (55) at (-23, 23.5) {};
		\node [] (56) at (-21, 23) {};
		\node [] (57) at (-25, 26.5) {};
		\node [] (59) at (-21, 26.5) {};
		\node [] (60) at (-25, 26.75) {};
		\node [] (63) at (-17.25, 26.75) {};
		\node [] (64) at (-17.25, 22.5) {};
		\node [] (65) at (-25, 22.5) {};

		\draw (0.center) to (1.center);
		\draw (1.center) to (2.center);
		\draw (2.center) to (7.center);
		\draw (2.center) to (5.center);
		\draw (5.center) to (6.center);
		\draw (4.center) to (5.center);
		\draw (4.center) to (3.center);
		\draw (3.center) to (0.center);
		\draw (1.center) to (4.center);
		\draw (8.center) to (9.center);
		\draw (9.center) to (10.center);
		\draw (10.center) to (15.center);
		\draw (10.center) to (13.center);
		\draw (13.center) to (14.center);
		\draw (12.center) to (13.center);
		\draw (12.center) to (11.center);
		\draw (11.center) to (8.center);
		\draw (9.center) to (12.center);
		\draw (16.center) to (17.center);
		\draw (17.center) to (18.center);
		\draw (18.center) to (23.center);
		\draw (18.center) to (21.center);
		\draw (21.center) to (22.center);
		\draw (20.center) to (21.center);
		\draw (20.center) to (19.center);
		\draw (19.center) to (16.center);
		\draw (17.center) to (20.center);
		\draw (15.center) to (14.center);
		\draw (23.center) to (22.center);
		\draw (14.center) to (19.center);
		\draw (16.center) to (15.center);
		\draw (41.center) to (42.center);
		\draw (46.center) to (47.center);
		\draw (23.center) to (48.center);
		\draw (48.center) to (49.center);
		\draw (49.center) to (22.center);
		\draw (50.center) to (51.center);
		\draw (51.center) to (52.center);
		\draw (52.center) to (53.center);
		\draw (53.center) to (50.center);
		\draw (57.center) to (59.center);
		\draw (59.center) to (56.center);
		\draw (56.center) to (54.center);
		\draw (54.center) to (57.center);
		\draw (60.center) to (63.center);
		\draw (60.center) to (65.center);
		\draw (65.center) to (64.center);
		\draw (64.center) to (63.center);

\end{tikzpicture}
    \caption{Window definition with UEP-EW  codes \cite{vukobratovic2012unequal}.}  \label{sm_EW}
\end{figure}   

The substantial reduction in complexity attained through approximate matrix multiplication discussed in Sec. \ref{sec:Approximate Matrix Computation} comes from identifying the matrix features which influence the multiplication outcome the most. 
In the following, we will rely on this observation and also choose the protection levels of our matrix products according to this influence.
To do so, we rely on UEP codes. 

UEP codes were firstly introduced in \cite{masnick1967linear} as a way to provide stronger protection for certain data digits while the other digits are protected less. 
The authors perform analysis for linear codes to protect a set of digits at a higher level than other sets of digits and provide methods of synthesizing UEP codes using parity check matrices. 
Since then, there have been many studies which investigate UEP codes to provide unequal protection for different positions in a data sequence.
%
%

UEP codes are also studied for multimedia transmission schemes with network coding, see e.g.,  \cite{thomos2009randomized, nguyen2010video}. These studies divide the source packet into several layers and importance levels to apply an RLC scheme which provides better protection for certain sub-packets for prior recovery. 
In \cite{vukobratovic2012unequal}, the authors consider a similar setup; further, they also provide performance analysis for packet-level UEP coding
over packet erasure channels.

While there are different ways of UEP coding, here we focus on the UEP codes introduced in \cite{vukobratovic2012unequal}, which also provides an analysis based on RLC.
Two families of codes are analyzed in \cite{vukobratovic2012unequal}: 
NOW-UEP and EW-UEP, respectively.  
In both codes, we consider the case of a layered message source which consists of $K$ equal-length source packets of $b$ bits each. 
Without loss of generality, assume that bits are ordered in decreasing levels of importance. 
A window containing a subset of these packets is then considered, and packets inside each window are selected according to a given probability.
Once selected, a random linear code is employed to produce parity bits.

In the NOW-UEP coding strategy, the selection windows are non-overlapping.
The encoding of the information bits is performed by selecting a window using a window selection polynomial $\Gamma(\xi)= \sum_{i \in [L]} \Gamma_i \xi^i$,
where $\Gamma_i$ is the window selection probability for the $i$-th importance level.
Then, the encoded messages are generated only from the importance levels of the selected type, as illustrated in Fig. \ref{sm_NOW}. With this coding strategy, separate and independent coding is applied to each importance level enabling different rate allocation for different levels with the help of the window selection probabilities. Thus, one can interpret this coding scheme as applying different error protection codes for each level.

The EW-UEP coding strategy also uses a probabilistic window selection polynomial $\Gamma(\xi)$ for the $i$-th importance level; however, the window definition is different from that of the NOW-UEP strategy. The EW-UEP constructs the $i$-th window by including all the packets whose importance level is $i$ or higher than $i$ as illustrated in Fig. \ref{sm_EW}. Therefore, the packets in the first window are included in all the encoded messages; hence they are the best-protected ones. However, the packets with lower importance are used less in the encoding process resulting in less protection for them. Therefore, progressive protection is provided for the source packets.
For instance, let us assume that the third importance level is selected according to the window selection distribution, then the coded packet includes all the source messages from the first, second, and third importance levels.
Thus, the EW strategy includes the most important matrices to the encoding process regardless of the importance level of the selected window to provide a better protection than the others.

It is also worth noting that encoding and decoding processes of RLC have a small complexity level which makes them feasible for real-time implementation. For instance, in \cite{shojania2009random, 5766632}, it is demonstrated that real-time implementation of RLC over wireless channels is even possible by using smartphones. Hence, the overhead introduced by our approach can be neglected compared to the complexity of matrix multiplication.

\section{Approximate Matrix Multiplication Through UEP Codes}
\label{sec:Approximate Matrix Multiplication with UEP Codes}
In this section, we propose a distributed coded approximate matrix multiplication scheme which aims to provide better protection for the matrix sub-products of $\Av_{nm}  \in \mathbb{R}^{U  \times H}$ and $\Bv_{mp}  \in \mathbb{R}^{H\times Q}$  with  larger norms, and thus produce a better  approximation of the  matrix product within the prescribed deadline.
In particular, the coding scheme relies on a  parametrization of the protection levels
matching the distribution of the Frobenius norms of the rows and columns of the matrices in question.

\subsection{Importance Level of a Sub-block}
%

\subsubsection{Importance levels for $\rtc$ multiplication}

Let us begin by classifying the matrix sub-blocks in \eqref{eq:sub-blocks rXc} according to their norms. For instance,  we may select three different levels for each sub-block, e.g., \textit{high}, \textit{medium}, and \textit{low} to classify the norms of $\Av_n$ and $\Bv_p$.
Let us refer to these levels as \emph{importance levels}, and assume that there are $S$ such levels and let the importance be decreasing in $s \in [S]$.
Given a matrix $\Av$/$\Bv$, we have $n_A(s)$/$n_B(s)$ blocks with the importance level $s \in [S]$.
Clearly, 
$N =\sum_{s } n_A(s)$, and $P =\sum_{s } n_B(s)$.

By construction, any sub-product $\Cv_{np}$ is obtained as the multiplication of sub-blocks in two classes: accordingly $\Cv_{np}$ has $L$ possible importance levels with $L = S(S+1)/2$.
For instance, for the example of three importance levels, for both $\Av_n$ and $\Bv_p$,  $\Cv_{np}$ can have importance  \textit{high $\times$ high}, \textit{high $\times$ medium}, \textit{high $\times$ low}, \textit{medium $\times$ medium}, etc.

From a high level-perspective, one would want the PS to be able to quickly recover those products corresponding to the importance level \textit{high $\times$ high},  while the importance level \textit{low $\times$ low} is not particularly urgent. 
We can obtain this desired behavior by employing UEP codes.

\subsubsection{Importance levels for $\ctr$ multiplication}
Similar to the previous case, we will classify the matrix sub-blocks in \eqref{eq:sub-blocks cXr} based on their norms into $S$ importance levels. 
However, this time the classification will be based on the column sub-blocks of $\Av$ and row sub-blocks of $\Bv$ which are denoted by $\Av_m$ and $\Bv_m$, respectively.  
Similarly, we have $n_A(s)$/$n_B(s)$ for $s \in [S]$ as the number of blocks with level of importance
learly, 
$M =\sum_{s} n_A(s)  =\sum_{s} n_B(s)$.

As a result of the multiplication of $\Av_m$ and $\Bv_m$, 
$\Cv_{m}$ will be comprised of $L$ different importance levels which will depend on the pairing of column and row importance levels of $\Av$ and $\Bv$, respectively. 
Note that, for $\ctr$ multiplication, we only have $M$ sub-block products due to multiplication of $\Av_m$ and $\Bv_m$, $m \in [M]$, which may result in less than $S(S+1)/2$ importance levels depending on the order of the sub-blocks. 
Different from the previously described block partitioning, each sub-block multiplication results in a matrix $\Cv_{m}$ whose size is same as the original multiplication $\Cv = \Av \Bv$. %
Note that, for the perfect recovery of $\Cv$, one needs to have $M$ separate $\Cv_{m}$ multiplications with $m \in [M]$.

Note that in \cite{charalambides2021approximate}, the authors consider weighted block sampling for column-times-row partitioning, and they optimize the sampling distribution based on the Frobenius norm of the multiplied sub-blocks. Similar to their study, one can decide the importance levels according to the Fronebius norms if available at the PS instead of the block statistics.

\subsection{UEP Coded Matrix Multiplication}
%
For clarity, in the following we describe the UEP coding strategies based on $\rtc$ partitioning  as given in \eqref{eq:sub-blocks rXc}. 
For the case $\ctr$, the UEP coding is performed in an analogous manner and thus shall not be detailed further. 
%
%
For the $\rtc$ partitioning, let us detail how the NOW-UEP and EW-UEP codes  in Sec.  \ref{sec:UEP codes} are applied to the matrix multiplication problem in Sec. \ref{sec:System model}.
%




The PS selects importance levels for  $\Av$ and $\Bv$, and then  encodes the corresponding rows and columns of $\Av$ and $\Bv$ for the $\rtc$ multiplication using 
\begin{equation} \tag{11}
\begin{aligned}
\Wv_{A}^w &= \sum_{i } \alpha^w(i) \Av_{\pi_A^w(i)},\\
\Wv_{B}^w &= \sum_{j} \beta^w(j) \Bv_{\pi_B^w(j)},
\end{aligned}
\end{equation}
for the $w$-th worker, $w \in [W]$, where $ \alpha^w(i)$ and $ \beta^w(j)$ are randomly
selected elements from the given finite field, and ${\pi_A^w(i)}$/${\pi_B^w(j)}$ is the row/column or column/row indices of $\Av$/$\Bv$ at the corresponding levels, respectively.

After the matrix sub-products are encoded, $\Wv_{A}^w$ and $\Wv_{B}^w$  are transmitted to the $w$-th worker. 
Each worker $w$ performs the sub-product multiplication operation resulting in $\Wv_{A}^w \Wv_{B}^w$, and  transmits the sub-product to the PS. 
Since wireless transmission requires delay and energy, we consider the delay due to stragglers' channel conditions as the communication cost.
%

The PS will have $\Wcal(T_{\max})$, which is a set of $\Wv_{A}^w \Wv_{B}^w$ products completed within a predetermined deadline $T_{\max}$.
%
At this point, to form the approximation $\Chv$, the PS  simply places the sub-product $\Chv_{nmp}=\Cv_{nmp}$ in the positions that can be obtained from $\Wcal(T_{\max})$, and sets $\Chv_{nmp}$ to the all-zero matrix otherwise.
%
With this operation, the final approximate matrix product $\Chv$ is obtained.

\section{Analysis of the Approximation Error}
\label{sec:Theoretical Analysis}
In this section, we bound the loss in \eqref{eq:loss} using UEP codes.
Our bounds rely on the results in  \cite{vukobratovic2012unequal} to characterize the performance of the proposed schemes in Sec. \ref{sec:Approximate Matrix Multiplication with UEP Codes}.
These bounds are applied to the case of matrices with i.i.d. entries as per the following assumption.

\begin{assumption}{\bf  i.i.d. entries matrix computation:}
\label{asm_iid}
For the matrix  $\Av$ and $\Bv$ we assume that the sub-matrices in the  $\rtc$ / $\ctr$ partitioning have i.i.d. entries with zero mean and variance $\sigma_{l,A}^2$ and  $\sigma_{l,B}^2$ corresponding to the $l$-th importance level of the final product $\Cv$.
\end{assumption}
%

\subsection{ Row-times-Column Case}
%
We assume for simplicity that the entries of the matrices are zero mean and with variance  $\sigma_{A_n}^2$ and $\sigma_{B_p}^2$ for the $n$-th and $p$-th sub-block of \eqref{eq:sub-blocks rXc}, respectively, for $\Av$ and $\Bv$, and they are uncorrelated, so that 
%
\ean{ \tag{12}
\mathbb{E} \Big[ \| \Cv_{np}  \|_F^2 \Big]  = UHQ \sigma_{A_n}^2 \sigma_{B_p}^2.
}
Let us denote the number of encoded matrix products received at time $t$ by $N(t)$, then the probability of receiving $w$ packets from $W$ workers at time $t$ is $P_{N(t)}(w)$, which is obtained as
\begin{equation} \tag{13}
P_{N(t)}(w) = \binom{W}{w} (1-F(t))^{W-w} \left(F(t)\right)^{w}.
\end{equation}

%

From  \cite[Eq. 5]{vukobratovic2012unequal}, we obtain a bound (which is achievable with large field sizes) on the decoding probabilities of NOW-UEP strategy for each importance level as a function of received matrices as 
\begin{equation} \tag{14} \label{dec_prob}
P_{d,l}(N)\leq \sum_{\substack{(n_{1},n_{2},\ldots,n_{L}) \\
\sum_{l} n_l=N }}
P_{\Gamma(\xi),N}(\mathbf{n}) \: \onev(n_{l}\geq k_{l}),
\end{equation}
where $\mathbf{n} = [n_{1},n_{2},\ldots,n_{L}]$, $k_{l}$ is the number of packets in class $l$, and 
\begin{equation} \tag{15}
P_{\Gamma(\xi),N}({\bf n})=\frac{N!}{n_{1}!n_{2}!\ldots n_{L}!} \Gamma_{1}^{n_{1}}\Gamma_{2}^{n_{2}}\ldots\Gamma_{L}^{n_{L}}.
\end{equation}
It is also worth noting that the decoding probabilities for EW-UEP coding can be calculated similarly as they are given in  \cite[Eq. 6-9]{ vukobratovic2012unequal}.

As an example, with three classes, $W=30$ workers, and window selection probabilities $(0.40,0.35,0.25)$, the decoding probabilities of each class with NOW and EW-UEP codes with $3$ sub-products in each level are as depicted in Fig. \ref{UEP_prob_fig}. The figure clearly illustrates how the most important class is protected better. 

\begin{figure}[t]
	\centering
		\begin{tikzpicture}
\definecolor{mycolor1}{rgb}{0.63529,0.07843,0.18431}%
\definecolor{mycolor2}{rgb}{0.00000,0.44706,0.74118}%
\definecolor{mycolor3}{rgb}{0.00000,0.49804,0.00000}%
\definecolor{mycolor4}{rgb}{0.87059,0.49020,0.00000}%
\definecolor{mycolor5}{rgb}{0.00000,0.44700,0.74100}%
\definecolor{mycolor6}{rgb}{0.74902,0.00000,0.74902}%

\definecolor{bittersweet}{rgb}{1.0, 0.44, 0.37}
\definecolor{citrine}{rgb}{0.56, 0.74, 0.56}
	\definecolor{manatee}{rgb}{0.4, 0.19, 0.28}
	\definecolor{mediumblue}{rgb}{0.0, 0.0, 0.8}

\begin{axis}[%
width=7cm,
height=3cm,
at={(0.88in,0.458in)},
scale only axis,
xmin=0,
xmax=30,
xlabel style={font=\color{white!15!black}},
xlabel={Received packets (N)},
ymin=0,
ymax=1,
ylabel style={font=\color{white!15!black}},
ylabel={Decoding probabilities},
axis background/.style={fill=white},
xmajorgrids,
ymajorgrids,
legend columns=2, 
legend style={at={(0.985,0.445)},font=\small,nodes={scale=0.85, transform shape},
        /tikz/column 2/.style={
            column sep=5pt,
        },
    },
]
\addplot [color=bittersweet, line width=1.5pt, mark=o, mark options={solid, bittersweet}]
  table[row sep=crcr]{%
1	0\\
2	0\\
3	0.064\\
4	0.1792\\
5	0.31744\\
6	0.45568\\
7	0.580096\\
8	0.68460544\\
9	0.768212992\\
10	0.8327102464\\
12	0.91655667712\\
14	0.96020841889792\\
16	0.981662785601536\\
18	0.991773643509858\\
20	0.996388527940871\\
22	0.998442476504533\\
24	0.999338206073081\\
26	0.999722330853574\\
28	0.999884754984442\\
30	0.999952616489873\\
};
\addlegendentry{${P}_{{d,1}}{(N)}$-NOW}

\addplot [color=bittersweet, dashed, line width=1.5pt, mark=o, mark options={solid, bittersweet}]
  table[row sep=crcr]{%
1	0\\
2	0\\
3	0.064\\
4	0.1792\\
5	0.31744\\
6	0.506138515625\\
7	0.6918362890625\\
8	0.830303090820312\\
9	1\\
10	1\\
12	1\\
14	1\\
16	1\\
18	1\\
20	1\\
22	1\\
24	1\\
26	0.999999999999999\\
28	1\\
30	1\\
};
\addlegendentry{EW}

\addplot [color=mediumblue, line width=1.5pt, mark=asterisk, mark options={solid, mediumblue}]
  table[row sep=crcr]{%
1	0\\
2	0\\
3	0.042875\\
4	0.12648125\\
5	0.235169375\\
6	0.35291484375\\
7	0.46771667578125\\
8	0.572186342929687\\
9	0.662726721125\\
10	0.738392608616797\\
12	0.848712421664733\\
14	0.916073426273875\\
16	0.954910263405345\\
18	0.976383071396483\\
20	0.98788229328966\\
22	0.993887496916696\\
24	0.996960498937053\\
26	0.998506888416244\\
28	0.999274192675958\\
30	0.999650392667499\\
};
\addlegendentry{${P}_{{d,2}}{(N)}$-NOW}

\addplot [color=mediumblue, dashed, line width=1.5pt, mark=asterisk, mark options={solid,mediumblue}]
  table[row sep=crcr]{%
1	0\\
2	0\\
3	0\\
4	0\\
5	0\\
6	0.105338515625\\
7	0.2798102890625\\
8	0.456098930820312\\
9	0.735924768\\
10	0.8224260864\\
12	0.92896489472\\
14	0.97507506962432\\
16	0.992094772953088\\
18	0.997683850536812\\
20	0.999363115161527\\
22	0.999833705556784\\
24	0.999958404053453\\
26	0.999989964659351\\
28	0.999997652340838\\
30	0.999999465184212\\
};
\addlegendentry{EW}

\addplot [color=manatee, line width=1.5pt,mark=x, mark options={solid, manatee}]
  table[row sep=crcr]{%
1	0\\
2	0\\
3	0.015625\\
4	0.05078125\\
5	0.103515625\\
6	0.16943359375\\
7	0.24359130859375\\
8	0.321456909179688\\
9	0.399322509765625\\
10	0.474407196044922\\
12	0.609324991703034\\
14	0.718872375786305\\
16	0.802888950100169\\
18	0.864694957272149\\
20	0.908739567535121\\
22	0.939350571667887\\
24	0.960198812657051\\
26	0.974162674603001\\
28	0.983385251462701\\
30	0.989404129316579\\
};
\addlegendentry{${P}_{{d,3}}{(N)}$-NOW}

\addplot [color=manatee,line width=1.5pt, dashed, mark=x, mark options={solid, manatee}]
  table[row sep=crcr]{%
1	0\\
2	0\\
3	0\\
4	0\\
5	0\\
6	0\\
7	0\\
8	0\\
8.99 0\\
9	0.273742509765625\\
10	0.379607196044922\\
12	0.565085631703033\\
14	0.701782430186304\\
16	0.797105004884169\\
18	0.862919722298453\\
20	0.908234162291729\\
22	0.939215010554648\\
24	0.96016417099703\\
26	0.974154170327237\\
28	0.983383233081949\\
30	0.989403663898145\\
};
\addlegendentry{EW}

\end{axis}
\end{tikzpicture}%
		\caption{Decoding probabilities of NOW-UEP and EW-UEP strategies with three classes, and $W=30$ workers.}  \label{UEP_prob_fig}
		\end{figure}
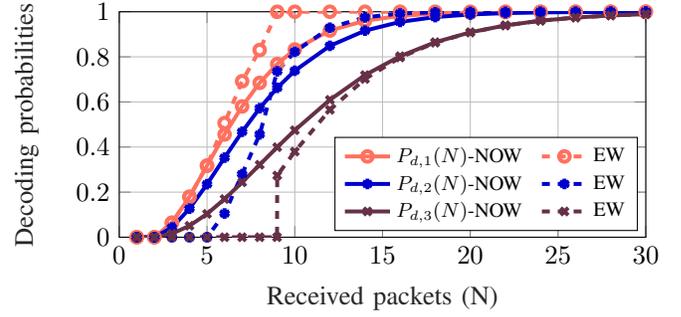

We can bound the performance of the coded matrix multiplication scheme in Sec.  \ref{sec:Approximate Matrix Multiplication with UEP Codes} as follows.  
%
\begin{theorem}{\bf UEP loss for the row-times-column partitioning:}
	\label{thm:NOW-UEP}
	Consider the loss minimization problem in Sec. \ref{sec:System model} for the case in which the set of matrix products $\Dcal(\{\Av,\Bv\})$ that satisfy Assumption \ref{asm_iid}.
	The expected value of the loss  in \eqref{eq:loss}
	%
	attainable with the NOW-UEP strategy described in  Sec. \ref{sec:Approximate Matrix Multiplication with UEP Codes} for $\rtc$ partitioning is 
	\begin{equation} \label{mse_t} \tag{16}
	\Ebb\lsb \Lcal (T_{\max})\rsb =  
	\sum_{w}  
	P_{N(T_{\max})}(w) \mathbb{E} [\|\Cv- \hat{\Cv}\|_F^2 | N(T_{\max}) = w ],
	\end{equation}
	with 
	\begin{align} 
	\mathbb{E} [\|\Cv-& \hat{\Cv}\|_F^2 | N(t)=w ]  \nonumber \\
	&= UHQ
	\sum_{l } 
	k_l \left(1-P_{d,l}(w)\right)
     \sigma_{l,A}^2 \sigma_{l,B}^2, \tag{17}
	\end{align}
	where $k_l$ is the number of blocks in the $l$-th importance level of $\Cv$, and the expectation in \eqref{mse_t} is taken over the random entries of $\Av$, $\Bv$.
\end{theorem}
Note that, since \eqref{dec_prob} is in fact an upper bound on the correct recovery probability, applying it to \eqref{mse_t} results in a lower bound on the expected loss, however, this bound is tight as the field size tends to infinity, i.e., the lower bound on the loss is asymptotically achievable. 
The analog of Theorem \ref{thm:NOW-UEP} for the EW-UEP can be obtained using corresponding decoding probabilities $P_{d,l}(N)$'s from the results in \cite[Eq. 6-9]{ vukobratovic2012unequal}. They are not presented here for brevity.
\begin{remark}
Note that, in  Theorem \ref{thm:NOW-UEP}, there exists a ``matching'' between the probabilistic structure of the matrices to be multiplied and their $\rtc$ block partitioning.
In reality, one would not observe such a neat organization of the matrix values, and instead would have to fit the row/column weight distribution in the data to design the UEP code resulting in the minimal loss.
\end{remark}

\subsection{ Column-times-Row Case}
We now consider the case of column-times-row multiplication with partitioning \eqref{eq:sub-blocks cXr}. 
With a similar sub-block partitioning, we assume for simplicity that the entries of the matrix are zero mean and with variance  $\sigma_{A_m}^2$ and $\sigma_{B_m}^2$ for the $m$-th sub-blocks of \eqref{eq:sub-blocks cXr}, respectively, for $\Av$ and $\Bv$. 
Since $\Av_m \in \mathbb{R}^{U \times H}$ and $\Bv_m \in \mathbb{R}^{H \times Q}$ for $m \in [M]$, the resulting multiplication is $\Cv_{m} \in  \mathbb{R}^{U \times Q}$ with the same dimension of the original multiplication result $\Cv$. 
Note that, the Frobenius norm of  $\Cv_{m}$ can be easily calculated as  $\mathbb{E} \Big[ \| \Cv_{m}  \|_F^2 \Big]  = UHQ \sigma_{A_m}^2 \sigma_{B_m}^2$.
Similar to the previous section, the decoding probability \eqref{dec_prob} as  given in \cite[Eq. 5]{vukobratovic2012unequal} is used for $\ctr$ matrix multiplication with NOW-UEP codes. 
Same approach can be extended to the EW-UEP code by using \cite[Eq. 6-9]{vukobratovic2012unequal}. 
Assuming that there are $W$ workers where $N(t)$ encoded matrix products are received up to time $t$, we have the following analog of Theorem \ref{thm:NOW-UEP} which can be used with both NOW- and EW-UEP schemes by using the relevant decoding probabilities.
\begin{theorem}{\bf UEP loss for the column-times-row multiplication:}
	\label{thm:NOW-UEP cXr}
	Consider the loss minimization problem in Sec. \ref{sec:System model} for the case in which the set of matrix product $\Dcal(\{\Av,\Bv\})$ is the set of matrices under Assumption \ref{asm_iid}. 
	The expected value of the loss  in \eqref{eq:loss}
	with the NOW-UEP strategy described in  Sec. \ref{sec:Approximate Matrix Multiplication with UEP Codes} for $\ctr$ partitioning is 
	\begin{equation} \label{mse_t_cXr} \tag{18}
	\Ebb\lsb \Lcal (T_{\max})\rsb =  
	\sum_{w}  
	P_{N(T_{\max})}(w) \mathbb{E} [\|\Cv- \hat{\Cv}\|_F^2 | N(T_{\max}) = w ],
	\end{equation}
	with 
	\begin{align} 
	\mathbb{E} [\|\Cv-& \hat{\Cv}\|_F^2 | N(t)=w ]  \nonumber \\
	&\leq MUHQ 
	\sum_{l} 
	k_l \left(1-P_{d,l}(w)\right)
	 \sigma_{l,A}^2 \sigma_{l,B}^2, \tag{19}
	\end{align}
		where $k_l$ is the number of blocks in the $l$-th importance level of $\Cv$, and the expectation in \eqref{mse_t_cXr} is taken over the random entries of $\Av$, $\Bv$.
\end{theorem}

\begin{proof}
Consider the partitioning given in \eqref{eq:sub-blocks cXr}.
\begin{align} 
\|\Cv- \hat{\Cv}\|_F &\stackrel{\text{(a)}}{\leq} 
\sum_{m}
\|\Cv_{m}- \hat{\Cv}_{m}\|_F, \tag{20}\\
\|\Cv- \hat{\Cv}\|_F^2 &\leq \left(
\sum_{m} 
\|\Cv_{m}- \hat{\Cv}_{m}\|_F \right)^2, \tag{21}\\ &\stackrel{\text{(b)}}{\leq} M 
 \sum_{m} 
\|\Cv_{m}- \hat{\Cv}_{m}\|_F^2, \tag{22}
\end{align}
where (a) is the result of triangular inequality, and (b) follows Cauchy-Schwartz inequality. The proof is concluded by taking the expectation on the random entries by conditioning on the number of received packets at time $t$. 
\end{proof}

\section{Matrix Approximation with Synthetic Data}
\label{sec:Numerical Examples}

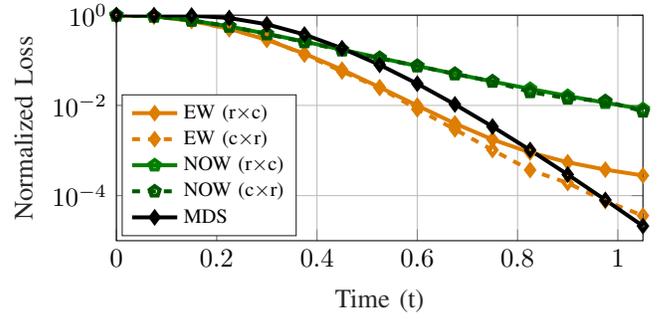
\begin{figure}[t]
	\centering
	\begin{tikzpicture}
\definecolor{mycolor1}{rgb}{0.63529,0.07843,0.18431}%
\definecolor{mycolor2}{rgb}{0.00000,0.44706,0.74118}%
\definecolor{mycolor3}{rgb}{0.00000,0.49804,0.00000}%
\definecolor{mycolor3x}{rgb}{0.00000,0.34804,0.00000}%
\definecolor{mycolor4}{rgb}{0.87059,0.49020,0.00000}%
\definecolor{mycolor5}{rgb}{0.00000,0.44700,0.74100}%
\definecolor{mycolor6}{rgb}{0.74902,0.00000,0.74902}%
\definecolor{mycolor_bt}{rgb}{0.00000,0.24706,0.54118}%

\begin{semilogyaxis}[%
width=7cm,
height=3cm,
at={(0.88in,0.458in)},
scale only axis,
xmin=0,
xmax=1.05,
xlabel style={font=\color{white!15!black}},
xlabel={Time (t)},
ymode=log,
ymin=1e-05,
ymax=1,
yminorticks=true,
ylabel style={font=\color{white!15!black}},
ylabel={Normalized Loss},
axis background/.style={fill=white},
title style={font=\bfseries},
xmajorgrids,
ymajorgrids,
yminorgrids,
legend style={at={(0.01,0.02)}, font=\small,nodes={scale=0.85, transform shape},anchor=south west, legend cell align=left, align=left, draw=white!15!black}
]
\addplot [color=mycolor4, line width=1.5pt, mark=diamond, mark options={solid, mycolor4}]
  table[row sep=crcr]{%
0	1\\
0.075	0.945207686112801\\
0.15	0.751054614619548\\
0.225	0.500525497220872\\
0.30	0.284085945665539\\
0.375	0.140317662213895\\
0.45	0.0619527928817929\\
0.525	0.0252561413118881\\
0.60	0.00991739424640335\\
0.675	0.00397892569345011\\
0.75	0.00175602273338441\\
0.825	0.000905060622182875\\
0.9	0.000548681346237125\\
0.975	0.000376442860446157\\
1.05	0.000279687729728531\\
};
\addlegendentry{EW (r$\times$c)}

\addplot [color=mycolor4, dashed, line width=1.5pt, mark=diamond, mark options={solid, mycolor4}]
  table[row sep=crcr]{%
0	1\\
0.075	0.939425575929824\\
0.15	0.759541593248981\\
0.225	0.490399537070972\\
0.30	0.297824422076944\\
0.375	0.134056358726003\\
0.45	0.0565991636743375\\
0.525	0.0244467825176527\\
0.60	0.00823205817200589\\
0.675	0.00294350044051513\\
0.75	0.00102742116022025\\
0.825	0.000370171952288604\\
0.9	0.000190508579183213\\
0.975	7.5868667165015e-05\\
1.05	3.60581275621305e-05\\
};
\addlegendentry{EW (c$\times$r)}

\addplot [color=mycolor3, line width=1.5pt, mark=pentagon, mark options={solid, mycolor3}]
  table[row sep=crcr]{%
0	1\\
0.075	0.945542835188147\\
0.15	0.768893740051912\\
0.225	0.562998365040324\\
0.30	0.388640148130495\\
0.375	0.260185069151791\\
0.45	0.171874608311479\\
0.525	0.113255721376389\\
0.60	0.0749708861658958\\
0.675	0.0500882277135346\\
0.75	0.0338795068726435\\
0.825	0.0232479435756993\\
0.9	0.0162048036748147\\
0.975	0.011482846394433\\
1.05	0.00827499417912511\\
};
\addlegendentry{NOW (r$\times$c)}

\addplot [color=mycolor3x, dashed, line width=1.5pt, mark=pentagon, mark options={solid, mycolor3x}]
  table[row sep=crcr]{%
0	1\\
0.075	0.948726423133293\\
0.15	0.769218492769717\\
0.225	0.569792751260768\\
0.30	0.382334816896922\\
0.375	0.254949324146805\\
0.45	0.167615423585439\\
0.525	0.110733845441812\\
0.60	0.0749484025554908\\
0.675	0.0497839856723916\\
0.75	0.0341127526407141\\
0.825	0.0198981001872932\\
0.9 	0.0143092614495887\\
0.975	0.0124803867549797\\
1.05	0.00737413092978957\\
};
\addlegendentry{NOW (c$\times$r)}

\addplot [color=black, line width=1.5pt, mark=diamond, mark options={solid, black}]
  table[row sep=crcr]{%
0	1\\
0.075	0.999810788415944\\
0.15	0.982178724511079\\
0.225	0.866792070087092\\
0.30	0.62999536462713\\
0.375	0.372995622444694\\
0.45	0.184922984196325\\
0.525	0.0793464082436098\\
0.60	0.0303086382082969\\
0.675	0.0105377104412094\\
0.75	0.00339238705121272\\
0.825	0.00102470972665223\\
0.9	0.000293450427090071\\
0.975	8.03291076830972e-05\\
1.05	2.11581037760224e-05\\
};
\addlegendentry{MDS}

\end{semilogyaxis}

\end{tikzpicture}%
	    \caption{Normalized loss of the estimator using UEP codes with three classes with exponential latency model. \label{three_cls_mse}}
\end{figure}

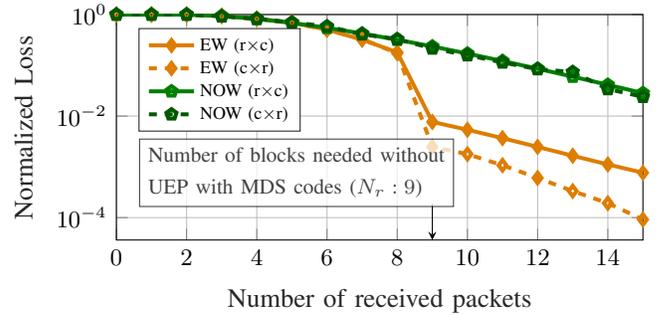
\begin{figure}[t]
	\centering
		\begin{tikzpicture}
\definecolor{mycolor1}{rgb}{0.63529,0.07843,0.18431}%
\definecolor{mycolor2}{rgb}{0.00000,0.44706,0.74118}%
\definecolor{mycolor3}{rgb}{0.00000,0.49804,0.00000}%
\definecolor{mycolor3x}{rgb}{0.00000,0.34804,0.00000}%
\definecolor{mycolor4}{rgb}{0.87059,0.49020,0.00000}%
\definecolor{mycolor5}{rgb}{0.00000,0.44700,0.74100}%
\definecolor{mycolor6}{rgb}{0.74902,0.00000,0.74902}%
\definecolor{mycolor_bt}{rgb}{0.00000,0.24706,0.54118}%

\begin{semilogyaxis}[%
font=\small,
width=7cm,
height=3cm,
scale only axis,
xmin=0,
xmax=15,
xlabel style={font=\color{white!15!black}},
xlabel={Number of received packets},
ymin=0,
ymax=1,
ylabel style={font=\color{white!15!black}},
ylabel={Normalized Loss},
axis background/.style={fill=white},
xmajorgrids,
ymajorgrids,
legend style={at={(0.35,0.94)},font=\small,nodes={scale=0.75, transform shape}, legend cell align=left, align=left, draw=white!15!black,
        /tikz/column 2/.style={
            column sep=5pt,
        },
    },
]
\addplot [color=mycolor4, line width=1.5pt, mark=diamond, mark options={solid, mycolor4}]
  table[row sep=crcr]{%
0	1\\
1	1\\
2	1\\
3	0.937667397126857\\
4	0.825468711955198\\
5	0.690830289749209\\
6	0.504483098596908\\
7	0.319375167967799\\
8	0.180223458397059\\
9	0.00766771990056992\\
10	0.0053810910610386\\
11	0.00367554804333531\\
12	0.00247088169127778\\
13	0.00165395129077423\\
14	0.0011151731254597\\
15	0.00076517307267959\\
};
\addlegendentry{EW (r$\times$c)}

\addplot [color=mycolor4, dashed, line width=1.5pt, mark=diamond, mark options={solid, mycolor4}]
  table[row sep=crcr]{%
0	1\\
1	1\\
2	1\\
3	0.944012616022131\\
4	0.828035023034128\\
5	0.680020579656033\\
6	0.504311427759591\\
7	0.327595337668321\\
8	0.173500683036074\\
9	0.00248047543168654\\
10	0.00177955783168587\\
11	0.00107353014997552\\
12	0.000604686723597005\\
13	0.000333324742079006\\
14	0.00019261915276529\\
15	9.16811255447639e-05\\
};
\addlegendentry{EW (c$\times$r)}

\addplot [color=mycolor3, line width=1.5pt, mark=pentagon, mark options={solid, mycolor3}]
  table[row sep=crcr]{%
0	1\\
1	1\\
2	1\\
3	0.936596816411006\\
4	0.822302509435111\\
5	0.684927794771122\\
6	0.547310075595021\\
7	0.42321382840558\\
8	0.318750768771067\\
9	0.234984198275905\\
10	0.170196957186756\\
11	0.121467383334207\\
12	0.085618888339811\\
13	0.0597168226101168\\
14	0.0412775444729554\\
15	0.0283131305876915\\
};
\addlegendentry{NOW (r$\times$c)}

\addplot [color=mycolor3x, dashed, line width=1.5pt, mark=pentagon, mark options={solid, mycolor3x}]
  table[row sep=crcr]{%
0	1\\
1	1\\
2	1\\
3	0.92028908118834\\
4	0.826982980545032\\
5	0.700731716374683\\
6	0.583907467743197\\
7	0.407950661845499\\
8	0.327729929136885\\
9	0.21470082075396\\
10	0.158145181676641\\
11	0.113855151098652\\
12	0.0839981116590052\\
13	0.0737421115055769\\
14	0.0340551159960104\\
15	0.023927992753932\\
};
\addlegendentry{NOW (c$\times$r)}

\end{semilogyaxis}

\begin{axis}[%
width=7cm,
height=3cm,
at={(0in,0in)},
scale only axis,
xmin=0,
xmax=1,
ymin=0,
ymax=1,
axis line style={draw=none},
ticks=none,
axis x line*=bottom,
axis y line*=left
]
\draw[-{stealth}, color=black] (axis cs:0.6,0.16) -- (axis cs:0.6,0.0);
\node[below right, align=left, draw=black,fill= white, opacity=0.75]
at (rel axis cs:0.043,0.45) {\footnotesize Number of blocks needed without\\\footnotesize UEP with MDS codes ($N_r:9$)};


\end{axis}
\end{tikzpicture}%

	    \caption{Normalized loss of the estimator using UEP codes with three classes as a function of number of received packets.}   \label{three_cls_mse_N}
	    \end{figure}

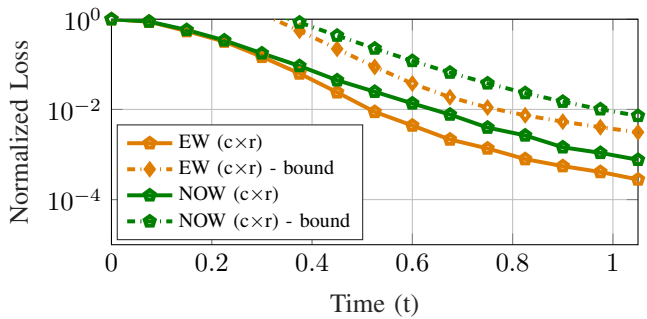
\begin{figure}[t]
	\centering
		\begin{tikzpicture}
\definecolor{mycolor1}{rgb}{0.63529,0.07843,0.18431}%
\definecolor{mycolor2}{rgb}{0.00000,0.44706,0.74118}%
\definecolor{mycolor3}{rgb}{0.00000,0.49804,0.00000}%
\definecolor{mycolor4}{rgb}{0.87059,0.49020,0.00000}%
\definecolor{mycolor5}{rgb}{0.00000,0.44700,0.74100}%
\definecolor{mycolor6}{rgb}{0.74902,0.00000,0.74902}%

\begin{semilogyaxis}[%
width=7cm,
height=3cm,
at={(0.88in,0.407in)},
scale only axis,
xmin=0,
xmax=1.05,
xlabel style={font=\color{white!15!black}},
xlabel={Time (t)},
ymode=log,
ymin=1e-05,
ymax=1,
yminorticks=true,
ylabel style={font=\color{white!15!black}},
ylabel={Normalized Loss},
axis background/.style={fill=white},
title style={font=\bfseries},
xmajorgrids,
ymajorgrids,
yminorgrids,
legend style={at={(0.01,0.02)}, font=\small,nodes={scale=0.85, transform shape},anchor=south west, legend cell align=left, align=left, draw=white!15!black}
]
\addplot [color=mycolor4, line width=1.5pt, mark=pentagon, mark options={solid, mycolor4}]
  table[row sep=crcr]{%
0	1\\
0.075	0.883013026122217\\
0.15	0.544530528350157\\
0.225	0.319186586736746\\
0.3	    0.146808210539521\\
0.375	0.0626712000301525\\
0.45	0.0239527380273664\\
0.525	0.0088202317308249\\
0.6	    0.00439414366846221\\
0.675	0.00215473427827106\\
0.75	0.00136287302892085\\
0.825	0.000790845554651305\\
0.9	    0.000552704409001079\\
0.975	0.000409559273277503\\
1.05	0.000279877306132689\\
};
\addlegendentry{EW (c$\times$r)}

\addplot [color=mycolor4, dashdotted, line width=1.5pt, mark=diamond, mark options={solid, mycolor4}]
  table[row sep=crcr]{%
0	9.00330476818414\\
0.075	7.96758391846098\\
0.15	5.26460984082261\\
0.225	2.82901327387968\\
0.3	1.31637561655316\\
0.375	0.554177942997673\\
0.45	0.220324432746925\\
0.525	0.0875034643674592\\
0.6	0.0374885397154752\\
0.675	0.0186679711141965\\
0.75	0.0110264885244957\\
0.825	0.00740643066599842\\
0.9	0.00535703235565857\\
0.975	0.00403027781488918\\
1.05	0.00310448018392768\\
};
\addlegendentry{EW (c$\times$r) - bound}

\addplot [color=mycolor3, line width=1.5pt, mark=pentagon, mark options={solid, mycolor3}]
  table[row sep=crcr]{%
0	1\\
0.075	0.888058009257436\\
0.15	0.578659808963654\\
0.225	0.340064618153266\\
0.3	0.175445246332409\\
0.375	0.0919559244305209\\
0.45	0.0443509426329577\\
0.525	0.0246848614967616\\
0.6	0.0135880936260154\\
0.675	0.00770532236610828\\
0.75	0.0039548963408007\\
0.825	0.0026518094931641\\
0.9	0.00144700690573497\\
0.975	0.00108908171138753\\
1.05	0.000764831929482263\\
};
\addlegendentry{NOW (c$\times$r)}

\addplot [color=mycolor3, dashdotted, line width=1.5pt, mark=pentagon, mark options={solid, mycolor3}]
  table[row sep=crcr]{%
0	9.00330476818414\\
0.075	7.97388117610188\\
0.15	5.36762469535579\\
0.225	3.08782857105968\\
0.3	1.64439604493896\\
0.375	0.846929941777174\\
0.45	0.433248957770562\\
0.525	0.224209950373777\\
0.6	0.119053665674311\\
0.675	0.0656181186917602\\
0.75	0.0378886391535354\\
0.825	0.0230651890886348\\
0.9	0.0148444260734098\\
0.975	0.0100895333470626\\
1.05	0.00721182160929603\\
};
\addlegendentry{NOW (c$\times$r) - bound}

\end{semilogyaxis}

\end{tikzpicture}%
	    {\caption{Normalized loss of the estimator using UEP codes with c$\times$r multiplication and its upper bound. \label{bound_cxr}} }
\end{figure}

%
%
Let us begin by evaluating the multiplication in the $\rtc$ and  $\ctr$ cases for a class of matrices satisfying Assumption \ref{asm_iid}.
The task is performed with the help of $W=30$ workers whose task completion times are modeled by exponential latency model with parameter $\lambda=1$. 
%
%
We select $N = P = 3$, $U = Q = 300$, and $H = 900$ for the $\rtc$ case, and $ U = Q = 900$, $H = 100$, and $M = 9$ for the $\ctr$ case to be fair in the comparison of two different multiplication schemes in terms of the computational load of the workers.

For $\rtc$ multiplication, as discussed earlier, we classify each row and column blocks of $\Av$ and $\Bv$ with importance levels \textit{high}, \textit{medium}, and \textit{low}. 
The elements of each block are i.i.d. and distributed with $\mathcal{N}(0,10)$, $\mathcal{N}(0,1)$, and $\mathcal{N}(0,0.1)$, for \textit{high}, \textit{medium}, and \textit{low} levels, respectively. %
We assume that both $\Av$ and $\Bv$ have only one instance of row and column from each level, i.e., $N = 3$, $P = 3$ with descending importance levels.
$\Av_1$ and $\Bv_1$ are from the \textit{high} importance level, $\Av_2$ and $\Bv_2$ are from the \textit{medium} importance level, $\Av_3$ and $\Bv_3$ are from the \textit{low} importance level. 
We take the multiplication of (\textit{high} and \textit{high}) and (\textit{high} and \textit{medium}) blocks as \textit{class one}, (\textit{medium} and \textit{medium}) and (\textit{high} and \textit{low}) blocks as \textit{class two}, and the remaining as \textit{class three}. 
With this definition, we have $(k_1, k_2,k_3) = (3,3,3)$ sub-blocks in each class.

For $\ctr$ multiplication, the importance levels are \textit{high}, \textit{medium}, and \textit{low} as in $\rtc$ case with same distributions. 
We assume that $\Av_i$ and $\Bv_i$ are from \textit{high} importance level if $i \in \{1,2,3\}$, from \textit{medium} importance level if $i \in \{4,5,6\}$, and from \textit{low} importance level if $i \in \{7,8,9\}$. 
The multiplication of (\textit{high} and \textit{high}) blocks are considered as \textit{class one}, (\textit{medium} and \textit{medium}) blocks as \textit{class two} and (\textit{low} and \textit{low}) blocks as \textit{class three}, resulting in three sub-blocks in each class.  


\begin{figure*}[t]
	\centering
	\input{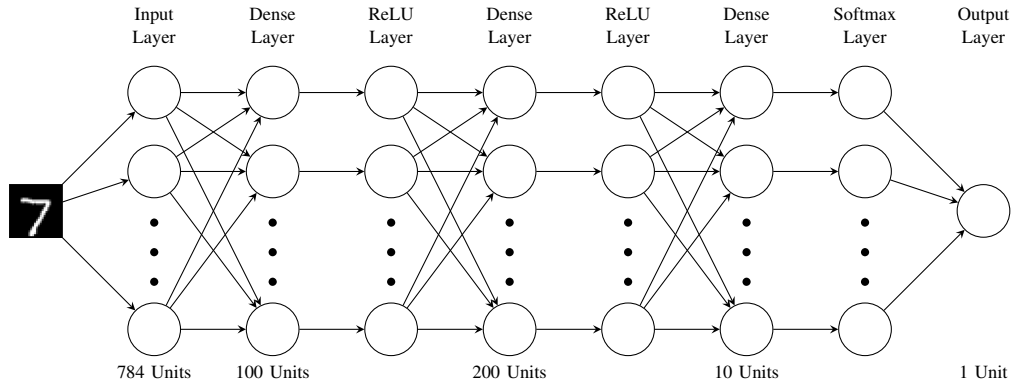}
	\caption{DNN used for classifying the MNIST dataset in Sec. \ref{sec:Back-propagation Matrices}.}  \label{fig:mnist_dnn_sketch}
\end{figure*}

\begin{table}[t]
\caption{UEP coding parameters for the matrix approximation used in Secs. \ref{sec:Numerical Examples} and \ref{sec:Back-propagation Matrices}.} \label{tab:code param}
\centering
\footnotesize
\begin{tabular}{|c|c|c|c|}
\hline
                               & Class 1 & Class 2 & Class 3 \\ \hline
\# of blocks                   & 3      & 3       & 3       \\ \hline
Window selection probs. & 0.40   & 0.35    & 0.25    \\ \hline
\end{tabular}
\end{table}

In our simulations, we select the window selection probabilities for both NOW and EW-UEP strategies as $(\Gamma_1, \Gamma_2, \Gamma_3) = (0.40, 0.35, 0.25)$ for both $\rtc$ and $\ctr$ multiplication as shown in Table \ref{tab:code param}. 
The decoding probabilities for each class are obtained through the formulation given in \cite{vukobratovic2012unequal}, as also depicted for the NOW-UEP and EW-UEP strategies in Fig. \ref{UEP_prob_fig}. 
As expected, the first class has higher decoding probability with both NOW and EW-UEP strategies since we choose a window selection function in which $\Gamma_1$ is greater than others which ensures better protection for the most important class.

In Fig. \ref{three_cls_mse}, these decoding probabilities are used to obtain the normalized expected loss values of $\rtc$ and $\ctr$ as a function of time $t$ along with the performance obtained with the Maximum Distance Separable (MDS) codes which are also used in \cite{lee2017speeding} for coded computation. 
With utilizing MDS codes, the perfect reconstruction is achieved when the number of received packets are equal to total number of sub-blocks. 
Until time $t = 0.44$, the UEP protection with NOW scheme performs better than that of MDS with both $\rtc$ and $\ctr$, since the UEP coding strategy enables early recovery of important classes with a small number of received packets. 
With the EW-UEP protection, we have higher protection for more important classes which gives a better approximation than the MDS coding until time $t = 0.825$ and $t = 0.975$, respectively, with $\rtc$ and $\ctr$ multiplication schemes. 
This also shows that $\ctr$ multiplication scheme is more efficient than $\rtc$ multiplication scheme.
In other words, if we are interested in an earlier recovery of certain important parts, using the UEP coding approach for matrix approximation is highly advantageous, especially with EW-UEP and $\ctr$ multiplication scheme. 
After time $t= 0.975$, the MDS code starts to perform better than  all the others since it can fully recover $\Cv$ after receiving nine packets. 
If we wait long enough, the UEP strategy will also fully recover the desired matrix product.

For further interpretation, we give the normalized loss values of matrix multiplication with MDS coding and approximate matrix multiplication using NOW and EW-UEP coding with both $\rtc$ and $\ctr$ in Fig. \ref{three_cls_mse_N} as a function of the number of received packets. 
The matrix multiplication with MDS codes needs to receive $\sum_{l} k_l$ packets to fully recover the result, where $k_l$ is the number of packets in the $l$-th level. 
Receiving less than $\sum_{l } k_l$ will not provide any partial information, and results in no recovery, hence the normalized loss with MDS coding is unity until it receives nine packets (the minimum required for recovery). 
However, matrix product approximation with NOW and EW-UEP coding strategies start to recover more important classes after receiving only very few packets, and continue to provide additional partial information after each received block. 
It is also important to note that the loss formulation given in Theorem \ref{thm:NOW-UEP} for $\rtc$ is an exact result that is achievable with large field sizes. Hence the simulation results given in both Figs. \ref{three_cls_mse} and \ref{three_cls_mse_N} for $\rtc$ case match our theoretical expectations.

Furthermore, in Fig. \ref{bound_cxr}, we compare the upper bound of the loss of $\ctr$ multiplication given in Theorem \ref{thm:NOW-UEP cXr} with simulation results using NOW and EW-UEP codes. 
As shown in the figure, the upper bound is not tight, however it can be used to illustrate the behavior of the loss curve. 
Even though the bounds are not tight, they successfully reflect the loss behavior of the proposed scheme. Thus, one can use these bounds to examine the difference between NOW-UEP and EW-UEP schemes or their performance with different corresponding window selection probabilities.
It is also worth noting that we choose the window selection distributions for the UEP codes arbitrarily. 
As a further improvement, this distribution can be optimized to minimize the expected loss.

\section{DNN Back-Propagation with Approximate Matrix Multiplication}
\label{sec:Back-propagation Matrices}

In this section, we consider an application of the approximate matrix multiplication approach in Sec.  \ref{sec:Approximate Matrix Multiplication with UEP Codes} to the training of DNNs on two classification datasets: the MNIST handwritten digits \cite{lecun2010mnist} and the CIFAR-10 images \cite{Krizhevsky09learningmultiple}. 
The next sub-section introduces the DNN settings, Sec. \ref{subsec:Gradient Sparsification} numerically validates the sparsity in gradient matrices, and Sec. \ref{subsec: DNN performance} presents the relevant training results.

\subsection{DNN Settings}
\label{subsec:DNN setting}

In this section, we consider the training for the MNIST and CIFAR-10 datasets. 
The MNIST dataset requires a quite simple DNN structure with three dense layers followed by a softmax layer, and training requires only a few epochs to attain a high classification accuracy.
 Unlike the MNIST, the CIFAR-10 dataset has more features and thus, requires both a more complex DNN architecture and more epochs for learning. 
 In this DNN model, we use two 2D Convolution Layers followed by a 2D Max Pooling. After the Max Pooling, the processed images pass through the fully connected layers, where we employ three dense layers

A conceptual representation of the network structure for the MNIST digit classification task is provided in Fig. \ref{fig:mnist_dnn_sketch}. The overall network structure for the CIFAR-10 classification task is defined in Table \ref{tab:cifar_model}. Further parameters utilized in the numerical evaluations are provided in Table \ref{tab:DNN_parameters}. 

\begin{table}[b]
	\footnotesize
	\centering
	\caption{Parameters and hyperparameters used for the training of the DNN models.}
	\label{tab:DNN_parameters}
	\begin{tabular}{|c|c|c|}
		\hline
		& MNIST & CIFAR-10 \\ \hline
        Training Samples & 60k & 50k \\ \hline
        Optimizer & \multicolumn{2}{|c|}{SGD} \\ \hline
        Learning Rate & \multicolumn{2}{|c|}{0.01} \\ \hline
        Loss & \multicolumn{2}{|c|}{Categorical Cross Entropy} \\ \hline
        Epochs & 3 & 120 \\ \hline
        Mini-Batch Sizes & \multicolumn{2}{|c|}{64} \\ \hline
	\end{tabular}
\end{table}

\begin{table}[b]
	\footnotesize
	\centering
	\caption{A summary of the model layers of the DNN used for CIFAR-10 classification.
	}
	\label{tab:cifar_model}
	\begin{tabular}{|c|c|c|c|}
		\hline
		Layer Name  & Kernel/Weights & Padding & Activation Layer \\ \hline
        Conv2D 1    & $3\times3\times32$ & Same & ReLU \\ \hline
        Conv2D 2    & $3\times3\times32$ & Valid & ReLU \\ \hline 
        MaxPooling2D  & $2\times2$ & - & - \\ \hline
        Flatten & - & - & - \\ \hline 
        Dense 1 & $7200\times512$ & - & ReLU \\ \hline
        Dense 2 & $512\times256$ & - & ReLU \\ \hline
        Dense 3 & $256\times12$ & - & Softmax \\ \hline
	\end{tabular}
\end{table}

	



Let us now denote the network as $D$ and the total number of layer as $I$. Let us further denote the weights, bias, inputs, outputs, and gradients of the corresponding layer as $\Vv_{i}$, $\bv_i$, $\Xv_{i}$, $\Ov_{i}$, and $\Gv_{i}$. From this we can define the output of each layer as
\ea{
    \Ov_{i} = D_{i}(\Xv_{i}),
    \tag{23}    
}
and the output of the network as
\ea{
    \Ov = D(\Xv),
    \tag{24}    
}
where $\Xv = \Xv_{1}$ is the initial input and $\Ov = \Ov_{I}$ is the final output. Note that $\Ov_{i} = \Xv_{i+1}$. Then the forward propagation of a dense layer can be defined as
\ea{ \label{eq:foward_prop_eq}
    D_{i}(\Xv_{i})  &= \Xv_{i} \Vv_{i} + \bv_{i}.
    \tag{25}    
}

However, for the back-propagation we must consider two matrix multiplications
\ea{ \label{eq:backprop_eq1}
    \Gv_{i} &=  \Gv_{i+1} \Vv_{i}^{T},
    \tag{26}    
}
and
\ea{ \label{eq:backprop_eq2}
    \Vv_{i}^{*} &=  \Xv_{i}^{T} \Gv_{i+1},
    \tag{27}    
}
where $\Vv_{i}^{*}$ is used for updating $\Vv_{i}$ and \eqref{eq:backprop_eq1} for calculating the $G_i$ used in $D_{i-1}$ for \eqref{eq:backprop_eq2}. $\Gv_{L}$ is calculated from the derivative of the loss between $\Ov_{L}$ and $\Yv$, the ground truth.

%
%


%


The dimensions of these matrices are defined in Table \ref{tab:backprop_dimensions} for both the MNIST and CIFAR-10 DNN models. The reason we focus on these matrices will be explained in Sec \ref{subsec:Gradient Sparsification}.

\begin{table}
	\footnotesize
	\centering
	\caption{A summary of the back-propagation matrices for the DNN models. For $^+$, see Remark \ref{remark:backprop_dense}.}
	\label{tab:backprop_dimensions}
	\begin{tabular}{|c|c|c|}
		\hline
		\multicolumn{3}{|c|}{MNIST} \\ \hline
		Layer       & $\Gv_{i}$ & $\Vv_{i}^{*}$            \\ \hline
		Dense 1     &  $(64\times100)\cdot(100\times784)$   & $(784\times64)\cdot(64\times100)$     \\ \hline
		Dense 2     & $(64\times200)\cdot(200\times100) $  &$(100\times64)\cdot(64\times200)$     \\ \hline
		Dense 3    & $(64\times10)\cdot(10\times200)$   & $(200\times64)\cdot(64\times10) $     \\ \hline
		
		\multicolumn{3}{|c|}{CIFAR-10} \\ \hline
		Layer       & $\Gv_{i}$ & $\Vv_{i}^{*} $            \\ \hline
		Conv 1/2    & + & + \\ \hline         
		Dense 1     &  $(64\times512)\cdot(512\times7200)$   & $(7200\times64)\cdot(64\times512)$     \\ \hline
		Dense 2     & $(64\times256)\cdot(256\times512) $  &$(512\times64)\cdot(64\times256)$     \\ \hline
		Dense 3    & $(64\times10)\cdot(10\times256)$   & $(256\times64)\cdot(64\times10) $     \\ \hline
	\end{tabular}
\end{table}

\begin{remark}
In our simulations we apply the proposed approximate matrix multiplication method only to the dense layer's back-propagation. %
An extension of the proposed method to the convolutional layers is left for future research. \label{remark:backprop_dense}
\end{remark}


\subsection{Sparsity of DNN Gradients}
\label{subsec:Gradient Sparsification}
Generally speaking, as the training of a DNN proceeds toward convergence, weights updates become increasingly smaller and thus gradients converge to zero.
In order to introduce resilience toward numerical errors and encourage sparsity in the trained model, gradient sparsification is often applied.
Broadly speaking, gradient sparsification is defined as the setting of some elements of the gradient updates to zero according to some policy.
%
%
For simplicity, let us assume that sparsification takes the form of a simple thresholding, that is 
%
\ea{ \label{eq:sparse_threshold}
    R (x)  &= 
    \begin{cases}
    x, & | x | > \tau \\
    0, &| x | \leq \tau,
    \end{cases}
    \tag{28}    
}
where $\tau \in [Z] $ is the threshold and $x$ is an element of matrix $\Av$/$\Bv$. 
%
The value of $\tau$ is chosen close to the machine precision at the beginning of training and is increased at each epoch. Also, $\tau$ is chosen differently for each layer: making the shallow layers less sparse than the deeper layers. 
As mentioned in  \cite{plancher2019application}, DNNs are fault-tolerant and resilient to these approximations in the gradient evaluations.

In applying the method in Sec. \ref{sec:Approximate Matrix Multiplication with UEP Codes}, we shall rely on the gradient matrix sparsity to define the protection level of each packet. 
%
%
Consider the DNN for the MNIST classification task described in Sec. \ref{subsec:DNN setting}, and use $\tau = 10^{-5}$  for sparcification of the gradient inputs, and $\tau = 10^{-4}$ for the weights and inputs of each layer.
%
This choice of $\tau$ at mini-batch iteration 389 / 937 of the first epoch and for each layer results in the sparsity level as given in Table \ref{tab:Gaussian_modeling_sparseness}. The empirical distribution of the remaining weights very much resembles that of a Gaussian distribution with zero mean and a variance which increases with the depth of the layer.
The Gaussian fitting of these non-sparse values are presented in Fig. \ref{fig:Gaussian_modeling}. 
It can be observed that indeed sparsity increases with the layer depth, as well as the variance of the fitted Gaussian.
This shows that indeed exploiting the variability in the matrix sub-block norms to apply unequal error protection has the potential of drastically improving the training performance.






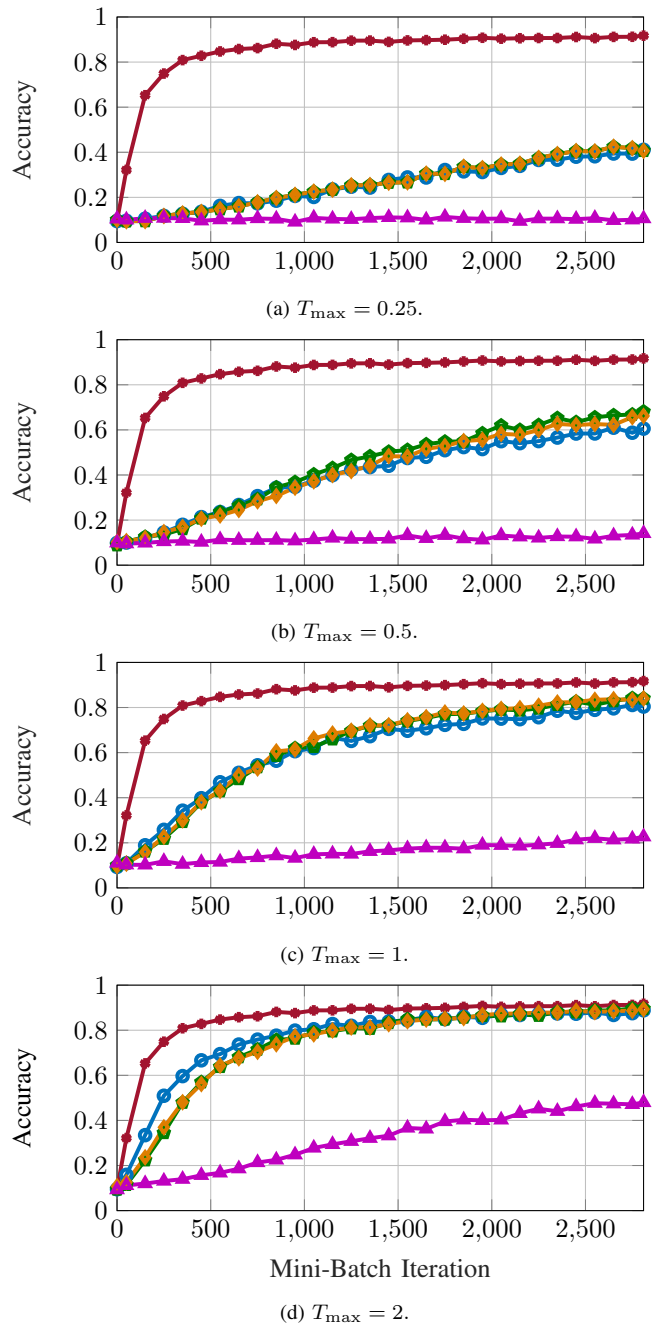
\begin{figure}[t]
	\centering
	\begin{subfigure}[b]{0.5\textwidth}
		\begin{tikzpicture}
\definecolor{mycolor1}{rgb}{0.63529,0.07843,0.18431}%
\definecolor{mycolor2}{rgb}{0.00000,0.44706,0.74118}%
\definecolor{mycolor3}{rgb}{0.00000,0.49804,0.00000}%
\definecolor{mycolor4}{rgb}{0.87059,0.49020,0.00000}%
\definecolor{mycolor5}{rgb}{0.00000,0.44700,0.74100}%
\definecolor{mycolor6}{rgb}{0.74902,0.00000,0.74902}%

\begin{axis}[%
width=7cm,
height=3cm,
scale only axis,
xmin=0,
xmax=2810,
xlabel style={font=\color{white!15!black}},
ymin=0,
ymax=1,
ylabel style={font=\color{white!15!black}},
ylabel={Accuracy},
axis background/.style={fill=white},
xmajorgrids,
ymajorgrids,
legend style={legend cell align=left, align=left, draw=white!15!black, nodes={scale=0.85, transform shape}, at={(0.01,0.65)}, anchor=west, fill opacity=0.8}
]

\addplot [color=mycolor1, line width=1.5pt, mark=asterisk, mark options={solid, mycolor1}]
  table[row sep=crcr]{%
0	0.0945312	\\
50	0.321562	\\
150	0.652969	\\
250	0.749062	\\
350	0.80875	\\
450	0.828125	\\
550	0.847187	\\
650	0.857812	\\
750	0.862031	\\
850	0.88125	\\
950	0.875938	\\
1050	0.887969	\\
1150	0.888281	\\
1250	0.895	\\
1350	0.895312	\\
1450	0.889531	\\
1550	0.896094	\\
1650	0.897344	\\
1750	0.899219	\\
1850	0.903438	\\
1950	0.907813	\\
2050	0.90375	\\
2150	0.905625	\\
2250	0.906719	\\
2350	0.906719	\\
2450	0.91125	\\
2550	0.90625	\\
2650	0.911719	\\
2750	0.912188	\\
2810	0.917188	\\
};

\addplot [color=mycolor2, line width=1.5pt, mark=o, mark options={solid, mycolor2}]
  table[row sep=crcr]{%
0	0.095	\\
50	0.0959375	\\
150	0.105313	\\
250	0.117813	\\
350	0.127188	\\
450	0.135937	\\
550	0.162188	\\
650	0.174687	\\
750	0.173437	\\
850	0.185312	\\
950	0.206563	\\
1050	0.201875	\\
1150	0.235313	\\
1250	0.247188	\\
1350	0.243438	\\
1450	0.277187	\\
1550	0.288438	\\
1650	0.288125	\\
1750	0.32125	\\
1850	0.315312	\\
1950	0.3125	\\
2050	0.330625	\\
2150	0.339687	\\
2250	0.365938	\\
2350	0.367188	\\
2450	0.380312	\\
2550	0.382188	\\
2650	0.394375	\\
2750	0.394062	\\
2810	0.410938	\\
};

\addplot [color=mycolor3, line width=1.5pt, mark=pentagon, mark options={solid, mycolor3}]
  table[row sep=crcr]{%
0	0.100312	\\
50	0.094375	\\
150	0.095	\\
250	0.1175	\\
350	0.128438	\\
450	0.135	\\
550	0.147813	\\
650	0.160312	\\
750	0.177813	\\
850	0.196875	\\
950	0.210938	\\
1050	0.225312	\\
1150	0.234375	\\
1250	0.253125	\\
1350	0.254375	\\
1450	0.265313	\\
1550	0.266875	\\
1650	0.305	\\
1750	0.304375	\\
1850	0.335625	\\
1950	0.329375	\\
2050	0.345313	\\
2150	0.349062	\\
2250	0.374688	\\
2350	0.391875	\\
2450	0.405313	\\
2550	0.40375	\\
2650	0.424375	\\
2750	0.418125	\\
2810	0.406562	\\
};

\addplot [color=mycolor4, line width=1.5pt, mark=diamond, mark options={solid, mycolor4}]
  table[row sep=crcr]{%
0	0.100312	\\
50	0.094375	\\
150	0.095	\\
250	0.1175	\\
350	0.128438	\\
450	0.135	\\
550	0.147813	\\
650	0.160312	\\
750	0.177813	\\
850	0.196875	\\
950	0.210938	\\
1050	0.225312	\\
1150	0.234375	\\
1250	0.253125	\\
1350	0.254375	\\
1450	0.265313	\\
1550	0.266875	\\
1650	0.305	\\
1750	0.304375	\\
1850	0.335625	\\
1950	0.329375	\\
2050	0.345313	\\
2150	0.349062	\\
2250	0.374688	\\
2350	0.391875	\\
2450	0.405313	\\
2550	0.40375	\\
2650	0.424375	\\
2750	0.418125	\\
2810	0.406562	\\
};

\addplot [color=mycolor6, line width=1.5pt, mark=triangle, mark options={solid, mycolor6}]
  table[row sep=crcr]{%
0	0.104688	\\
50	0.100312	\\
150	0.1075	\\
250	0.108125	\\
350	0.106875	\\
450	0.09625	\\
550	0.10125	\\
650	0.100625	\\
750	0.10625	\\
850	0.104375	\\
950	0.0909375	\\
1050	0.109063	\\
1150	0.104375	\\
1250	0.102813	\\
1350	0.108125	\\
1450	0.11125	\\
1550	0.109687	\\
1650	0.100312	\\
1750	0.113125	\\
1850	0.107188	\\
1950	0.105	\\
2050	0.105938	\\
2150	0.094375	\\
2250	0.105625	\\
2350	0.105625	\\
2450	0.103438	\\
2550	0.1075	\\
2650	0.098125	\\
2750	0.100625	\\
2810	0.105938	\\
};

\end{axis}

\begin{axis}[%
width=7cm,
height=3cm,
at={(0in,0in)},
scale only axis,
xmin=0,
xmax=1,
ymin=0,
ymax=1,
axis line style={draw=none},
ticks=none,
axis x line*=bottom,
axis y line*=left
]

\end{axis}
\end{tikzpicture}%
        \captionsetup{justification=centering}
		\caption{$T_{\max} = 0.25.$} \label{fig:MNIST_rc_a}
	\end{subfigure}
	
	\begin{subfigure}[b]{0.5\textwidth}
		\begin{tikzpicture}
\definecolor{mycolor1}{rgb}{0.63529,0.07843,0.18431}%
\definecolor{mycolor2}{rgb}{0.00000,0.44706,0.74118}%
\definecolor{mycolor3}{rgb}{0.00000,0.49804,0.00000}%
\definecolor{mycolor4}{rgb}{0.87059,0.49020,0.00000}%
\definecolor{mycolor5}{rgb}{0.00000,0.44700,0.74100}%
\definecolor{mycolor6}{rgb}{0.74902,0.00000,0.74902}%

\begin{axis}[%
width=7cm,
height=3cm,
scale only axis,
xmin=0,
xmax=2810,
xlabel style={font=\color{white!15!black}},
ymin=0,
ymax=1,
ylabel style={font=\color{white!15!black}},
ylabel={Accuracy},
axis background/.style={fill=white},
xmajorgrids,
ymajorgrids,
legend style={legend cell align=left, align=left, draw=white!15!black, nodes={scale=0.85, transform shape}, at={(0.01,0.65)}, anchor=west, fill opacity=0.8}
]

\addplot [color=mycolor1, line width=1.5pt, mark=asterisk, mark options={solid, mycolor1}]
  table[row sep=crcr]{%
0	0.0945312	\\
50	0.321562	\\
150	0.652969	\\
250	0.749062	\\
350	0.80875	\\
450	0.828125	\\
550	0.847187	\\
650	0.857812	\\
750	0.862031	\\
850	0.88125	\\
950	0.875938	\\
1050	0.887969	\\
1150	0.888281	\\
1250	0.895	\\
1350	0.895312	\\
1450	0.889531	\\
1550	0.896094	\\
1650	0.897344	\\
1750	0.899219	\\
1850	0.903438	\\
1950	0.907813	\\
2050	0.90375	\\
2150	0.905625	\\
2250	0.906719	\\
2350	0.906719	\\
2450	0.91125	\\
2550	0.90625	\\
2650	0.911719	\\
2750	0.912188	\\
2810	0.917188	\\
};

\addplot [color=mycolor2, line width=1.5pt, mark=o, mark options={solid, mycolor2}]
  table[row sep=crcr]{%
0	0.09875	\\
50	0.100312	\\
150	0.120625	\\
250	0.145	\\
350	0.179063	\\
450	0.213125	\\
550	0.235625	\\
650	0.266875	\\
750	0.30625	\\
850	0.343125	\\
950	0.3475	\\
1050	0.374063	\\
1150	0.399375	\\
1250	0.42375	\\
1350	0.435	\\
1450	0.440937	\\
1550	0.475313	\\
1650	0.480625	\\
1750	0.51	\\
1850	0.52375	\\
1950	0.515	\\
2050	0.55	\\
2150	0.54125	\\
2250	0.550312	\\
2350	0.565937	\\
2450	0.58375	\\
2550	0.584063	\\
2650	0.609688	\\
2750	0.5875	\\
2810	0.605625	\\
};

\addplot [color=mycolor3, line width=1.5pt, mark=pentagon, mark options={solid, mycolor3}]
  table[row sep=crcr]{%
0	0.08875	\\
50	0.104375	\\
150	0.125312	\\
250	0.135625	\\
350	0.160625	\\
450	0.2075	\\
550	0.235	\\
650	0.262813	\\
750	0.292812	\\
850	0.344687	\\
950	0.368437	\\
1050	0.4025	\\
1150	0.430625	\\
1250	0.46625	\\
1350	0.482812	\\
1450	0.502812	\\
1550	0.510938	\\
1650	0.535937	\\
1750	0.547813	\\
1850	0.549063	\\
1950	0.584688	\\
2050	0.619375	\\
2150	0.599688	\\
2250	0.62125	\\
2350	0.652188	\\
2450	0.634375	\\
2550	0.655312	\\
2650	0.664687	\\
2750	0.668125	\\
2810	0.68125	\\
};

\addplot [color=mycolor4, line width=1.5pt, mark=diamond, mark options={solid, mycolor4}]
  table[row sep=crcr]{%
0	0.103125	\\
50	0.106563	\\
150	0.119688	\\
250	0.146563	\\
350	0.169375	\\
450	0.20625	\\
550	0.220625	\\
650	0.245	\\
750	0.283438	\\
850	0.307188	\\
950	0.34375	\\
1050	0.372188	\\
1150	0.4	\\
1250	0.4175	\\
1350	0.440937	\\
1450	0.48375	\\
1550	0.481563	\\
1650	0.515938	\\
1750	0.528125	\\
1850	0.550312	\\
1950	0.55625	\\
2050	0.582187	\\
2150	0.57875	\\
2250	0.59875	\\
2350	0.626563	\\
2450	0.621563	\\
2550	0.625938	\\
2650	0.622188	\\
2750	0.656563	\\
2810	0.661875	\\
};

\addplot [color=mycolor6, line width=1.5pt, mark=triangle, mark options={solid, mycolor6}]
  table[row sep=crcr]{%
0	0.099375	\\
50	0.0953125	\\
150	0.09875	\\
250	0.105	\\
350	0.106875	\\
450	0.1025	\\
550	0.113437	\\
650	0.110625	\\
750	0.110625	\\
850	0.11125	\\
950	0.108438	\\
1050	0.11375	\\
1150	0.120625	\\
1250	0.115625	\\
1350	0.115625	\\
1450	0.116562	\\
1550	0.131875	\\
1650	0.119063	\\
1750	0.1325	\\
1850	0.117813	\\
1950	0.1125	\\
2050	0.131562	\\
2150	0.125938	\\
2250	0.120938	\\
2350	0.127188	\\
2450	0.12625	\\
2550	0.116562	\\
2650	0.130312	\\
2750	0.134063	\\
2810	0.140937	\\
};

\end{axis}

\begin{axis}[%
width=7cm,
height=3cm,
at={(0in,0in)},
scale only axis,
xmin=0,
xmax=1,
ymin=0,
ymax=1,
axis line style={draw=none},
ticks=none,
axis x line*=bottom,
axis y line*=left
]

\end{axis}
\end{tikzpicture}%
        \captionsetup{justification=centering}
		\caption{$T_{\max} = 0.5.$} \label{fig:MNIST_rc_b}
	\end{subfigure}
	
	\begin{subfigure}[b]{0.5\textwidth}
		\begin{tikzpicture}
\definecolor{mycolor1}{rgb}{0.63529,0.07843,0.18431}%
\definecolor{mycolor2}{rgb}{0.00000,0.44706,0.74118}%
\definecolor{mycolor3}{rgb}{0.00000,0.49804,0.00000}%
\definecolor{mycolor4}{rgb}{0.87059,0.49020,0.00000}%
\definecolor{mycolor5}{rgb}{0.00000,0.44700,0.74100}%
\definecolor{mycolor6}{rgb}{0.74902,0.00000,0.74902}%

\begin{axis}[%
width=7cm,
height=3cm,
scale only axis,
xmin=0,
xmax=2810,
xlabel style={font=\color{white!15!black}},
ymin=0,
ymax=1,
ylabel style={font=\color{white!15!black}},
ylabel={Accuracy},
axis background/.style={fill=white},
xmajorgrids,
ymajorgrids,
]

\addplot [color=mycolor1, line width=1.5pt, mark=asterisk, mark options={solid, mycolor1}]
  table[row sep=crcr]{%
0	0.0945312	\\
50	0.321562	\\
150	0.652969	\\
250	0.749062	\\
350	0.80875	\\
450	0.828125	\\
550	0.847187	\\
650	0.857812	\\
750	0.862031	\\
850	0.88125	\\
950	0.875938	\\
1050	0.887969	\\
1150	0.888281	\\
1250	0.895	\\
1350	0.895312	\\
1450	0.889531	\\
1550	0.896094	\\
1650	0.897344	\\
1750	0.899219	\\
1850	0.903438	\\
1950	0.907813	\\
2050	0.90375	\\
2150	0.905625	\\
2250	0.906719	\\
2350	0.906719	\\
2450	0.91125	\\
2550	0.90625	\\
2650	0.911719	\\
2750	0.912188	\\
2810	0.917188	\\
};

\addplot [color=mycolor2, line width=1.5pt, mark=o, mark options={solid, mycolor2}]
  table[row sep=crcr]{%
0	0.093125	\\
50	0.109375	\\
150	0.189062	\\
250	0.2575	\\
350	0.3425	\\
450	0.397813	\\
550	0.468438	\\
650	0.510312	\\
750	0.54375	\\
850	0.564063	\\
950	0.60625	\\
1050	0.62	\\
1150	0.665625	\\
1250	0.651875	\\
1350	0.672813	\\
1450	0.705313	\\
1550	0.696875	\\
1650	0.707812	\\
1750	0.723125	\\
1850	0.7275	\\
1950	0.751563	\\
2050	0.750625	\\
2150	0.748125	\\
2250	0.756875	\\
2350	0.785625	\\
2450	0.775937	\\
2550	0.789687	\\
2650	0.79625	\\
2750	0.81125	\\
2810	0.804375	\\
};

\addplot [color=mycolor3, line width=1.5pt, mark=pentagon, mark options={solid, mycolor3}]
  table[row sep=crcr]{%
0	0.1	\\
50	0.11	\\
150	0.161875	\\
250	0.218438	\\
350	0.293125	\\
450	0.379688	\\
550	0.429375	\\
650	0.483125	\\
750	0.535937	\\
850	0.58625	\\
950	0.619687	\\
1050	0.625938	\\
1150	0.659375	\\
1250	0.695625	\\
1350	0.71875	\\
1450	0.72	\\
1550	0.74	\\
1650	0.75375	\\
1750	0.770625	\\
1850	0.771875	\\
1950	0.783125	\\
2050	0.79125	\\
2150	0.7875	\\
2250	0.795937	\\
2350	0.812187	\\
2450	0.825937	\\
2550	0.8125	\\
2650	0.82875	\\
2750	0.84125	\\
2810	0.840313	\\
};

\addplot [color=mycolor4, line width=1.5pt, mark=diamond, mark options={solid, mycolor4}]
  table[row sep=crcr]{%
0	0.096875	\\
50	0.1075	\\
150	0.158438	\\
250	0.226875	\\
350	0.3	\\
450	0.380312	\\
550	0.430625	\\
650	0.500625	\\
750	0.529687	\\
850	0.605313	\\
950	0.61375	\\
1050	0.662813	\\
1150	0.684375	\\
1250	0.696875	\\
1350	0.720625	\\
1450	0.720938	\\
1550	0.743125	\\
1650	0.756875	\\
1750	0.777813	\\
1850	0.773438	\\
1950	0.785	\\
2050	0.79375	\\
2150	0.797188	\\
2250	0.807813	\\
2350	0.822187	\\
2450	0.824375	\\
2550	0.834063	\\
2650	0.836875	\\
2750	0.833438	\\
2810	0.84	\\
};

\addplot [color=mycolor6, line width=1.5pt, mark=triangle, mark options={solid, mycolor6}]
  table[row sep=crcr]{%
0	0.110625	\\
50	0.1025	\\
150	0.102188	\\
250	0.117813	\\
350	0.105313	\\
450	0.113125	\\
550	0.114687	\\
650	0.13	\\
750	0.135	\\
850	0.143125	\\
950	0.132188	\\
1050	0.149062	\\
1150	0.150938	\\
1250	0.150312	\\
1350	0.162188	\\
1450	0.166563	\\
1550	0.174063	\\
1650	0.178125	\\
1750	0.177813	\\
1850	0.17375	\\
1950	0.19	\\
2050	0.189062	\\
2150	0.18625	\\
2250	0.191562	\\
2350	0.197813	\\
2450	0.21375	\\
2550	0.22	\\
2650	0.213438	\\
2750	0.217188	\\
2810	0.226875	\\
};

\end{axis}

\begin{axis}[%
width=7cm,
height=3cm,
at={(0in,0in)},
scale only axis,
xmin=0,
xmax=1,
ymin=0,
ymax=1,
axis line style={draw=none},
ticks=none,
axis x line*=bottom,
axis y line*=left
]

\end{axis}
\end{tikzpicture}%
        \captionsetup{justification=centering}
		\caption{$T_{\max} = 1.$} \label{fig:MNIST_rc_c}
	\end{subfigure}
	
	\begin{subfigure}[b]{0.5\textwidth}
		\begin{tikzpicture}
\definecolor{mycolor1}{rgb}{0.63529,0.07843,0.18431}%
\definecolor{mycolor2}{rgb}{0.00000,0.44706,0.74118}%
\definecolor{mycolor3}{rgb}{0.00000,0.49804,0.00000}%
\definecolor{mycolor4}{rgb}{0.87059,0.49020,0.00000}%
\definecolor{mycolor5}{rgb}{0.00000,0.44700,0.74100}%
\definecolor{mycolor6}{rgb}{0.74902,0.00000,0.74902}%

\begin{axis}[%
width=7cm,
height=3cm,
scale only axis,
xmin=0,
xmax=2810,
xlabel style={font=\color{white!15!black}},
xlabel={Mini-Batch Iteration},
ymin=0,
ymax=1,
 ylabel={Accuracy},
axis background/.style={fill=white},
xmajorgrids,
ymajorgrids,
legend style={legend cell align=left, align=left, draw=white!15!black, nodes={scale=0.85, transform shape}, at={(0.01,0.65)}, anchor=west, fill opacity=0.8}
]

\addplot [color=mycolor1, line width=1.5pt, mark=asterisk, mark options={solid, mycolor1}]
  table[row sep=crcr]{%
0	0.0945312	\\
50	0.321562	\\
150	0.652969	\\
250	0.749062	\\
350	0.80875	\\
450	0.828125	\\
550	0.847187	\\
650	0.857812	\\
750	0.862031	\\
850	0.88125	\\
950	0.875938	\\
1050	0.887969	\\
1150	0.888281	\\
1250	0.895	\\
1350	0.895312	\\
1450	0.889531	\\
1550	0.896094	\\
1650	0.897344	\\
1750	0.899219	\\
1850	0.903438	\\
1950	0.907813	\\
2050	0.90375	\\
2150	0.905625	\\
2250	0.906719	\\
2350	0.906719	\\
2450	0.91125	\\
2550	0.90625	\\
2650	0.911719	\\
2750	0.912188	\\
2810	0.917188	\\
};

\addplot [color=mycolor2, line width=1.5pt, mark=o, mark options={solid, mycolor2}]
  table[row sep=crcr]{%
0	0.095	\\
50	0.15875	\\
150	0.334688	\\
250	0.509375	\\
350	0.596562	\\
450	0.666562	\\
550	0.694063	\\
650	0.736563	\\
750	0.758437	\\
850	0.775937	\\
950	0.798438	\\
1050	0.804375	\\
1150	0.8275	\\
1250	0.819688	\\
1350	0.836875	\\
1450	0.838437	\\
1550	0.844063	\\
1650	0.858437	\\
1750	0.849063	\\
1850	0.859375	\\
1950	0.855	\\
2050	0.869687	\\
2150	0.8675	\\
2250	0.871563	\\
2350	0.874375	\\
2450	0.873125	\\
2550	0.87875	\\
2650	0.867812	\\
2750	0.875625	\\
2810	0.888437	\\
};

\addplot [color=mycolor3, line width=1.5pt, mark=pentagon, mark options={solid, mycolor3}]
  table[row sep=crcr]{%
0	0.0940625	\\
50	0.114687	\\
150	0.22375	\\
250	0.344687	\\
350	0.479063	\\
450	0.569688	\\
550	0.637813	\\
650	0.680937	\\
750	0.714375	\\
850	0.756563	\\
950	0.763437	\\
1050	0.790312	\\
1150	0.795625	\\
1250	0.809063	\\
1350	0.806875	\\
1450	0.8325	\\
1550	0.84875	\\
1650	0.847812	\\
1750	0.848437	\\
1850	0.856875	\\
1950	0.864375	\\
2050	0.862187	\\
2150	0.870313	\\
2250	0.864688	\\
2350	0.880625	\\
2450	0.880313	\\
2550	0.887188	\\
2650	0.888125	\\
2750	0.89625	\\
2810	0.895312	\\
};

\addplot [color=mycolor4, line width=1.5pt, mark=diamond, mark options={solid, mycolor4}]
  table[row sep=crcr]{%
0	0.108125	\\
50	0.125625	\\
150	0.23625	\\
250	0.369375	\\
350	0.480625	\\
450	0.56125	\\
550	0.643125	\\
650	0.675625	\\
750	0.704375	\\
850	0.737812	\\
950	0.7725	\\
1050	0.781875	\\
1150	0.799063	\\
1250	0.8125	\\
1350	0.8125	\\
1450	0.827187	\\
1550	0.840313	\\
1650	0.8475	\\
1750	0.85125	\\
1850	0.854062	\\
1950	0.866875	\\
2050	0.87375	\\
2150	0.87375	\\
2250	0.878437	\\
2350	0.878125	\\
2450	0.885938	\\
2550	0.880625	\\
2650	0.885	\\
2750	0.886563	\\
2810	0.893125	\\
};

\addplot [color=mycolor6, line width=1.5pt, mark=triangle, mark options={solid, mycolor6}]
  table[row sep=crcr]{%
0	0.093125	\\
50	0.111875	\\
150	0.12	\\
250	0.13125	\\
350	0.139375	\\
450	0.156562	\\
550	0.1675	\\
650	0.185938	\\
750	0.213438	\\
850	0.225	\\
950	0.247812	\\
1050	0.277813	\\
1150	0.29375	\\
1250	0.3075	\\
1350	0.320937	\\
1450	0.332188	\\
1550	0.366563	\\
1650	0.3625	\\
1750	0.395313	\\
1850	0.402813	\\
1950	0.400313	\\
2050	0.4025	\\
2150	0.43125	\\
2250	0.451562	\\
2350	0.440625	\\
2450	0.461562	\\
2550	0.47625	\\
2650	0.474062	\\
2750	0.47125	\\
2810	0.479687	\\
};

\end{axis}

\begin{axis}[%
width=7cm,
height=3cm,
at={(0in,0in)},
scale only axis,
xmin=0,
xmax=1,
ymin=0,
ymax=1,
axis line style={draw=none},
ticks=none,
axis x line*=bottom,
axis y line*=left
]

\end{axis}
\end{tikzpicture}%
        \captionsetup{justification=centering}
		\caption{$T_{\max} = 2.$} \label{fig:MNIST_rc_d}
	\end{subfigure}
	\caption{MNIST classification accuracy using  row-times-column matrix multiplication for the DNN proposed in Sec. \ref{subsec:DNN setting}. }\label{fig:MNIST_rc}
\end{figure}

\begin{figure}[t]
	\begin{subfigure}[b]{0.5\textwidth}
		\begin{tikzpicture}
\definecolor{mycolor1}{rgb}{0.63529,0.07843,0.18431}%
\definecolor{mycolor2}{rgb}{0.00000,0.44706,0.74118}%
\definecolor{mycolor3}{rgb}{0.00000,0.49804,0.00000}%
\definecolor{mycolor4}{rgb}{0.87059,0.49020,0.00000}%
\definecolor{mycolor5}{rgb}{0.00000,0.44700,0.74100}%
\definecolor{mycolor6}{rgb}{0.74902,0.00000,0.74902}%

\begin{axis}[%
width=7cm,
height=3cm,
scale only axis,
xmin=0,
xmax=2810,
xlabel style={font=\color{white!15!black}},
ymin=0,
ymax=1,
ylabel style={font=\color{white!15!black}},
ylabel={Accuracy},
axis background/.style={fill=white},
xmajorgrids,
ymajorgrids,
legend style={legend cell align=left, align=left, draw=white!15!black, nodes={scale=0.85, transform shape}, at={(0.01,0.65)}, anchor=west, fill opacity=0.8}
]

\addplot [color=mycolor1, line width=1.5pt, mark=asterisk, mark options={solid, mycolor1}]
  table[row sep=crcr]{%
0	0.0945312	\\
50	0.321562	\\
150	0.652969	\\
250	0.749062	\\
350	0.80875	\\
450	0.828125	\\
550	0.847187	\\
650	0.857812	\\
750	0.862031	\\
850	0.88125	\\
950	0.875938	\\
1050	0.887969	\\
1150	0.888281	\\
1250	0.895	\\
1350	0.895312	\\
1450	0.889531	\\
1550	0.896094	\\
1650	0.897344	\\
1750	0.899219	\\
1850	0.903438	\\
1950	0.907813	\\
2050	0.90375	\\
2150	0.905625	\\
2250	0.906719	\\
2350	0.906719	\\
2450	0.91125	\\
2550	0.90625	\\
2650	0.911719	\\
2750	0.912188	\\
2810	0.917188	\\
};

\addplot [color=mycolor2, line width=1.5pt, dashed, mark=o, mark options={solid, mycolor2}]
  table[row sep=crcr]{%
0	0.1025	\\
50	0.111875	\\
150	0.1125	\\
250	0.13125	\\
350	0.125312	\\
450	0.145937	\\
550	0.158125	\\
650	0.164687	\\
750	0.16875	\\
850	0.181875	\\
950	0.186875	\\
1050	0.206875	\\
1150	0.209687	\\
1250	0.2325	\\
1350	0.252812	\\
1450	0.25125	\\
1550	0.26	\\
1650	0.28625	\\
1750	0.300938	\\
1850	0.299375	\\
1950	0.3175	\\
2050	0.32	\\
2150	0.335625	\\
2250	0.348438	\\
2350	0.355312	\\
2450	0.367812	\\
2550	0.368125	\\
2650	0.381875	\\
2750	0.374063	\\
2810	0.405313	\\
};

\addplot [color=mycolor3, line width=1.5pt, dashed, mark=pentagon, mark options={solid, mycolor3}]
  table[row sep=crcr]{%
0	0.089375	\\
50	0.103438	\\
150	0.133437	\\
250	0.140937	\\
350	0.155625	\\
450	0.185312	\\
550	0.2125	\\
650	0.207813	\\
750	0.220625	\\
850	0.245938	\\
950	0.259062	\\
1050	0.289375	\\
1150	0.285	\\
1250	0.312188	\\
1350	0.334375	\\
1450	0.3525	\\
1550	0.34875	\\
1650	0.37875	\\
1750	0.405	\\
1850	0.40375	\\
1950	0.415312	\\
2050	0.46125	\\
2150	0.449688	\\
2250	0.465938	\\
2350	0.48875	\\
2450	0.498125	\\
2550	0.505313	\\
2650	0.536563	\\
2750	0.559063	\\
2810	0.552188	\\
};

\addplot [color=mycolor4, line width=1.5pt, dashed, mark=diamond, mark options={solid, mycolor4}]
  table[row sep=crcr]{%
0	0.102188	\\
50	0.10125	\\
150	0.132188	\\
250	0.129062	\\
350	0.153125	\\
450	0.178125	\\
550	0.192188	\\
650	0.211875	\\
750	0.226875	\\
850	0.23875	\\
950	0.247188	\\
1050	0.27	\\
1150	0.289687	\\
1250	0.287187	\\
1350	0.319688	\\
1450	0.346562	\\
1550	0.35625	\\
1650	0.352812	\\
1750	0.392813	\\
1850	0.4025	\\
1950	0.434688	\\
2050	0.43	\\
2150	0.45	\\
2250	0.462813	\\
2350	0.465938	\\
2450	0.500313	\\
2550	0.50375	\\
2650	0.519687	\\
2750	0.528125	\\
2810	0.534375	\\
};

\addplot [color=mycolor6, line width=1.5pt, dashed, mark=triangle, mark options={solid, mycolor6}]
  table[row sep=crcr]{%
0	0.105625	\\
50	0.106563	\\
150	0.115625	\\
250	0.107813	\\
350	0.105313	\\
450	0.1125	\\
550	0.118125	\\
650	0.105625	\\
750	0.109687	\\
850	0.11625	\\
950	0.115312	\\
1050	0.119063	\\
1150	0.122813	\\
1250	0.128125	\\
1350	0.115	\\
1450	0.12625	\\
1550	0.11125	\\
1650	0.119063	\\
1750	0.113125	\\
1850	0.12375	\\
1950	0.121875	\\
2050	0.115937	\\
2150	0.123125	\\
2250	0.115937	\\
2350	0.113437	\\
2450	0.12875	\\
2550	0.125312	\\
2650	0.116875	\\
2750	0.135625	\\
2810	0.118438	\\
};

\end{axis}

\begin{axis}[%
width=7cm,
height=3cm,
at={(0in,0in)},
scale only axis,
xmin=0,
xmax=1,
ymin=0,
ymax=1,
axis line style={draw=none},
ticks=none,
axis x line*=bottom,
axis y line*=left
]

\end{axis}
\end{tikzpicture}%
        \captionsetup{justification=centering}
		\caption{$T_{\max} = 0.25.$} \label{fig:MNIST_cr_a}
	\end{subfigure}
	
	\begin{subfigure}[b]{0.5\textwidth}
		\begin{tikzpicture}
\definecolor{mycolor1}{rgb}{0.63529,0.07843,0.18431}%
\definecolor{mycolor2}{rgb}{0.00000,0.44706,0.74118}%
\definecolor{mycolor3}{rgb}{0.00000,0.49804,0.00000}%
\definecolor{mycolor4}{rgb}{0.87059,0.49020,0.00000}%
\definecolor{mycolor5}{rgb}{0.00000,0.44700,0.74100}%
\definecolor{mycolor6}{rgb}{0.74902,0.00000,0.74902}%

\begin{axis}[%
width=7cm,
height=3cm,
scale only axis,
xmin=0,
xmax=2810,
xlabel style={font=\color{white!15!black}},
ymin=0,
ymax=1,
ylabel style={font=\color{white!15!black}},
ylabel={Accuracy},
axis background/.style={fill=white},
xmajorgrids,
ymajorgrids,
legend style={legend cell align=left, align=left, draw=white!15!black, nodes={scale=0.85, transform shape}, at={(0.01,0.65)}, anchor=west, fill opacity=0.8}
]

\addplot [color=mycolor1, line width=1.5pt, mark=asterisk, mark options={solid, mycolor1}]
  table[row sep=crcr]{%
0	0.0945312	\\
50	0.321562	\\
150	0.652969	\\
250	0.749062	\\
350	0.80875	\\
450	0.828125	\\
550	0.847187	\\
650	0.857812	\\
750	0.862031	\\
850	0.88125	\\
950	0.875938	\\
1050	0.887969	\\
1150	0.888281	\\
1250	0.895	\\
1350	0.895312	\\
1450	0.889531	\\
1550	0.896094	\\
1650	0.897344	\\
1750	0.899219	\\
1850	0.903438	\\
1950	0.907813	\\
2050	0.90375	\\
2150	0.905625	\\
2250	0.906719	\\
2350	0.906719	\\
2450	0.91125	\\
2550	0.90625	\\
2650	0.911719	\\
2750	0.912188	\\
2810	0.917188	\\
};

\addplot [color=mycolor2, line width=1.5pt, dashed, mark=o, mark options={solid, mycolor2}]
  table[row sep=crcr]{%
0	0.105625	\\
50	0.113125	\\
150	0.141563	\\
250	0.142813	\\
350	0.17625	\\
450	0.194062	\\
550	0.224375	\\
650	0.26	\\
750	0.295938	\\
850	0.325313	\\
950	0.360312	\\
1050	0.366563	\\
1150	0.385937	\\
1250	0.430625	\\
1350	0.422187	\\
1450	0.468125	\\
1550	0.472813	\\
1650	0.489375	\\
1750	0.492188	\\
1850	0.504687	\\
1950	0.504375	\\
2050	0.52875	\\
2150	0.537813	\\
2250	0.553438	\\
2350	0.560312	\\
2450	0.57625	\\
2550	0.56875	\\
2650	0.595625	\\
2750	0.595938	\\
2810	0.593125	\\
};

\addplot [color=mycolor3, line width=1.5pt, dashed, mark=pentagon, mark options={solid, mycolor3}]
  table[row sep=crcr]{%
0	0.104063	\\
50	0.119375	\\
150	0.160312	\\
250	0.194062	\\
350	0.260625	\\
450	0.285	\\
550	0.329375	\\
650	0.375625	\\
750	0.40375	\\
850	0.450625	\\
950	0.482187	\\
1050	0.522188	\\
1150	0.551875	\\
1250	0.584375	\\
1350	0.625	\\
1450	0.631875	\\
1550	0.659375	\\
1650	0.672188	\\
1750	0.706875	\\
1850	0.710313	\\
1950	0.679375	\\
2050	0.708438	\\
2150	0.7275	\\
2250	0.718437	\\
2350	0.73625	\\
2450	0.718437	\\
2550	0.716562	\\
2650	0.726562	\\
2750	0.739688	\\
2810	0.734688	\\
};

\addplot [color=mycolor4, line width=1.5pt, dashed, mark=diamond, mark options={solid, mycolor4}]
  table[row sep=crcr]{%
0	0.098125	\\
50	0.110312	\\
150	0.174375	\\
250	0.22	\\
350	0.240937	\\
450	0.288125	\\
550	0.325625	\\
650	0.372812	\\
750	0.40625	\\
850	0.46	\\
950	0.496875	\\
1050	0.532188	\\
1150	0.553438	\\
1250	0.575	\\
1350	0.59125	\\
1450	0.620313	\\
1550	0.65125	\\
1650	0.684688	\\
1750	0.683438	\\
1850	0.695625	\\
1950	0.684063	\\
2050	0.678125	\\
2150	0.68125	\\
2250	0.679375	\\
2350	0.674687	\\
2450	0.69	\\
2550	0.690937	\\
2650	0.700313	\\
2750	0.705	\\
2810	0.696562	\\
};

\addplot [color=mycolor6, line width=1.5pt, dashed, mark=triangle, mark options={solid, mycolor6}]
  table[row sep=crcr]{%
0	0.09625	\\
50	0.1	\\
150	0.0934375	\\
250	0.093125	\\
350	0.104688	\\
450	0.09625	\\
550	0.104063	\\
650	0.110625	\\
750	0.109063	\\
850	0.119688	\\
950	0.115312	\\
1050	0.112187	\\
1150	0.120938	\\
1250	0.11875	\\
1350	0.119688	\\
1450	0.123438	\\
1550	0.119375	\\
1650	0.120313	\\
1750	0.12375	\\
1850	0.129375	\\
1950	0.124688	\\
2050	0.128438	\\
2150	0.126875	\\
2250	0.125312	\\
2350	0.129375	\\
2450	0.1175	\\
2550	0.134687	\\
2650	0.1275	\\
2750	0.124375	\\
2810	0.133125	\\
};

\end{axis}

\begin{axis}[%
width=7cm,
height=3cm,
at={(0in,0in)},
scale only axis,
xmin=0,
xmax=1,
ymin=0,
ymax=1,
axis line style={draw=none},
ticks=none,
axis x line*=bottom,
axis y line*=left
]

\end{axis}
\end{tikzpicture}%
        \captionsetup{justification=centering}
		\caption{$T_{\max} = 0.5.$} \label{fig:MNIST_cr_b}
	\end{subfigure}
	
	\begin{subfigure}[b]{0.5\textwidth}
		\begin{tikzpicture}
\definecolor{mycolor1}{rgb}{0.63529,0.07843,0.18431}%
\definecolor{mycolor2}{rgb}{0.00000,0.44706,0.74118}%
\definecolor{mycolor3}{rgb}{0.00000,0.49804,0.00000}%
\definecolor{mycolor4}{rgb}{0.87059,0.49020,0.00000}%
\definecolor{mycolor5}{rgb}{0.00000,0.44700,0.74100}%
\definecolor{mycolor6}{rgb}{0.74902,0.00000,0.74902}%

\begin{axis}[%
width=7cm,
height=3cm,
scale only axis,
xmin=0,
xmax=2810,
xlabel style={font=\color{white!15!black}},
ymin=0,
ymax=1,
ylabel style={font=\color{white!15!black}},
ylabel={Accuracy},
axis background/.style={fill=white},
xmajorgrids,
ymajorgrids,
legend style={legend cell align=left, align=left, draw=white!15!black, nodes={scale=0.85, transform shape}, at={(0.01,0.65)}, anchor=west, fill opacity=0.8}
]

\addplot [color=mycolor1, line width=1.5pt, mark=asterisk, mark options={solid, mycolor1}]
  table[row sep=crcr]{%
0	0.0945312	\\
50	0.321562	\\
150	0.652969	\\
250	0.749062	\\
350	0.80875	\\
450	0.828125	\\
550	0.847187	\\
650	0.857812	\\
750	0.862031	\\
850	0.88125	\\
950	0.875938	\\
1050	0.887969	\\
1150	0.888281	\\
1250	0.895	\\
1350	0.895312	\\
1450	0.889531	\\
1550	0.896094	\\
1650	0.897344	\\
1750	0.899219	\\
1850	0.903438	\\
1950	0.907813	\\
2050	0.90375	\\
2150	0.905625	\\
2250	0.906719	\\
2350	0.906719	\\
2450	0.91125	\\
2550	0.90625	\\
2650	0.911719	\\
2750	0.912188	\\
2810	0.917188	\\
};

\addplot [color=mycolor2, line width=1.5pt, dashed, mark=o, mark options={solid, mycolor2}]
  table[row sep=crcr]{%
0	0.11	\\
50	0.12375	\\
150	0.180938	\\
250	0.252188	\\
350	0.318437	\\
450	0.396875	\\
550	0.453437	\\
650	0.495	\\
750	0.524062	\\
850	0.560937	\\
950	0.595938	\\
1050	0.610313	\\
1150	0.62625	\\
1250	0.6475	\\
1350	0.675937	\\
1450	0.69	\\
1550	0.70125	\\
1650	0.707812	\\
1750	0.731563	\\
1850	0.731875	\\
1950	0.744687	\\
2050	0.753125	\\
2150	0.773438	\\
2250	0.77375	\\
2350	0.780937	\\
2450	0.796562	\\
2550	0.781875	\\
2650	0.808125	\\
2750	0.79	\\
2810	0.802188	\\
};

\addplot [color=mycolor3, line width=1.5pt, dashed, mark=pentagon, mark options={solid, mycolor3}]
  table[row sep=crcr]{%
0	0.0953125	\\
50	0.145937	\\
150	0.253125	\\
250	0.356875	\\
350	0.425938	\\
450	0.50375	\\
550	0.568125	\\
650	0.63125	\\
750	0.687813	\\
850	0.721875	\\
950	0.735313	\\
1050	0.742188	\\
1150	0.778125	\\
1250	0.772813	\\
1350	0.771563	\\
1450	0.77625	\\
1550	0.79125	\\
1650	0.775625	\\
1750	0.777813	\\
1850	0.802188	\\
1950	0.82375	\\
2050	0.816875	\\
2150	0.82875	\\
2250	0.84375	\\
2350	0.832187	\\
2450	0.839063	\\
2550	0.843437	\\
2650	0.83875	\\
2750	0.84	\\
2810	0.844688	\\
};

\addplot [color=mycolor4, line width=1.5pt, dashed, mark=diamond, mark options={solid, mycolor4}]
  table[row sep=crcr]{%
0	0.1075	\\
50	0.158125	\\
150	0.261562	\\
250	0.35375	\\
350	0.422187	\\
450	0.510938	\\
550	0.58625	\\
650	0.654687	\\
750	0.721562	\\
850	0.726562	\\
950	0.737187	\\
1050	0.76875	\\
1150	0.749062	\\
1250	0.758125	\\
1350	0.760938	\\
1450	0.7675	\\
1550	0.778125	\\
1650	0.784687	\\
1750	0.800625	\\
1850	0.807187	\\
1950	0.8175	\\
2050	0.813125	\\
2150	0.84	\\
2250	0.819063	\\
2350	0.838125	\\
2450	0.847187	\\
2550	0.825937	\\
2650	0.849688	\\
2750	0.843437	\\
2810	0.86	\\
};

\addplot [color=mycolor6, line width=1.5pt, dashed, mark=triangle, mark options={solid, mycolor6}]
  table[row sep=crcr]{%
0	0.1	\\
50	0.1075	\\
150	0.103438	\\
250	0.117813	\\
350	0.105625	\\
450	0.113125	\\
550	0.120938	\\
650	0.128438	\\
750	0.122813	\\
850	0.132812	\\
950	0.140937	\\
1050	0.148438	\\
1150	0.154688	\\
1250	0.1525	\\
1350	0.155938	\\
1450	0.155625	\\
1550	0.169687	\\
1650	0.163438	\\
1750	0.171563	\\
1850	0.17875	\\
1950	0.164687	\\
2050	0.187812	\\
2150	0.1925	\\
2250	0.195625	\\
2350	0.210938	\\
2450	0.204063	\\
2550	0.207813	\\
2650	0.218438	\\
2750	0.21375	\\
2810	0.234063	\\
};

\end{axis}

\begin{axis}[%
width=7cm,
height=3cm,
at={(0in,0in)},
scale only axis,
xmin=0,
xmax=1,
ymin=0,
ymax=1,
axis line style={draw=none},
ticks=none,
axis x line*=bottom,
axis y line*=left
]

\end{axis}
\end{tikzpicture}%
        \captionsetup{justification=centering}
		\caption{$T_{\max} = 1.$}\label{fig:MNIST_cr_c}
	\end{subfigure}
	
	\begin{subfigure}[b]{0.5\textwidth}
		\begin{tikzpicture}
\definecolor{mycolor1}{rgb}{0.63529,0.07843,0.18431}%
\definecolor{mycolor2}{rgb}{0.00000,0.44706,0.74118}%
\definecolor{mycolor3}{rgb}{0.00000,0.49804,0.00000}%
\definecolor{mycolor4}{rgb}{0.87059,0.49020,0.00000}%
\definecolor{mycolor5}{rgb}{0.00000,0.44700,0.74100}%
\definecolor{mycolor6}{rgb}{0.74902,0.00000,0.74902}%

\begin{axis}[%
width=7cm,
height=3cm,
scale only axis,
xmin=0,
xmax=2810,
xlabel style={font=\color{white!15!black}},
xlabel={Mini-Batch Iteration},
ymin=0,
ymax=1,
ylabel style={font=\color{white!15!black}},
ylabel={Accuracy},
axis background/.style={fill=white},
xmajorgrids,
ymajorgrids,
legend style={legend cell align=left, align=left, draw=white!15!black, nodes={scale=0.85, transform shape}, at={(0.01,0.65)}, anchor=west, fill opacity=0.8}
]

\addplot [color=mycolor1, line width=1.5pt, mark=asterisk, mark options={solid, mycolor1}]
  table[row sep=crcr]{%
0	0.0945312	\\
50	0.321562	\\
150	0.652969	\\
250	0.749062	\\
350	0.80875	\\
450	0.828125	\\
550	0.847187	\\
650	0.857812	\\
750	0.862031	\\
850	0.88125	\\
950	0.875938	\\
1050	0.887969	\\
1150	0.888281	\\
1250	0.895	\\
1350	0.895312	\\
1450	0.889531	\\
1550	0.896094	\\
1650	0.897344	\\
1750	0.899219	\\
1850	0.903438	\\
1950	0.907813	\\
2050	0.90375	\\
2150	0.905625	\\
2250	0.906719	\\
2350	0.906719	\\
2450	0.91125	\\
2550	0.90625	\\
2650	0.911719	\\
2750	0.912188	\\
2810	0.917188	\\
};

\addplot [color=mycolor2, line width=1.5pt, dashed, mark=o, mark options={solid, mycolor2}]
  table[row sep=crcr]{%
0	0.110312	\\
50	0.1725	\\
150	0.359063	\\
250	0.51625	\\
350	0.606563	\\
450	0.650312	\\
550	0.69375	\\
650	0.7275	\\
750	0.755938	\\
850	0.78125	\\
950	0.799375	\\
1050	0.815625	\\
1150	0.818438	\\
1250	0.83	\\
1350	0.83375	\\
1450	0.839688	\\
1550	0.860938	\\
1650	0.844375	\\
1750	0.8525	\\
1850	0.85875	\\
1950	0.868437	\\
2050	0.868437	\\
2150	0.8725	\\
2250	0.870938	\\
2350	0.8675	\\
2450	0.869687	\\
2550	0.877188	\\
2650	0.88125	\\
2750	0.875	\\
2810	0.891563	\\
};

\addplot [color=mycolor3, line width=1.5pt, dashed, mark=pentagon, mark options={solid, mycolor3}]
  table[row sep=crcr]{%
0	0.0928125	\\
50	0.197813	\\
150	0.344375	\\
250	0.483125	\\
350	0.600938	\\
450	0.69	\\
550	0.757812	\\
650	0.8	\\
750	0.794375	\\
850	0.800312	\\
950	0.811875	\\
1050	0.822812	\\
1150	0.82625	\\
1250	0.823438	\\
1350	0.831562	\\
1450	0.8525	\\
1550	0.860625	\\
1650	0.847812	\\
1750	0.860625	\\
1850	0.853437	\\
1950	0.85875	\\
2050	0.8575	\\
2150	0.858437	\\
2250	0.880313	\\
2350	0.872188	\\
2450	0.871563	\\
2550	0.881875	\\
2650	0.867812	\\
2750	0.884687	\\
2810	0.877188	\\
};

\addplot [color=mycolor4, line width=1.5pt, dashed, mark=diamond, mark options={solid, mycolor4}]
  table[row sep=crcr]{%
0	0.106563	\\
50	0.214375	\\
150	0.360625	\\
250	0.491875	\\
350	0.61125	\\
450	0.69625	\\
550	0.767188	\\
650	0.775937	\\
750	0.804063	\\
850	0.812813	\\
950	0.822812	\\
1050	0.829375	\\
1150	0.836875	\\
1250	0.839688	\\
1350	0.849063	\\
1450	0.851562	\\
1550	0.8625	\\
1650	0.86	\\
1750	0.859062	\\
1850	0.86125	\\
1950	0.8725	\\
2050	0.865625	\\
2150	0.871875	\\
2250	0.88	\\
2350	0.869062	\\
2450	0.873125	\\
2550	0.8875	\\
2650	0.8775	\\
2750	0.875938	\\
2810	0.880625	\\
};

\addplot [color=mycolor6, line width=1.5pt, dashed, mark=triangle, mark options={solid, mycolor6}]
  table[row sep=crcr]{%
0	0.100937	\\
50	0.0928125	\\
150	0.108438	\\
250	0.12125	\\
350	0.147187	\\
450	0.160625	\\
550	0.162188	\\
650	0.1875	\\
750	0.204375	\\
850	0.222188	\\
950	0.22875	\\
1050	0.270937	\\
1150	0.276875	\\
1250	0.300312	\\
1350	0.316875	\\
1450	0.331562	\\
1550	0.35	\\
1650	0.35625	\\
1750	0.359375	\\
1850	0.37875	\\
1950	0.40375	\\
2050	0.400313	\\
2150	0.423438	\\
2250	0.423125	\\
2350	0.434375	\\
2450	0.44	\\
2550	0.439063	\\
2650	0.474062	\\
2750	0.454688	\\
2810	0.473438	\\
};

\end{axis}

\begin{axis}[%
width=7cm,
height=3cm,
at={(0in,0in)},
scale only axis,
xmin=0,
xmax=1,
ymin=0,
ymax=1,
axis line style={draw=none},
ticks=none,
axis x line*=bottom,
axis y line*=left
]

\end{axis}
\end{tikzpicture}%
        \captionsetup{justification=centering}
		\caption{$T_{\max} = 2.$}\label{fig:MNIST_cr_d}
	\end{subfigure}
	\caption{MNIST classification accuracy using  column-times-row matrix multiplication for the DNN proposed in Sec. \ref{subsec:DNN setting}.} \label{fig:MNIST_cr}

\end{figure}
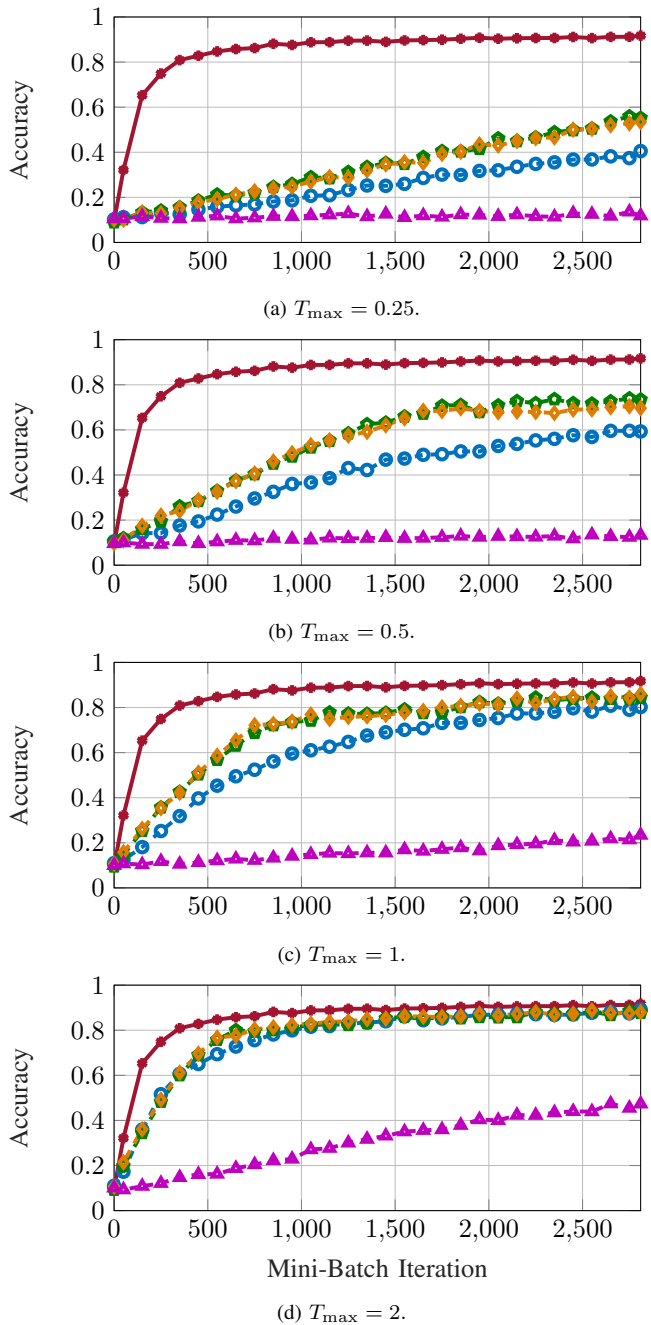


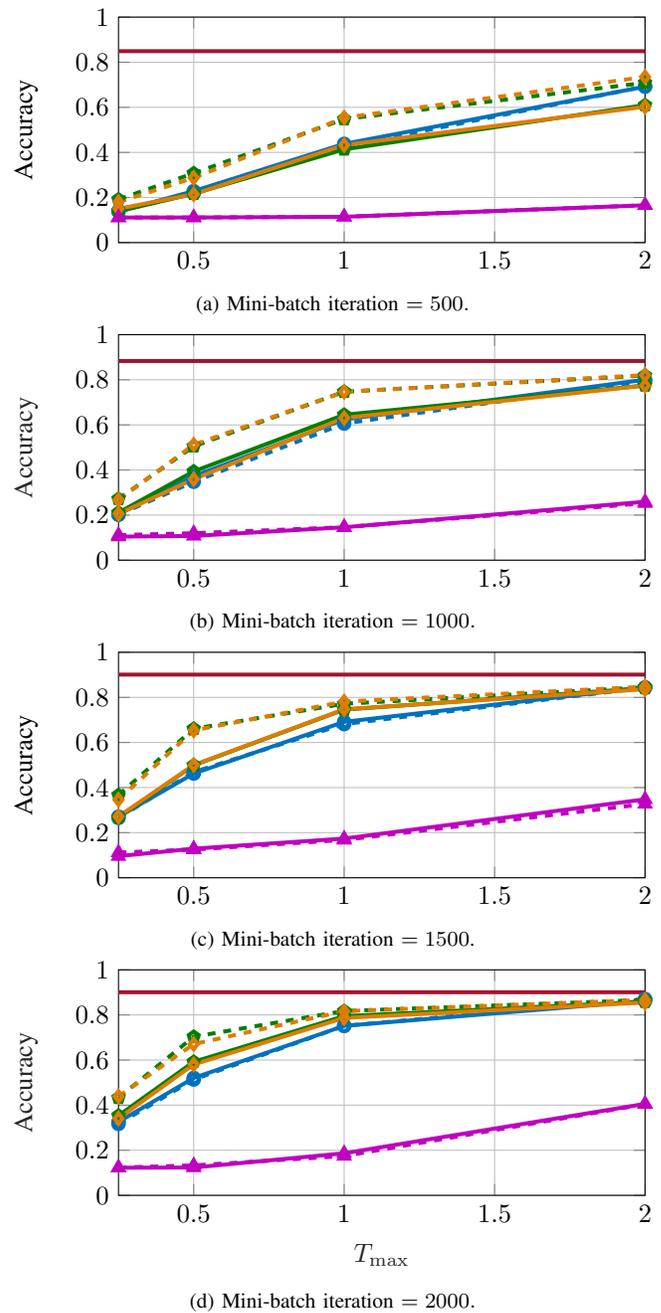
\begin{figure}
	\begin{subfigure}[b]{0.5\textwidth}
	    \centering
        \begin{tikzpicture}
\definecolor{mycolor1}{rgb}{0.63529,0.07843,0.18431}%
\definecolor{mycolor2}{rgb}{0.00000,0.44706,0.74118}%
\definecolor{mycolor3}{rgb}{0.00000,0.49804,0.00000}%
\definecolor{mycolor4}{rgb}{0.87059,0.49020,0.00000}%
\definecolor{mycolor5}{rgb}{0.00000,0.44700,0.74100}%
\definecolor{mycolor6}{rgb}{0.74902,0.00000,0.74902}%

\begin{axis}[%
width=7cm,
height=3cm,
scale only axis,
xmin=0.25,
xmax=2,
xlabel style={font=\color{white!15!black}},
ymin=0,
ymax=1,
ylabel={Accuracy},
axis background/.style={fill=white},
xmajorgrids,
ymajorgrids,
legend style={legend cell align=left, align=left, draw=white!15!black, nodes={scale=0.85, transform shape}, at={(0.01,0.65)}, anchor=west, fill opacity=0.8}
]

\addplot [color=mycolor1, line width=1.5pt, mark options={solid, mycolor1}]
  table[row sep=crcr]{%
0.25	0.848906	\\
0.5	0.848906	\\
1	0.848906	\\
2	0.848906	\\
};

\addplot [color=mycolor2, line width=1.5pt, mark=o, mark options={solid, mycolor2}]
  table[row sep=crcr]{%
0.25	0.137813	\\
0.5	0.228437	\\
1	0.437812	\\
2	0.691562	\\
};

\addplot [color=mycolor2, line width=1.5pt, dashed, mark=o, mark options={solid, mycolor2}]
  table[row sep=crcr]{%
0.25	0.14	\\
0.5	0.219375	\\
1	0.425	\\
2	0.6925	\\
};

\addplot [color=mycolor3, line width=1.5pt, mark=pentagon, mark options={solid, mycolor3}]
  table[row sep=crcr]{%
0.25	0.138437	\\
0.5	0.216875	\\
1	0.413438	\\
2	0.609688	\\
};

\addplot [color=mycolor3, line width=1.5pt, dashed, mark=pentagon, mark options={solid, mycolor3}]
  table[row sep=crcr]{%
0.25	0.191562	\\
0.5	0.3075	\\
1	0.547813	\\
2	0.70875	\\
};

\addplot [color=mycolor4, line width=1.5pt, mark=diamond, mark options={solid, mycolor4}]
  table[row sep=crcr]{%
0.25	0.15	\\
0.5	0.214688	\\
1	0.432188	\\
2	0.6025	\\
};

\addplot [color=mycolor4, line width=1.5pt, dashed, mark=diamond, mark options={solid, mycolor4}]
  table[row sep=crcr]{%
0.25	0.18	\\
0.5	0.286875	\\
1	0.554063	\\
2	0.734688	\\
};

\addplot [color=mycolor6, line width=1.5pt, mark=triangle, mark options={solid, mycolor6}]
  table[row sep=crcr]{%
0.25	0.11125	\\
0.5	0.11125	\\
1	0.114375	\\
2	0.165625	\\
};

\addplot [color=mycolor6, line width=1.5pt, dashed, mark=triangle, mark options={solid, mycolor6}]
  table[row sep=crcr]{%
0.25	0.11125	\\
0.5	0.11125	\\
1	0.114375	\\
2	0.165625	\\
};

\end{axis}

\begin{axis}[%
width=7cm,
height=3cm,
at={(0in,0in)},
scale only axis,
xmin=0,
xmax=1,
ymin=0,
ymax=1,
axis line style={draw=none},
ticks=none,
axis x line*=bottom,
axis y line*=left
]

\end{axis}
\end{tikzpicture}%
        \captionsetup{justification=centering}
		\caption{ Mini-batch iteration $ = 500.$} 	\label{fig:rxc_vs_cxr_a}
	\end{subfigure}
	
	\begin{subfigure}[b]{0.5\textwidth}
	    \centering
        \begin{tikzpicture}
\definecolor{mycolor1}{rgb}{0.63529,0.07843,0.18431}%
\definecolor{mycolor2}{rgb}{0.00000,0.44706,0.74118}%
\definecolor{mycolor3}{rgb}{0.00000,0.49804,0.00000}%
\definecolor{mycolor4}{rgb}{0.87059,0.49020,0.00000}%
\definecolor{mycolor5}{rgb}{0.00000,0.44700,0.74100}%
\definecolor{mycolor6}{rgb}{0.74902,0.00000,0.74902}%

\begin{axis}[%
width=7cm,
height=3cm,
scale only axis,
xmin=0.25,
xmax=2,
xlabel style={font=\color{white!15!black}},
ymin=0,
ymax=1,
ylabel style={font=\color{white!15!black}},
ylabel={Accuracy},
axis background/.style={fill=white},
xmajorgrids,
ymajorgrids,
legend style={legend cell align=left, align=left, draw=white!15!black, nodes={scale=0.85, transform shape}, at={(0.01,0.65)}, anchor=west, fill opacity=0.8}
]

\addplot [color=mycolor1, line width=1.5pt, mark options={solid, mycolor1}]
  table[row sep=crcr]{%
0.25	0.883281	\\
0.5	0.883281	\\
1	0.883281	\\
2	0.883281	\\
};

\addplot [color=mycolor2, line width=1.5pt, mark=o, mark options={solid, mycolor2}]
  table[row sep=crcr]{%
0.25	0.203437	\\
0.5	0.3725	\\
1	0.624062	\\
2	0.800312	\\
};

\addplot [color=mycolor2, line width=1.5pt, dashed, mark=o, mark options={solid, mycolor2}]
  table[row sep=crcr]{%
0.25	0.20125	\\
0.5	0.347813	\\
1	0.605938	\\
2	0.795937	\\
};

\addplot [color=mycolor3, line width=1.5pt, mark=pentagon, mark options={solid, mycolor3}]
  table[row sep=crcr]{%
0.25	0.210938	\\
0.5	0.39375	\\
1	0.645	\\
2	0.775	\\
};

\addplot [color=mycolor3, line width=1.5pt, dashed, mark=pentagon, mark options={solid, mycolor3}]
  table[row sep=crcr]{%
0.25	0.271875	\\
0.5	0.505	\\
1	0.747188	\\
2	0.81875	\\
};

\addplot [color=mycolor4, line width=1.5pt, mark=diamond, mark options={solid, mycolor4}]
  table[row sep=crcr]{%
0.25	0.205625	\\
0.5	0.359063	\\
1	0.630625	\\
2	0.774375	\\
};

\addplot [color=mycolor4, line width=1.5pt, dashed, mark=diamond, mark options={solid, mycolor4}]
  table[row sep=crcr]{%
0.25	0.2675	\\
0.5	0.512188	\\
1	0.746563	\\
2	0.820937	\\
};

\addplot [color=mycolor6, line width=1.5pt, mark=triangle, mark options={solid, mycolor6}]
  table[row sep=crcr]{%
0.25	0.104375	\\
0.5	0.1075	\\
1	0.145937	\\
2	0.259062	\\
};

\addplot [color=mycolor6, line width=1.5pt, dashed, mark=triangle, mark options={solid, mycolor6}]
  table[row sep=crcr]{%
0.25	0.11125	\\
0.5	0.119688	\\
1	0.14625	\\
2	0.2525	\\
};

\end{axis}

\begin{axis}[%
width=7cm,
height=3cm,
at={(0in,0in)},
scale only axis,
xmin=0,
xmax=1,
ymin=0,
ymax=1,
axis line style={draw=none},
ticks=none,
axis x line*=bottom,
axis y line*=left
]

\end{axis}
\end{tikzpicture}%
        \captionsetup{justification=centering}
		\caption{ Mini-batch iteration $ = 1000.$} 	\label{fig:rxc_vs_cxr_b}
	\end{subfigure}

	\begin{subfigure}[b]{0.5\textwidth}
	    \centering
		\begin{tikzpicture}
\definecolor{mycolor1}{rgb}{0.63529,0.07843,0.18431}%
\definecolor{mycolor2}{rgb}{0.00000,0.44706,0.74118}%
\definecolor{mycolor3}{rgb}{0.00000,0.49804,0.00000}%
\definecolor{mycolor4}{rgb}{0.87059,0.49020,0.00000}%
\definecolor{mycolor5}{rgb}{0.00000,0.44700,0.74100}%
\definecolor{mycolor6}{rgb}{0.74902,0.00000,0.74902}%

\begin{axis}[%
width=7cm,
height=3cm,
scale only axis,
xmin=0.25,
xmax=2,
xlabel style={font=\color{white!15!black}},
ymin=0,
ymax=1,
ylabel={Accuracy},
axis background/.style={fill=white},
xmajorgrids,
ymajorgrids,
legend style={legend cell align=left, align=left, draw=white!15!black, nodes={scale=0.85, transform shape}, at={(0.01,0.65)}, anchor=west, fill opacity=0.8}
]

\addplot [color=mycolor1, line width=1.5pt, mark options={solid, mycolor1}]
  table[row sep=crcr]{%
0.25	0.900937	\\
0.5	0.900937	\\
1	0.900937	\\
2	0.900937	\\
};

\addplot [color=mycolor2, line width=1.5pt, mark=o, mark options={solid, mycolor2}]
  table[row sep=crcr]{%
0.25	0.270937	\\
0.5	0.462187	\\
1	0.691562	\\
2	0.844375	\\
};

\addplot [color=mycolor2, line width=1.5pt, dashed, mark=o, mark options={solid, mycolor2}]
  table[row sep=crcr]{%
0.25	0.265625	\\
0.5	0.46875	\\
1	0.681562	\\
2	0.843437	\\
};

\addplot [color=mycolor3, line width=1.5pt, mark=pentagon, mark options={solid, mycolor3}]
  table[row sep=crcr]{%
0.25	0.2675	\\
0.5	0.499063	\\
1	0.745625	\\
2	0.84125	\\
};

\addplot [color=mycolor3, line width=1.5pt, dashed, mark=pentagon, mark options={solid, mycolor3}]
  table[row sep=crcr]{%
0.25	0.366563	\\
0.5	0.660937	\\
1	0.771563	\\
2	0.845	\\
};

\addplot [color=mycolor4, line width=1.5pt, mark=diamond, mark options={solid, mycolor4}]
  table[row sep=crcr]{%
0.25	0.271562	\\
0.5	0.5	\\
1	0.747188	\\
2	0.83625	\\
};

\addplot [color=mycolor4, line width=1.5pt, dashed, mark=diamond, mark options={solid, mycolor4}]
  table[row sep=crcr]{%
0.25	0.34625	\\
0.5	0.65375	\\
1	0.78125	\\
2	0.847187	\\
};

\addplot [color=mycolor6, line width=1.5pt, mark=triangle, mark options={solid, mycolor6}]
  table[row sep=crcr]{%
0.25	0.095625	\\
0.5	0.129062	\\
1	0.17375	\\
2	0.347813	\\
};

\addplot [color=mycolor6, line width=1.5pt, dashed, mark=triangle, mark options={solid, mycolor6}]
  table[row sep=crcr]{%
0.25	0.111875	\\
0.5	0.125	\\
1	0.167813	\\
2	0.3275	\\
};

\end{axis}

\begin{axis}[%
width=7cm,
height=3cm,
at={(0in,0in)},
scale only axis,
xmin=0,
xmax=1,
ymin=0,
ymax=1,
axis line style={draw=none},
ticks=none,
axis x line*=bottom,
axis y line*=left
]

\end{axis}
\end{tikzpicture}%
        \captionsetup{justification=centering}
		\caption{ Mini-batch iteration $ =  1500$.} 	\label{fig:rxc_vs_cxr_c} 
	\end{subfigure}

	\begin{subfigure}[b]{0.5\textwidth}
	    \centering
		\begin{tikzpicture}
\definecolor{mycolor1}{rgb}{0.63529,0.07843,0.18431}%
\definecolor{mycolor2}{rgb}{0.00000,0.44706,0.74118}%
\definecolor{mycolor3}{rgb}{0.00000,0.49804,0.00000}%
\definecolor{mycolor4}{rgb}{0.87059,0.49020,0.00000}%
\definecolor{mycolor5}{rgb}{0.00000,0.44700,0.74100}%
\definecolor{mycolor6}{rgb}{0.74902,0.00000,0.74902}%

\begin{axis}[%
width=7cm,
height=3cm,
scale only axis,
xmin=0.25,
xmax=2,
xlabel style={font=\color{white!15!black}},
xlabel={$T_{\max}$},
ymin=0,
ymax=1,
ylabel style={font=\color{white!15!black}},
ylabel={Accuracy},
axis background/.style={fill=white},
xmajorgrids,
ymajorgrids,
legend style={legend cell align=left, align=left, draw=white!15!black, nodes={scale=0.85, transform shape}, at={(0.65,0.35)}, anchor=west, fill opacity=0.8}
]

\addplot [color=mycolor1, line width=1.5pt, mark options={solid, mycolor1}]
  table[row sep=crcr]{%
0.25	0.900781	\\
0.5	0.900781	\\
1	0.900781	\\
2	0.900781	\\
};

\addplot [color=mycolor2, line width=1.5pt, mark=o, mark options={solid, mycolor2}]
  table[row sep=crcr]{%
0.25	0.325937	\\
0.5	0.520625	\\
1	0.751563	\\
2	0.862187	\\
};

\addplot [color=mycolor2, line width=1.5pt, dashed, mark=o, mark options={solid, mycolor2}]
  table[row sep=crcr]{%
0.25	0.316875	\\
0.5	0.51375	\\
1	0.751563	\\
2	0.870938	\\
};

\addplot [color=mycolor3, line width=1.5pt, mark=pentagon, mark options={solid, mycolor3}]
  table[row sep=crcr]{%
0.25	0.35375	\\
0.5	0.592187	\\
1	0.794687	\\
2	0.858437	\\
};

\addplot [color=mycolor3, line width=1.5pt, dashed, mark=pentagon, mark options={solid, mycolor3}]
  table[row sep=crcr]{%
0.25	0.432812	\\
0.5	0.703438	\\
1	0.8175	\\
2	0.86625	\\
};

\addplot [color=mycolor4, line width=1.5pt, mark=diamond, mark options={solid, mycolor4}]
  table[row sep=crcr]{%
0.25	0.340938	\\
0.5	0.580625	\\
1	0.786875	\\
2	0.855	\\
};

\addplot [color=mycolor4, line width=1.5pt, dashed, mark=diamond, mark options={solid, mycolor4}]
  table[row sep=crcr]{%
0.25	0.43875	\\
0.5	0.670937	\\
1	0.816562	\\
2	0.86625	\\
};

\addplot [color=mycolor6, line width=1.5pt, mark=triangle, mark options={solid, mycolor6}]
  table[row sep=crcr]{%
0.25	0.1225	\\
0.5	0.123438	\\
1	0.186562	\\
2	0.405625	\\
};

\addplot [color=mycolor6, line width=1.5pt, dashed, mark=triangle, mark options={solid, mycolor6}]
  table[row sep=crcr]{%
0.25	0.123125	\\
0.5	0.132812	\\
1	0.174375	\\
2	0.404375	\\
};

\end{axis}

\begin{axis}[%
width=7cm,
height=3cm,
at={(0in,0in)},
scale only axis,
xmin=0,
xmax=1,
ymin=0,
ymax=1,
axis line style={draw=none},
ticks=none,
axis x line*=bottom,
axis y line*=left
]

\end{axis}
\end{tikzpicture}%
        \captionsetup{justification=centering}
		\caption{ Mini-batch iteration $ = 2000.$} 	\label{fig:rxc_vs_cxr_d}
		\vspace{0.4cm}
	\end{subfigure}
	\caption{MNIST classification accuracy at different $T_{\max}$ for the DNN proposed in Sec. \ref{subsec:DNN setting}.}	\label{fig:rxc_vs_cxr}
\end{figure}

\begin{figure}
\hspace{-2.2cm}
	    \centering
	    \begin{tikzpicture}
\definecolor{mycolor1}{rgb}{0.63529,0.07843,0.18431}%
\definecolor{mycolor2}{rgb}{0.00000,0.44706,0.74118}%
\definecolor{mycolor3}{rgb}{0.00000,0.49804,0.00000}%
\definecolor{mycolor4}{rgb}{0.87059,0.49020,0.00000}%
\definecolor{mycolor5}{rgb}{0.00000,0.44700,0.74100}%
\definecolor{mycolor6}{rgb}{0.74902,0.00000,0.74902}%

\begin{axis}[%
hide axis,
xmin=0,
xmax=0.4,
ymin=0,
ymax=0.4,
legend style={legend cell align=left, align=left, draw=white!15!black}
]

\addlegendimage{color=mycolor1, line width=1.5pt, mark=asterisk, mark options={solid, mycolor1}}
\addlegendentry{No stragglers}

\addlegendimage{color=mycolor2, line width=1.5pt, mark=o, mark options={solid, mycolor2}}
\addlegendentry{Uncoded ($\rtc$)}

\addlegendimage{color=mycolor2, line width=1.5pt, dashed, mark=o, mark options={solid, mycolor2}}
\addlegendentry{Uncoded ($\ctr$)}

\addlegendimage{color=mycolor3, line width=1.5pt, mark=pentagon, mark options={solid, mycolor3}}
\addlegendentry{NOW ($\rtc$)}

\addlegendimage{color=mycolor3, line width=1.5pt, dashed, mark=pentagon, mark options={solid, mycolor3}}
\addlegendentry{NOW ($\ctr$)}

\addlegendimage{color=mycolor4, line width=1.5pt, mark=diamond, mark options={solid, mycolor4}}
\addlegendentry{EW ($\rtc$)}

\addlegendimage{color=mycolor4, line width=1.5pt, dashed, mark=diamond, mark options={solid, mycolor4}}
\addlegendentry{EW ($\ctr$)}

\addlegendimage{color=mycolor6, line width=1.5pt, mark=triangle, mark options={solid, mycolor6}}
\addlegendentry{2-Block repetition ($\rtc$)}

\addlegendimage{color=mycolor6, line width=1.5pt, dashed, mark=triangle, mark options={solid, mycolor6}}
\addlegendentry{2-Block repetition ($\ctr$)}

\end{axis}

\end{tikzpicture}%
	\vspace{-1.5cm}
	\caption{Legend for Fig. \ref{fig:MNIST_rc} to Fig. \ref{fig:rxc_vs_cxr}.} \label{fig:legend}
\end{figure}
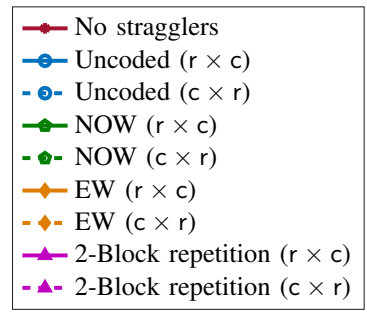


\subsection{DNN Performance with UEP Coded Matrix Multiplication}
\label{subsec: DNN performance}

%
%
\begin{table}[b]
	\footnotesize
	\centering
	\caption{A summary of the encoding parameters in Sec. \ref{subsec: DNN performance}.}
	\label{tab:model_parameters}
	\begin{tabular}{|c|c|c|}
		\hline
		Encoding Type  & $W$ & $\Omega$ \\ \hline
		Uncoded   & 9 & 9/9                  \\ \hline
		NOW/EW - UEP &  15 & 9/15                 \\ \hline
		2-Block Rep &  18 & 9/18                 \\ \hline  
	\end{tabular}
\end{table}

\begin{table}[b]
	\footnotesize
	\centering
	\caption{A summary of the number of sub-blocks belonging to each importance level for UEP codes used in Sec. \ref{subsec: DNN performance}.}
	\label{tab:model_uep_parameters}
	\begin{tabular}{|c|c|}
		\hline
		Importance Level  & $n_{C}(s)$ \\ \hline
		High   & 1 \\ \hline
		Medium & 2 \\ \hline
		Low &  6    \\ \hline  
	\end{tabular}
\end{table}

%
%
In this section, we finally investigate the accuracy of the proposed solution for the two DNN models described in Sec. \ref{subsec:DNN setting}. 
In our comparisons, we show the performance of the DNN training with
\begin{itemize}
    \item {\bf red line:}  centralized computation with no stragglers,
    \item {\bf green line:} distributed  computation with  NOW-UEP codes,
    \item {\bf yellow line:} distributed  computation with  EW-UEP codes,
    \item {\bf purple line:} distributed  computation with 2-block repetition.
\end{itemize}
In all cases above, 
\begin{itemize}
    \item {\bf continuous line:} is for the $\rtc$ paradigm, while 
    \item {\bf dotted line:} is for the $\ctr$ paradigm.
\end{itemize}
%
%
		

For a fair comparison among these scenarios, we scale the time required to complete a task as $F(\Omega t)$, where $\Omega$ is the number of matrix sub-products divided by the number of workers (\textit{see Remark \ref{rem:Comparison across models}}).
%
%
For the simulations, we have that the total number of matrix sub-products is 9, where $N = P = 3$, thus $N P = 9$ for $\rtc$ and $M = 9$ for $\ctr$ case.
Additionally,  we consider an exponential latency model $\lambda=0.5$.  We take $T_{\max} \in \{ 0.25, 0.5, 1, 2\}$. 
%
%
Other simulation parameters are given in Tables \ref{tab:model_parameters} and \ref{tab:model_uep_parameters}. The dimension of the encoding matrices $\Av$/$\Bv$ for each dense layer and gradient are shown in Table \ref{tab:backprop_dimensions}. 
The settings above are for both the $\rtc$ and the $\ctr$ cases.


Next, let us describe how the importance levels are obtained in our coding scheme. 
Similar to  \cite{yuster2005fast}, which proposes a fast matrix multiplication algorithm, the column/row indexes are permuted  so as to obtain a descending magnitude of the column/row weights. 
Note that ordering has average complexity $\Ocal(n \log n)$ in the number of columns/rows, so that the computational burden at the PS is minimal as compared to the matrix multiplication complexity. 
%
%
Once ordered in decreasing magnitudes, column/row sub-blocks are formed by dividing them into three groups of (roughly) equal size.
The sub-blocks are then encoded using the NOW/EW-UEP code as specified in Table \ref{tab:code param}.

\noindent
{\bf MNIST:}
Figs. \ref{fig:MNIST_rc} and \ref{fig:MNIST_cr} depict the results of using different multiplication strategies at different $T_{\max}$, $\rtc$ for the former and $\ctr$ for the latter with the MNIST dataset, and Fig. \ref{fig:rxc_vs_cxr} serves as a merger of these two by fixing the mini-batch iteration to compare the accuracy trade-off when using different deadlines. Ideally, we would want to increase the accuracy while decreasing the deadline, so the model can achieve convergence in the least time possible. Note that since the no straggler receives all sub-blocks back by any deadline, it is a constant and our benchmark.

We observe from Fig. \ref{fig:rxc_vs_cxr}  that for $T_{\max} < 2$, the UEP coding strategy shows an advantage over the others, with a more significant accuracy gap with $\ctr$ at $T_{\max}=0.5$ and $1$. By selecting two deadlines that are multiples of each other by a factor of $x$ in Figs. \ref{fig:rxc_vs_cxr_a} and \ref{fig:rxc_vs_cxr_b}, such as $1$ and $2$ by a factor of 2, we can observe a clear trade-off in the accuracy: if it takes $n$ iterations for the longer deadline to achieve a $y$ accuracy, it will take $x$ times $n$ iterations for the shorter deadline to achieve an accuracy close to $y$. However, further study is necessary to determine whether this holds true with other datasets and DNN architectures.

Although the uncoded scheme provides no protection against stragglers, missing a few of blocks does not cause much damage to the learning, showing once again the inherent fault-tolerance of DNNs. 
Perhaps not too surprisingly, employing block repetition coding increases the number of workers required but does not result in a better overall performance as compared to the uncoded scheme.
This result follows from the  choice of  the waiting time  distribution and it provides  an understanding on the computational complexity scaling in our simulations, as discussed in Remark \ref{rem:Comparison across models}.
%
Consider the following two scenarios: (i) one  machine completes a job, versus (ii) the same job is given to two machines and the job is completed whenever one of the two machines finishes. To compare these two scenarios in a fair manner, the expected value of the waiting time of the second scenario must be double of the first one.
For our choice of waiting time distribution, scenario (i) performs better than scenario (ii) above, so that there is no intrinsic advantage in distributing one job to multiple workers. 
%
%
%
%
%
%
%
%
From Fig. \ref{fig:rxc_vs_cxr}, we observe that $\rtc$ and $\ctr$ partitioning have roughly the same performance as uncoded and 2-block repetition.
However, in the case of UEP codes, $\ctr$ shows an advantage over $\rtc$, i.e., mini-batch 1000 and 1500 for $T_{\max} = 0.25$ to $1$.
%
Further exploration of this phenomenon is need to be understood whether this is true for this specific DNN or it is more general. 
%
%

\noindent
{\bf CIFAR-10:}
For the simulations with the CIFAR-10 dataset in Fig. \ref{fig:cifar10_big}, %
we use $\lambda=0.5$ with an exponential latency model and $T_{\max} = 1$. 
Since the model evolves slowly, the first few epochs result in gradient computations that are rather uniform in row/column norms. 
For this reason, we let the model train for the first $30$ epochs without stragglers.
After these first set of epochs, there is enough sparsification in the gradients
to justify UEP coding: accordingly Fig. \ref{fig:cifar10_big} presents simulation results between $30$ and $120$ epochs. 
%
The calculations for the convolutional layer are performed without stragglers through central computations. 
For the dense layers, we use the coding strategy summarized in Table \ref{tab:model_parameters}  and the corresponding scaling of the waiting time distribution. 
The only exception is the set of matrices in  \eqref{eq:backprop_eq2} for the last layer, for which we use the uncoded scheme since they were not sufficiently sparse.

As we can observe from Fig. \ref{fig:cifar10_big}, after around $60$ epochs, the UEP codes widen the accuracy margin and learn faster than the other two encoding strategies. This is in part due to the gradient evolution where the weights of the sub-blocks show a higher variance than they did between epochs $30$ to $60$ due to sparsification and convergence. This effect continues to amplify as the epoch progresses, promoting the need of unequal error protection.
In particular, it appears that both classes of the UEP codes employed yield an accuracy which tends to those attainable without stragglers, i.e., $1$, while the uncoded transmission performance saturates below $0.9$.
It is interesting to note that this situation is substantially different from the scenario in Fig. \ref{fig:MNIST_rc} in which the MNIST dataset DNN is trained with the same value of $T_{\max}$.
In the latter case, in fact, no substantial improvement is provided by the UEP codes over the uncoded transmission. 

We close this section by noting that all the simulation results can be reproduced using Matlab and Python codes available at: 
\url{https://github.com/HernandezEduin/UEP-Straggler-Mitigation}.

\subsection{Future Work}

In Sec. \ref{sec:Back-propagation Matrices} we apply the distributed approximate matrix computation scheme in Sec. \ref{sec:Approximate Matrix Multiplication with UEP Codes} to the evaluation of gradient through back-propagation for training a DNN. 
The performance of the proposed approach is evaluated only through numerical simulations, and a more detailed analysis of the relationship between distributed computation and distributed learning is left for future research.
Here we would like to point out some aspects of this problem which we believe deserve further investigation. 

\noindent
{\bf Error feedback:} in the current implementation for DNN training, the back-propagation algorithm does not account for the specific nature of the gradient noise arising from the application of scheme in Sec. \ref{sec:Approximate Matrix Multiplication with UEP Codes}.
A more efficient scheme can be implemented by accounting for the accumulation of the approximation error at each gradient evaluation, similarly to the approach in \cite{stich2018sparsified}.

\noindent
{\bf Computation accuracy/learning accuracy:} DNN training requires the successive evaluation of matrix products. Currently, it is not clear how the error in approximate computation propagates through these successive evaluations.
Additionally, it is not well-understood how the error in the gradient approximation affects the learning accuracy at different iterations.
Once these aspects have been better understood, the scheme in Sec. \ref{sec:Back-propagation Matrices} can be optimally designed at each DNN layer and at each iteration to yield the best expected accuracy by a chosen deadline. 

\noindent
{\bf Computation time:} The current modeling of the computation delay in \eqref{eq:waiting} is an accurate model for queuing delays and might not properly account for the computation time requirement more specific to DNN training. By more appropriately choosing the response time in \eqref{eq:waiting}, one could obtain a more practical design for the scheme in Sec. \ref{sec:Approximate Matrix Multiplication with UEP Codes}.

\noindent
{\bf Computation load:} Inherent in the design of the scheme in Sec. \ref{sec:Approximate Matrix Multiplication with UEP Codes} is the assumption that the computational cost of matrix summation and scaling is negligible when compared to that of matrix multiplication. 
Accordingly, the PS can perform any type of coding which only employs linear combination of matrices. 
In actuality, the computational load might be better expressed in terms of memory allocation or overall number of flops. In this scenario, the design of the scheme in Sec.  \ref{sec:Approximate Matrix Multiplication with UEP Codes} should be modified to account for such limitations. 

\noindent
%
{\bf Extension to Federated Learning (FL): }
The current coding scheme is implemented in a non-standard distributed learning fashion, where computations are distributed from the PS to the workers on a per layer basis of the DNN. One can imagine the extension of the scheme to the FL setting in which the computational nodes (workers) have access to the training dataset through the cloud or have them locally available. In such a scenario, the PS would have its complexity reduced even further with the use of UEP codes. 

\noindent
{\bf Optimization of UEP codes: }
In our work, we did not perform any optimization of UEP codes, instead we chose the window selection distributions for NOW-UEP and EW-UEP schemes arbitrarily. 
An interesting direction for future research is the UEP code optimization for matrix product approximation. To do this, the window selection probabilities can be optimized to minimize the loss in the matrix approximation for the given setup.
For instance, in the distributed DNN model, one should optimize these probabilities to minimize the error according to the current sub-weights. 





\section{Conclusions} \label{sec: conclusion}
In this paper we have studied distributed approximate matrix multiplication using UEP codes with the objective of mitigating the straggler phenomenon.
The proposed approach has a wide range of applications as it allows one to  speed up  large-scale operations which are common in machine learning and data mining algorithms.
%
%
%
We use UEP codes to provide better protection for the sub-operations which have higher effects on the resulting matrix product by better protecting the sub-operations with larger norms.
We validate the effectiveness of the proposed approach through analytical assessments based on simplified models for sparse matrices, and compare our results with those obtained with MDS codes via simulations. 
Furthermore, the proposed strategy is applied to the back-propagation steps of a DNN for the MNIST digit classification task and CIFAR-10 image classification task. 
Our results clearly show that, in the presence of stragglers, we can have a performance close to the centralized training earlier by striking a balance between the precision of the updates and the response time of the edge devices.    %

\bibliographystyle{IEEEtran}
\bibliography{bibs_matrix_mult}
\end{document}